\documentclass[12pt]{article}
\usepackage{amsmath}
\usepackage{graphicx}
\usepackage{enumerate}
\usepackage{url} % not crucial - just used below for the URL 
\usepackage{ulem}

\newcommand{\blind}{1}

\usepackage{authblk}
\usepackage[utf8]{inputenc}
\usepackage{longtable}
\usepackage{booktabs}
\usepackage{multirow}
\usepackage{bm}
\usepackage{amsmath, amsfonts, amsthm, amssymb, amsbsy, graphicx, dsfont, mathtools, enumerate, enumitem}
\usepackage{graphicx}
\usepackage[english]{babel}
\usepackage{graphicx}
\usepackage{subfigure, caption}
\usepackage{natbib}
\usepackage{multibib}
\setcitestyle{authoryear,open={(},close={)}}
\newcites{response}{References}
\usepackage[CJKbookmarks=true,
            bookmarksnumbered=true,
			bookmarksopen=true,
			colorlinks=true,
			citecolor=blue, %red
			linkcolor=blue,
			anchorcolor=red,
			urlcolor=blue]{hyperref}
\usepackage[usenames]{color}

\usepackage[letterpaper, left=1.2truein, right=1.2truein, top = 1.2truein, bottom = 1.2truein]{geometry}
\usepackage{prettyref,soul}
\usepackage[font=small,labelfont=it]{caption}
\RequirePackage[framemethod=TikZ]{mdframed}

\newtheorem{theorem}{Theorem}

\newtheorem{corollary}{Corollary}
\newtheorem{lemma}{Lemma}
\newtheorem{proposition}{Proposition}

\usepackage{apptools}
\AtAppendix{\counterwithin{lemma}{section}}
\AtAppendix{\counterwithin{theorem}{section}}
\AtAppendix{\counterwithin{corollary}{section}}
\AtAppendix{\counterwithin{proposition}{section}}

\theoremstyle{definition}

\theoremstyle{remark}
\newtheorem*{remark}{Remark}

\newcommand{\argmin}{\mathop{\rm argmin}}
\newcommand{\argmax}{\mathop{\rm argmax}}

\newcommand{\Prob}{\mathbb{P}}
\newcommand{\im}{\mathrm{im}}
\newcommand{\dfdr}{\ensuremath{\mathrm{FDR}_{\mathrm{dir}}}}
\newcommand{\dfdp}{\ensuremath{\mathrm{FDP}_{\mathrm{dir}}}}
\newcommand{\Expect}{\mathbb{E}}

\newcommand{\rank}{\mathop{\sf rank}}

\newcommand{\diag}{\mathop{\text{diag}}}

\newcommand{\D}{\mathcal{D}}

\newcommand{\Nm}{\mathcal{N}}
\newcommand{\var}{\mathrm{var}}

\newcommand{\fnu}{\ensuremath{\eta(\nu)}}
\newcommand{\gnu}{\ensuremath{\vartheta(\nu)}}

\def\E{\mathbb{E}}
\def\R{\mathbb{R}}
\def\A{\mathcal{A}}
\def\S{S}
\def\diag{\mathrm{diag}}
\def\sign{\mathrm{sign}}
\def\vecs{\mathbf{s}}
\def\F{\mathcal{F}}
\def\G{\mathcal{G}}
\newcommand{\mdfdr}{\ensuremath{\mathrm{mFDR}_{\mathrm{dir}}}}

\newcommand{\co}[1]{\ensuremath{C_{#1}}}

\begin{document}

\def\spacingset#1{\renewcommand{\baselinestretch}%
{#1}\small\normalsize} \spacingset{1}

%%%%%%%%%%%%%%%%%%%%%%%%%%%%%%%%%%%%%%%%%%%%%%%%%%%%%%%%%%%%%%%%%%%%%%%%%%%%%%

\if1\blind
{
  \title{
  Split Knockoffs for Multiple Comparisons: Controlling the Directional False Discovery Rate
  }
  \author{Yang Cao$^1$, Xinwei Sun$^2$\thanks{\url{sunxinwei@fudan.edu.cn}} and Yuan Yao$^1$\thanks{\url{yuany@ust.hk}}
    ~\\
    $^1$Hong Kong University of Science and Technology \\
    $^2$Fudan University
    }
    \date{}
  \maketitle
} \fi

\if0\blind
{
  \bigskip
  \bigskip
  \bigskip
  \begin{center}
    {\LARGE 
    Split Knockoffs for Multiple Comparisons: Controlling the Directional False Discovery Rate}
\end{center}
  \medskip
} \fi

\bigskip
\begin{abstract}

Multiple comparisons in hypothesis testing often encounter structural constraints in various applications. For instance, in structural Magnetic Resonance Imaging for Alzheimer's Disease, the focus extends beyond examining atrophic brain regions to include comparisons of anatomically adjacent regions. These constraints can be modeled as linear transformations of parameters, where the sign patterns play a crucial role in estimating directional effects. This class of problems, encompassing total variations, wavelet transforms, fused LASSO, trend filtering, and more, presents an open challenge in effectively controlling the directional false discovery rate. In this paper, we propose an extended Split Knockoff method specifically designed to address the control of directional false discovery rate under linear transformations. Our proposed approach relaxes the stringent linear manifold constraint to its neighborhood, employing a variable splitting technique commonly used in optimization. This methodology yields an orthogonal design that benefits both power and directional false discovery rate control. By incorporating a sample splitting scheme, we achieve effective control of the directional false discovery rate, with a notable reduction to zero as the relaxed neighborhood expands. To demonstrate the efficacy of our method, we conduct simulation experiments and apply it to two real-world scenarios: Alzheimer's Disease analysis and human age comparisons.
\end{abstract}

\noindent%
{\it Keywords:} Multiple Comparison Hypothesis Test, Structural Sparsity, Variable Splitting, Alzheimer's Disease
\vfill

\newpage

\tableofcontents

\newpage
%\spacingset{1.9} % DON'T change the spacing!
\section{Introduction}

Modern hypothesis testing is often interested in studying whether differences exist among multiple pairwise comparisons of parameters (\emph{e.g.}, $\beta_i-\beta_j$), where $\beta_i$ ($i=1,...,p$) measures the effect of the explanatory variable $X_i$ to a response variable $Y$, in the following linear model: 
\begin{equation}
\label{eq: linear regression model}
     Y = \sum_{i=1}^p X_i \beta_i^* + \varepsilon, \ \ \varepsilon \sim \mathcal{N}(0,\sigma^2).
\end{equation}
There are extensive studies \citep{tukey1991philosophy, benjamini2002john} on multiple comparisons with hypothesis $H^{(i,j)}: \beta_i-\beta_j=0$ for \emph{each} $(i,j)$. This can be regarded as a \emph{complete} graph of pairwise comparisons $G=(V,E)$, where the vertex set $V=\{1,\ldots,p\}$ and the edge set $E$ consists of all pairs $\{(i,j): 1\leq i<j \leq p\}$.  
However, in many applications, only a subset of pairwise comparisons is of interest due to structural constraints on parameters. It results in an incomplete comparison graph $G$. 

For instance, in studies of structural Magnetic Resonance Imaging (sMRI) for Alzheimer's Disease (AD) where patients suffer from atrophy on some brain regions, $X_i$ measures the normalized gray matter volume of $i$-th brain region and $\beta_i$ thus models the influence of its atrophy with respect to the severity of the disease, measured by the Alzheimer's Disease Assessment Scale (ADAS) score $Y$. Since anatomically adjacent brain regions are expected with the same degree of degeneration in normal brains, it results in a \emph{geometric smoothing} prior, \emph{i.e.}, typically $\{\beta_i-\beta_j: \text{$i$ and $j$ are adjacent}\}$ are assumed to be zeros. One is only interested in comparisons of anatomically adjacent pair $(i,j)$ with $\beta_i \neq \beta_j$, indicating an abnormal connection of regions with different degrees of degeneration due to disease. Such constraints can be represented by the topology of graph $G=(V,E)$, where $V$ consists of $p$ regions and $E=\{(i,j):i\sim j\}$ collects anatomically adjacent region pairs. Figure~\ref{fig: ad front} illustrates such a graph, where directed black arrows mark abnormal connections, \emph{e.g.} from Lingual Gyrus to severely atrophied region Hippocampus. Details will be given in Section \ref{sec: AD-exp}.

\begin{figure}[!ht]
    \centering
    \includegraphics[width=.6\textwidth]{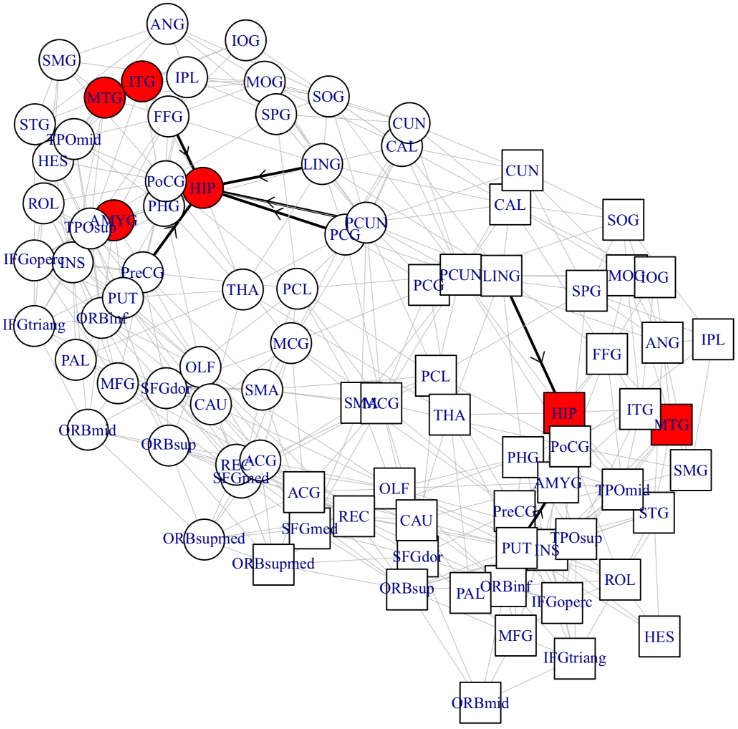}
    \caption{Selected brain regions and connections in Alzheimer's Disease. Each vertex represents a Cerebrum brain region in Automatic Anatomical Labeling (AAL) atlas \citep{tzourio2002automated}, whose abbreviations and full names are provided in Table~\ref{tab:name of region}. Here, vertices with a circle shape represent the left brain regions, while squares represent the right brain regions. A shallow undirected edge 
   represents that the associated two regions are adjacent. Using 
   our methodology (Section \ref{sec: AD-exp}), the selected atrophy regions are marked in red, while the directed bold arrows represent the selected connections pointing from a %less infected or hypertrophy regions
    normal region to significantly atrophy regions.} 
    \label{fig: ad front}
\end{figure}

In this paper, we consider the following general model which includes the examples of multiple comparisons as special cases,
\begin{align}
    \label{eq: structural sparsity model}
     Y = \sum_{i=1}^p X_i \beta^*_i + \varepsilon, \ \ \gamma^* = D\beta^*, \ \ \varepsilon \sim \mathcal{N}(0,\sigma^2),
\end{align}
where $D:\R^p \to \R^m$ is a linear transformation and the sparsity and/or sign patterns of $\gamma^* \in \R^m$ are of interest. 

Various choices of $D$ lead to different structures in applications. For example, taking $D$ to be the graph difference (gradient) operator on a graph meets the geometric multiple comparison problem above, where the graph can be complete or incomplete. Taking $D$ to be 1-d fused lasso matrix assumes sparsity between adjacent parameters ordered in a line, such as the copy numbers of genome orders in comparative gnomic hybridization (CGH) data \citep{tibshirani2005sparsity}; the gene with abnormality (different copy number with its neighbor) can help us understand the human cancer. 
Taking $D$ as the 2-D  grid gradients has been used in total variation edge detection for images \citep{ROF92}. One can also construct $D$ as general basis and frames such as wavelets \citep{donoho1995adapting}. 
In particular, taking $D$ to be the identity matrix leads to the traditional variable selection problem in regression. 

In multiple comparisons, controlling the false discoveries of mistakenly claiming $\beta_i \neq \beta_j$ is calibrated as Type-I error. One recent approach to achieve this error rate control in multiple comparisons is proposed in \cite{cao2021controlling}, namely the Split Knockoff method, as a generalization of the Knockoff method \citep{barber2015controlling} to handle the transformational sparsity \eqref{eq: structural sparsity model}. However, controlling the Type-I error in  multiple comparisons does not always hold significant meaning, since the effects of two associated statistics are generically different in most real applications, as stated in \cite{tukey1991philosophy, gelman2000type}. In case of Alzheimer's Disease, due to the intrinsic variation of patients and measurement noise of device, it leads that $\beta_i - \beta_j$ differs to zero %\approx 0$ 
for each adjacent pair. In this case, controlling the Type-I error is meaningless, since all rejected hypothesis are true discoveries by default. On the other hand, given a pair with different effects, we are more interested in figuring out the direction of the connection, \emph{i.e.}, which of the connected regions is more severely damaged in AD compared to the other. As proposed by \cite{gelman2000type}, such a more informative goal of identifying the direction of each comparison, can be calibrated by the Type-S error, with the false positive referring to claiming $\beta_i > \beta_j$ if $\beta_i < \beta_j$ or vise-versa. 

In order to enhance the reproducibility of discovering such directional effects from noisy data, our target in this paper is to recover the sign pattern of $\gamma = D\beta \in\R^{m}$ in the general model \eqref{eq: structural sparsity model}. %inspired by \cite{barber2019knockoff}, 
To be precise, we aim to control the following \emph{directional False Discovery Rate} ($\dfdr$): 
\begin{align*}
    \dfdr = \E\left[\dfdp\right] = \E\left[\frac{|\{i: i\in \widehat{S}, \widehat{\sign}_i\neq \sign(\gamma^*_i)\}|}{|\widehat{S}|\vee 1}\right],
\end{align*}
where $\widehat{S}$ is the estimated support set and $\widehat{\sign}$ is the estimated sign of the parameter.

In the special case of linear regression where $D$ is an identity matrix, \cite{barber2019knockoff} proposes to extend the knockoff method to control the $\dfdr$. However, in more general settings of $D$ including the multiple comparisons, controlling the $\dfdr$ remains an open problem, as  
the naive construction of knockoffs by ignoring the structural constraint will break the antisymmetry property \citep{barber2015controlling}, e.g. see details in Section \ref{sec: genlasso with knockoff}.

In this paper, we propose a new procedure to control the directional false discovery rate ($\dfdr$) under transformations by extending the Split Knockoff method introduced by \cite{cao2021controlling}. In this procedure, the linear constraint $\gamma = D\beta$ is relaxed to its Euclidean neighborhood. This relaxation allows us to consider the structural constraint as an inherent part of the design rather than a constraint that could potentially disrupt the symmetry between the design matrix and its knockoffs.

By combining this relaxed constraint with a sample splitting scheme, the inverse supermartingale structure presented in \cite{cao2021controlling} is generalized to control the $\dfdr$ in the extended Split Knockoff method. Furthermore, we provide theoretical demonstration that as the relaxed neighborhood enlarges, the $\dfdr$ control of the extended Split Knockoff method decreases to zero.

The paper is organized as following. In Section~\ref{sec: methodology}, the construction of Split Knockoffs is introduced, exploiting both the variable splitting scheme and the sample splitting scheme. 
In Section~\ref{sec: analysis}, an analysis on the $\dfdr$ control of Split Knockoffs is presented. 
In Section~\ref{sec: simu_exp}, simulation experiments are conducted to show that Split Knockoffs achieve the desired $\dfdr$ control, with possible improvement on power due to a better incoherence compared with standard Knockoffs. In Section~\ref{sec: applications}, experiments are conducted in two real world applications: discovering lesion brain regions and abnormal connections in Alzheimer's Disease, as well as making pairwise comparisons on human ages based on  voluntarily annotated data of human face images.

\section{Split Knockoffs}
\label{sec: methodology}

To overcome the hurdle brought by the linear constraint $\gamma=D\beta$ (e.g. see details in Section \ref{sec: genlasso with knockoff}), the Split Knockoff method starts from a relaxation of the linear constraint to its Euclidean neighbourhood, often known as the variable splitting scheme in optimization. 
The variable splitting scheme, together with 
a sample splitting scheme to introduced later, ensures Split Knockoffs the desired theoretical $\dfdr$ control to be presented later in Theorem \ref{theorem: directional fdr} in Section \ref{sec: analysis}. 

The relaxation of the linear constraint $\gamma=D\beta$ to its Euclidean neighbourhood is implemented through the following Split LASSO \citep{cao2021controlling} regularization path,
\begin{align}
    \label{eq: split lasso}
    (\beta(\lambda), \gamma(\lambda)) :=\argmin_{\beta, \gamma}\frac{1}{2n}\|y-X\beta\|_2^2+\frac{1}{2\nu}\|D\beta-\gamma\|_2^2+\lambda\|\gamma\|_1, \mbox{ for $\lambda>0$,}
\end{align}
where $\nu>0$ is a (tuning) parameter that controls the Euclidean gap between $D\beta$ and $\gamma$. In other words,  Equation \eqref{eq: split lasso} allows  
an Euclidean gap (penalized by $\frac{1}{\nu}$) between $D\beta$ and $\gamma$, instead of forcing $D\beta$ to be equal to $\gamma$.

In the following, we will outline the construction of Split Knockoffs targeting to control the $\dfdr$ under transformations. Subsequently, we will delve into the specific details of this construction.

\begin{figure}[!ht]
    \centering
    \includegraphics[width=\textwidth]{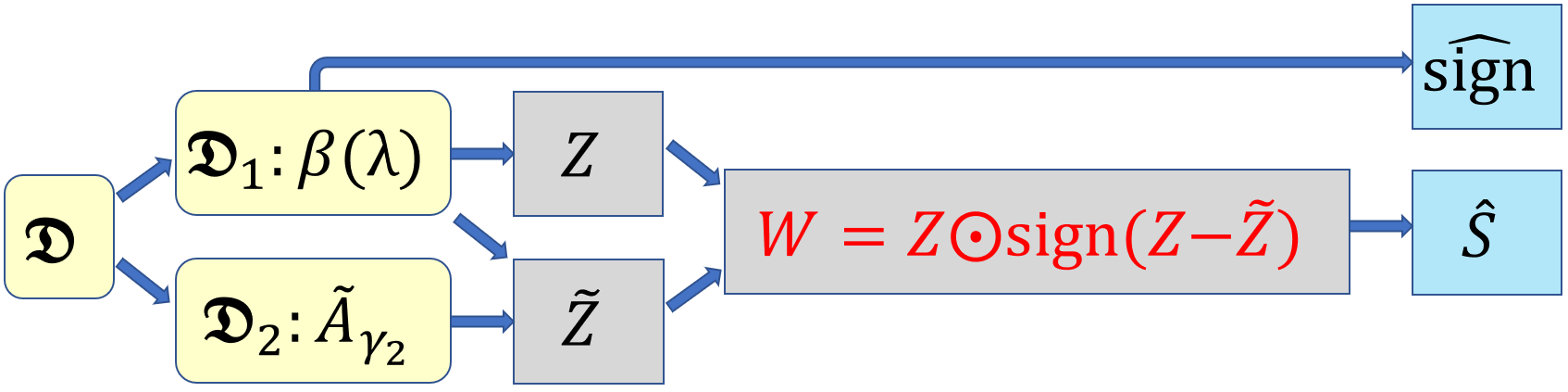}
    \caption{A sketch flowchart on the procedure of Split Knockoffs.} 
    \label{fig: procedure of Split Knockoffs}
\end{figure}

\begin{enumerate}
    \item In Section \ref{sec: sample split}, the dataset $\D=(X, y)$ is split into two parts $\D_1=(X_1, y_1)$  and $\D_2=(X_2, y_2)$ with sample sizes $n_1$ and $n_2$ respectively.
    \item In Section \ref{sec: beta(lambda)}, Split LASSO is performed on $\D_1$ to give an estimation $\beta(\lambda)$ for $\beta$,  which will be used in the subsequent estimates below.
    \item In Section \ref{sec: copy}, the \emph{Split Knockoff copy} $\tilde{A}_{\gamma_2}$ 
    is constructed on $\D_2$. Its associated fake features $\tilde{\gamma}$ serve as the control group, whose comparisons (in significance levels) with the original features $\gamma$ determine the estimated support set.
    \item In Section \ref{sec: z and tz}, the significance level $Z$ and the directional effect estimator $\widehat{\sign}$ of $\gamma$ are determined by $\beta(\lambda)$  through Split LASSO, while the significance level $\tilde{Z}$ of $\tilde{\gamma}$ is determined by $\beta(\lambda)$ and $\D_2$ through Split LASSO.

    \item In Section \ref{sec: w and hats}, the $W$ statistics are defined by $W = Z\odot\sign(Z-\tilde{Z})$, where $\odot$ represents the element-wise product. The estimated support set $\widehat{S}$ is designed to include features with large and positive $W$ statistics, i.e. sufficiently significant features with greater significance levels compared with their fake copies.
\end{enumerate}

The whole procedure is illustrated by the flowchart in Figure \ref{fig: procedure of Split Knockoffs}. The detailed constructions are given in the following subsections respectively.

\subsection{Sample Splitting}
\label{sec: sample split}

The dataset $\D=(X, y)$ is split into $\D_1=(X_1, y_1)$  and $\D_2=(X_2, y_2)$ with sample sizes $n_1$ and $n_2$ satisfying $n = n_1+n_2$. In addition, the construction of the Split Knockoff copy in Section \ref{sec: copy} requires $X_2^TX_2$ to be invertible and $n_2\ge m+p$. When a limited sample size fails the above requirements, variable screening can be applied on $\D_1$ to screen off a subset of $\beta$ and $\gamma$ to reduce the dimensionality, which will be discussed in Section \ref{sec: screen for hd}. 

It is important to note that the sample splitting scheme is a crucial component for controlling the $\dfdr$. Without employing the sample splitting scheme, conducting Split Knockoffs can lead to an inflation in $\dfdr$ control. A thorough discussion on this topic, including additional details, will be provided in Section \ref{sec: inflation}.

\subsection{Estimation for $\beta$ on $\D_1$}
\label{sec: beta(lambda)}

The first subset of data $\D_1$ is used for a preliminary estimation of parameter $\beta$. %, whose directional effects are not directly of concern in this paper, 
In fact, one obtains a continuous function $\beta(\lambda)$ with respect to $\lambda>0$, by solving the following Split LASSO regularization path restricted on $\D_1=(X_1, y_1)$,
\begin{align}
    \beta(\lambda):=\arg\min_{\beta}\min_\gamma\frac{1}{2n}\|y_1-X_1\beta\|_2^2+\frac{1}{2\nu}\|D\beta-\gamma\|_2^2+\lambda\|\gamma\|_1,\mbox{ for $\lambda>0$.} \label{eq: beta(lambda)}
\end{align}
Later, the function $\beta(\lambda)$ will be used to determine the significance levels $Z$, $\tilde{Z}$ and the directional effect estimator $\widehat{\sign}$ in Section \ref{sec: z and tz}.

\subsection{Split Knockoff Copy on $\D_2$}
\label{sec: copy}

The Split Knockoff copy is constructed based on the second dataset $\D_2$, whence independent to $\D_1$ and $\beta(\lambda)$. One starts from a reformulation of Split LASSO. Restricted on $\D_2 = (X_2, y_2)$ (with associated Gaussian noise $\varepsilon_2\in\R^{n_2}$), Split LASSO \eqref{eq: split lasso} can be viewed as the partially penalized LASSO on $\gamma$ (among $\beta$ and $\gamma$) 
\begin{align*}
    \frac{1}{2}\|\tilde{y}_2-A_{\beta_2}\beta-A_{\gamma_2}\gamma\|_2^2 + \lambda \|\gamma\|_1 = \frac{1}{2n}\|y_2-X_2\beta\|_2^2+\frac{1}{2\nu}\|D\beta-\gamma\|_2^2 + \lambda \|\gamma\|_1,
\end{align*}
with respect to the following reformulated regression model,
\begin{equation} \label{eq: reformulated model on D_2}
  \tilde{y}_2=A_{\beta_2}\beta^*+A_{\gamma_2}\gamma^*+\tilde{\varepsilon}_2,
\end{equation}
where ``2'' in the subscripts is a reminder that the symbols are restricted on $\D_2$, and
\begin{equation}
\tilde{y}_2=
\left(
\begin{array}{c}
\frac{y_2}{\sqrt{n_2}} \\
0_m
\end{array}
\right),
A_{\beta_2}=
\left(
\begin{array}{c}
\frac{X_2}{\sqrt{n_2}} \\
\frac{D}{\sqrt{\nu}} 
\end{array}
\right),
A_{\gamma_2}=
\left(
\begin{array}{c}
0_{n_2\times m} \\
-\frac{I_m}{\sqrt{\nu}} 
\end{array}
\right),
\tilde{\varepsilon}_2=
\left(
\begin{array}{c}
\frac{\varepsilon_2}{\sqrt{n}} \\
0_m 
\end{array}
\right).\label{eq: new design matrix}
\end{equation}

Crucially, model \eqref{eq: reformulated model on D_2} leads to an \emph{orthogonal design} $A_{\gamma_2}$ for $\gamma$, which will further lead to the \emph{orthogonal Split Knockoff copy} to be shown later in this section. Such an orthogonal design for $\gamma$ also results in weaker incoherence conditions to help identify strong nonnulls, at a risk of losing weak ones, which will be discussed in Section \ref{sec:pathconsistency}.

Yet, model \eqref{eq: reformulated model on D_2} introduces heterogeneous noise $\tilde{\varepsilon}_2$, which breaks the crucial exchangeability property in provable $\dfdr$ control \citep{barber2015controlling, barber2019knockoff} for Split Knockoffs (details in Section \ref{sec: fail exchange}). 
However, with the help of the orthogonal design $A_{\gamma_2}$ and the sample splitting scheme, the broken exchangeability in Split Knockoffs will not cause any loss in the $\dfdr$ control, as presented later by Theorem \ref{theorem: directional fdr} in Section \ref{sec: analysis}.

Now it's ready to construct the Split Knockoff copy matrix $\tilde{A}_{\gamma_2}$ with respect to the design matrix $A_{\gamma_2}$ in model \eqref{eq: reformulated model on D_2}, as a matrix satisfying
\begin{align}
    \tilde{A}_{\gamma_2}^T\tilde{A}_{\gamma_2} = A_{\gamma_2}^TA_{\gamma_2},\ 
    A_{\beta_2}^T\tilde{A}_{\gamma_2} = A_{\beta_2}^TA_{\gamma_2},\ A_{\gamma_2}^T\tilde{A}_{\gamma_2} =  A_{\gamma_2}^TA_{\gamma_2}-\diag(\vecs),\label{eq: copy}
\end{align}
similarly as in \cite{cao2021controlling}, where $\vecs\in\R^m$ is a proper non-negative vector. The detailed construction of the Split Knockoff copy matrix when $X_2^TX_2$ is invertible and $n_2\ge m+p$ is provided in Section \ref{sec: construct split knockoff}. 
In addition, Proposition \ref{prop: structure of tagamma} below (proof in Section \ref{sec: structure of tagamma}) characterizes the detailed structure of $\tilde{A}_{\gamma_2}$.

\begin{proposition}
\label{prop: structure of tagamma}
    Let $\tilde{A}_{\gamma_2,1}\in\R^{n_2\times m}$ be the submatrix consisting of the first $n_2$ rows of $\tilde{A}_{\gamma_2}$, and  $\tilde{A}_{\gamma_2,2}\in\R^{m\times m}$ be the remaining submatrix, then
    \begin{enumerate}
    \item $\tilde{A}_{\gamma_2,2}$ is diagonal: $\tilde{A}_{\gamma_2,2}= -\frac{I_m}{\sqrt{\nu}}+\sqrt{\nu}\diag(\vecs)$.
        \item $\tilde{A}_{\gamma_2,1}$ converts $X_2$ to $D$ up to a scaling: $\tilde{A}_{\gamma_2,1}^TX_2=-\sqrt{n}\diag(\vecs)D$.
        \item $\tilde{A}_{\gamma_2,1}$ is orthogonal up to a scaling:
    $
    \tilde{A}_{\gamma_2,1}^T\tilde{A}_{\gamma_2,1}=\diag(\vecs)(2I_m-\diag(\vecs)\nu)$.
    \end{enumerate}
\end{proposition}
By Proposition \ref{prop: structure of tagamma}, the Split Knockoff copy $\tilde{A}_{\gamma_2}$ depends on $D$ and $X_2$ but not on $y_2$. Particularly, $\tilde{A}_{\gamma_2}$ is an orthogonal matrix whose submatrices $\tilde{A}_{\gamma_2,1}$ and $\tilde{A}_{\gamma_2,2}$ are both orthogonal, up to a scaling. The orthogonality directly leads to the independence between $\sign(W_i)$ presented later by Lemma \ref{lemma: est of sign} in Section \ref{sec: analysis}, which is crucial for the provable $\dfdr$ control.

\subsection{Significance Levels and Directional Effect Estimator}
\label{sec: z and tz}

In this section, the significance levels $Z$ and $\tilde{Z}$ of the original feature and its associated split knockoff copy feature will be defined through Split LASSO regularization paths. 
In ideal cases of the Split LASSO path, the non-nulls become nonzero at larger $\lambda$ compared with the nulls. Therefore, the supremum of the regularization parameter $\lambda$ where features become nonzero can be used to represent the significance of the features. 

Specifically, the solution path $\gamma(\lambda)$ and $\tilde{\gamma}(\lambda)$ in Split LASSO are defined formally by $\beta(\lambda)$ and $\D_2$ as
\begin{subequations}
\label{eq: stage 1 and 2}
    \begin{align} 
    \gamma(\lambda)&:=\arg\min_{\gamma} \frac12\|\tilde y_2-A_{\beta_2}\beta(\lambda)-A_{\gamma_2}\gamma\|_2^2+\lambda\|\gamma\|_1,\label{eq:stage1}\\
    \tilde{\gamma}(\lambda) &:=\arg \min_{\tilde{\gamma}} \frac12\|\tilde y_2-A_{\beta_2}\beta(\lambda)-\tilde{A}_{\gamma_2}\tilde{\gamma}\|_2^2+\lambda\|\tilde{\gamma}\|_1, \label{eq:stage2}
    \end{align}
\end{subequations}
for $\lambda>0$. In particular, although $\gamma(\lambda)$ in Equation \eqref{eq:stage1} is ostensibly correlated with both $\beta(\lambda)$ and $\D_2$, it is in fact determined by $\beta(\lambda)$ only, as
\begin{align}
    \gamma(\lambda)&=\arg\min_{\gamma}\frac12\|\tilde y_2-A_{\beta_2}\beta(\lambda)-A_{\gamma_2}\gamma\|_2^2+\lambda\|\gamma\|_1 \nonumber\\
    &= \arg\min_{\gamma}\frac12\left\|\left(
\begin{array}{c}
\frac{y_2}{\sqrt{n_2}} \\
0_m
\end{array}
\right) - 
\left(
\begin{array}{c}
\frac{X_2}{\sqrt{n_2}} \\
\frac{D}{\sqrt{\nu}} 
\end{array}
\right)\beta(\lambda)-
\left(
\begin{array}{c}
0_{n_2\times m} \\
-\frac{I_m}{\sqrt{\nu}} 
\end{array}
\right)\gamma
\right\|_2^2+\lambda\|\gamma\|_1,\nonumber\\
& = \arg\min_{\gamma}\frac{1}{2\nu}\|D\beta(\lambda)-\gamma\|_2^2+\lambda\|\gamma\|_1.\label{eq: simp stage1}
\end{align}
Such a property is exploited in the construction of $W$ statistics to ensure the conditional independence between $|W|$ and $\sign(W)$, to be presented later in Section \ref{sec: w and hats}.

For all $i$, the significance levels $Z_i$ and $\tilde{Z}_i$ are given by 
\begin{align}
    Z_i=  \sup\left\{\lambda:\gamma_i(\lambda)\neq0\right\}, \ \tilde{Z}_i=  \sup\left\{\lambda:\tilde\gamma_i(\lambda)\neq0\right\}.\label{def: z and tilde z}
\end{align}
For all $i$, the sign of $\gamma_i(\lambda)$ upon being nonzero in Equation \eqref{eq:stage1} or \eqref{eq: simp stage1} is recorded as
\begin{align}
    r_i=  \lim_{\lambda\to Z_i-}\sign(\gamma_i(\lambda)). \label{def: r}
\end{align}
Such $r_i$ will be used to estimate the directional effect $\widehat{\sign}_i$ in the next Section \ref{sec: w and hats}.

\subsection{$W$ statistics and Estimated Support Set}
\label{sec: w and hats}

Intuitively, a selected feature should have high significance greater than its fake split knockoff copy. Such an intuition can be characterized by the following $W$ statistics\footnote{This particular $W$ statistics is known as $W^{\mathrm{S}}$ in \cite{cao2021controlling}.},
\begin{align}
    \label{eq: def_w}
    W = Z\odot\sign(Z-\tilde{Z}).
\end{align}
For the $i$-th feature, $W_i>0$ indicates a desired feature of higher significance than its fake copy, $Z_i>\tilde{Z}_i$; while $W_i<0$ suggests its lower significance than the fake copy, i.e. a false discovery. 
This definition of $W$ statistics is a slight modification of that in \cite{barber2015controlling,barber2019knockoff} with the following merits.
\begin{enumerate}
    \item From Equation \eqref{eq: def_w}, magnitude $|W| = Z$ and $\sign(W) = \sign(Z-\tilde{Z})$ will be independent from each other conditional on $\D_1$, since $Z$ is determined by $\beta(\lambda)$ from $\D_1$ via Equation \eqref{eq: simp stage1}, while $\tilde{Z}$ is determined by both $\beta(\lambda)$ and $\D_2$ via Equation \eqref{eq:stage2}. This independence is crucial for the supermartingale inequality \eqref{eq: conditional target} in achieving theoretical $\dfdr$ control, which will be shown in Section \ref{sec: analysis}.
    \item Our definition of $W$ leads to an improved selection power without losing the $\dfdr$ control, compared with the original definition of $W$ statistics in \cite{barber2015controlling,barber2019knockoff}, which will be discussed in Section \ref{sec: inclusion} with Proposition \ref{prop: inclusion}.
\end{enumerate}
From the definition of $W$ statistics, features with large and positive $W$ statistics should be selected. For any preset nominal $\dfdr$ level $q$, a data dependent threshold $T_q$ based on the $W$ statistics \eqref{eq: def_w} is defined similarly to \cite{barber2015controlling, barber2019knockoff} by 
\begin{align*}
  &\mbox{(Split Knockoff)}\ \ \ \ \  T_q=\min\left\{\lambda:\frac{|\{i:W_i\le-\lambda\}|}{1\vee|\{i:W_i\ge \lambda\}|}\le q\right\},\\
  &\mbox{(Split Knockoff+)}\ \ \ T_q=\min\left\{\lambda:\frac{1+|\{i:W_i\le-\lambda\}|}{1\vee|\{i:W_i\ge \lambda\}|}\le q\right\},
\end{align*}
or $T_q=+\infty$ if the respective set is empty.
In both cases, the selector ($\widehat{S}$) and estimated direction effects on selected features ($\widehat{\sign}_{\widehat{S}}$) are given by
\begin{equation}
  \widehat{S}:=\{i:W_i\ge T_q\},\ \ \ \  \widehat{\sign}_{\widehat{S}} := r_{\widehat{S}}.\label{eq: output}
\end{equation}

An analysis for the $\dfdr$ control by Split Knockoffs will be given in Section \ref{sec: analysis}.

\section{$\textbf{FDR}_{\textbf{dir}}$ Control of Split Knockoffs}

\label{sec: analysis}

In this section, we present the theoretical results on the $\dfdr$ control of Split Knockoffs. Specifically, Theorem \ref{theorem: directional fdr} guarantees the  universal $\dfdr$ control of Split Knockoffs with respect to any tuning parameter $\nu>0$. For Split Knockoff+, the exact $\dfdr$ is under control, while Split Knockoff controls the modified directional false discovery rate ($\mdfdr$), firstly defined in \cite{barber2019knockoff}, where $q^{-1}$ is added in the denominator, having little effects when the support set is large. 

\begin{theorem}
    \label{theorem: directional fdr}
    For any linear transformation $D$, any $0<q<1$ and $\nu>0$, there holds:
    \begin{itemize}
        \item[(a)] ($\mdfdr$ of Split Knockoff)
    \begin{align*}
        \E\left[\frac{|\{i\in \widehat{S}: \widehat{\sign}_i \neq \sign(\gamma^*_i)\}|}{|\widehat{S}|+q^{-1}}\right]\le \min(\alpha(\nu), 1) q,
    \end{align*}
    \item[(b)] ($\dfdr$ of Split Knockoff+)
    \begin{align*}
        \E\left[\frac{|\{i\in \widehat{S}: \widehat{\sign}_i \neq \sign(\gamma^*_i)\}|}{|\widehat{S}|\vee 1}\right]\le \min(\alpha(\nu), 1) q,
    \end{align*}
    \end{itemize}
    where $\alpha(\nu)$ given by Equation \eqref{def: alpha(nu)} is a decreasing function satisfying $\lim_{\nu\to\infty}\alpha(\nu) = 0$.
\end{theorem}

\begin{remark}
As presented in Theorem \ref{theorem: directional fdr}, Split Knockoffs achieve the desired $\dfdr$ control for all $\nu>0$. Although the $\dfdr$ of Split Knockoffs can decrease to zero with the increase of $\nu$,  
overshooting $\nu$ to push the $\dfdr$ toward zero may cause a loss in the selection power. In practice, it is recommended to optimize over $\nu$ for the best selection power since the $\dfdr$ of Split Knockoffs is below the nominal level uniformly for all $\nu>0$. The details will be discussed by simulation experiments in Section \ref{sec: simu_exp} and the sign consistency of Split LASSO in Section \ref{sec:pathconsistency}. 
\end{remark}

In the following, a brief guideline on how to achieve Theorem \ref{theorem: directional fdr} is provided, while a complete proof of Theorem \ref{theorem: directional fdr} will be given in Section \ref{sec: proof thm}. First of all, from the standard procedure of Knockoffs in \cite{barber2015controlling, barber2019knockoff}, the following inequality is sufficient for Theorem \ref{theorem: directional fdr} (details in Section \ref{sec: proof thm}):
\begin{align}
    \Expect\left[\frac{\sum_{i}1\{W_i\ge T_q, \widehat{\sign}_i\neq\sign(\gamma^*_i)\}}{1+\sum_{i}1\{W_i\le -T_q, \widehat{\sign}_i\neq\sign(\gamma^*_i)\}}\right]\le \min(\alpha(\nu), 1).\label{eq: target}
\end{align}
On the way to achieve  Equation \eqref{eq: target}, the sample splitting scheme and the variable splitting scheme build up the following two crucial points, respectively.
\begin{enumerate}
    \item \label{benefit 1}\textbf{Benefits of Sample Splitting.} 
    Conditional on $\D_1$, the magnitude of $W$ statistics ($|W|$) and the sign estimator ($\widehat{\sign}$) are \emph{independent} from the sign of $W$ statistics ($\sign(W)$).%, benefiting from the sample splitting scheme;
    \item \label{benefit 2}\textbf{Benefits of Variable Splitting.} 
    Conditional on $\D_1$, 
    the signs of $W$ statistics ($\sign(W_i)$) are independent from each other. Moreover, for any $\D_1$, there holds
    \begin{align*}
        &\left.\left\{\max_{i:  \widehat{\sign}_i\neq\sign(\gamma^*_i)}\frac{\Prob[W_i>0]}{\Prob[W_i<0]}\right|\D_1\right\}\le\min(\alpha(\nu), 1).
    \end{align*}
\end{enumerate}
With the above benefits, informally there holds
\begin{align}
    &\Expect\left[\frac{\sum_{i}1\{W_i\ge T_q, \widehat{\sign}_i\neq\sign(\gamma^*_i)\}}{1+\sum_{i}1\{W_i\le -T_q, \widehat{\sign}_i\neq\sign(\gamma^*_i)\}}\right] ,\nonumber\\
    =& \Expect\left[\left.\Expect\left[ \frac{\sum_{i: |W_i|\ge T_q, \widehat{\sign}_i\neq\sign(\gamma^*_i)}1\{W_i>0\}}{1+\sum_{i: |W_i|\ge T_q, \widehat{\sign}_i\neq\sign(\gamma^*_i)}1\{W_i<0\}}\right|\D_1\right]\right],\nonumber\\
    \lesssim & \left.\left\{\max_{i: |W_i|\ge T_q,\widehat{\sign}_i\neq\sign(\gamma^*_i)}\frac{\Prob[W_i>0]}{\Prob[W_i<0]}\right|\D_1\right\}\le\left.\left\{\max_{i:  \widehat{\sign}_i\neq\sign(\gamma^*_i)}\frac{\Prob[W_i>0]}{\Prob[W_i<0]}\right|\D_1\right\}\le\min(\alpha(\nu), 1),\label{eq: bound by martingale}
\end{align}
where the last line can be rigorously formalized by a supermartingale inequality similarly as Lemma 1 in \cite{barber2019knockoff} and an extension to that in \cite{cao2021controlling}. 
The respective technical details will be deferred to the proof of Theorem \ref{theorem: directional fdr} in Section \ref{sec: proof thm} to save pages. In the following, we focus on how the sample splitting scheme and the variable splitting scheme in Split Knockoffs enjoy the above benefits respectively.

\subsection{Benefits of Sample Splitting}

The insight on the benefits of sample splitting can be seen from a deeper view into the Split LASSO paths, i.e. the Karush–Kuhn–Tucker (KKT) conditions that the Split LASSO path \eqref{eq: stage 1 and 2} satisfies, which is given by Lemma \ref{lemma: distribution zeta}. 
\begin{lemma}
\label{lemma: distribution zeta}
The KKT conditions that Equation \eqref{eq: stage 1 and 2} should satisfy is
    \begin{subequations}
        \label{eq: kkts}
        \begin{align}
        &\lambda\rho(\lambda) + \frac{\gamma(\lambda)}{\nu}= \frac{D\beta(\lambda)}{\nu},\label{eq: feature sig}\\
        &\lambda\tilde{\rho}(\lambda) + \frac{\tilde{\gamma}(\lambda)}{\nu} = \frac{D\beta(\lambda)}{\nu} +\zeta,\label{eq: knockoff sig}
    \end{align}
\end{subequations}
where $\rho(\lambda) \in \partial \|\gamma(\lambda)\|_1$, $\tilde{\rho}(\lambda) \in \partial \|\tilde{\gamma}(\lambda)\|_1$, and $\zeta:=\tilde{A}_{\gamma_2}^T\tilde{y}_2\in\R^m$ follows the distribution
\begin{align}
        \zeta\sim \mathcal{N}\left(-\diag(\vecs)\gamma^*, \frac{1}{n_2}\diag(\vecs)(2I_m-\diag(\vecs)\nu)\sigma^2\right), \label{eq: zeta dis}
    \end{align}
    where $\mathcal{N}(\mu, \Sigma)$ denotes the multivariate Gaussian distribution.
\end{lemma}

\begin{remark}
    Equation \eqref{eq: zeta dis} shows that $\zeta$ consists of independent Gaussian random variables. Such a property results from the orthogonality of split knockoff copy matrix $\tilde{A}_{\gamma_2}$ given in Proposition \ref{prop: structure of tagamma}. 
\end{remark}

Following Equation \eqref{eq: simp stage1}, 
the solution path $\gamma(\lambda)$ in Equation \eqref{eq:stage1} is determined by $\D_1$ through $\beta(\lambda)$. Since $|W|=Z$ \eqref{def: z and tilde z} and $\widehat{\sign}_i=r_i$ \eqref{def: r} are determined by $\gamma(\lambda)$, $\D_1$ thus determines $|W|$ and $\widehat{\sign}$.

On the other hand, the signs of the $W$ statistics ($\sign(W_i) = \sign(Z_i-\tilde{Z_i})$) rely on the difference between Equation \eqref{eq:stage1} and \eqref{eq:stage2}, where the only difference in their KKT conditions \eqref{eq: kkts} lies in the random variable $\zeta=\tilde{A}_{\gamma_2}^T\tilde{y}_2$ determined by $\D_2$. Therefore, conditional on $\D_1$ which determines $\beta(\lambda)$ and $\gamma(\lambda)$ (consequently $|W|$ and $\widehat{\sign}$), $\sign(W_i)$ is determined by $\D_2$ (though $\zeta_i$) for all $i$. 
In this regard, conditional on $\D_1$, $|W|$ and $\widehat{\sign}$ are independent from $\sign(W)$.

Moreover, it can be further shown that the sample splitting scheme is essential for achieving the $\dfdr$ control. Specifically, conducting Split Knockoffs without implementing the sample splitting scheme can lead to an inflation in $\dfdr$ control. We will provide a detailed discussion on this topic in Section \ref{sec: inflation}.

\subsection{Benefits of Variable Splitting}

Furthermore, Lemma \ref{lemma: distribution zeta} shows that conditional on $\D_1$, $\sign(W_i)$ is determined by $\zeta_i$, independent Gaussian random variables. 
Consequently, $\sign(W_i)$ are independent from each other, conditional on $\D_1$ that determines $|W|$ and $\widehat{\sign}$. With further detailed calculation, Lemma \ref{lemma: est of sign} shows that Bernoulli random variables $\sign(W_i)$ are biased toward the negative sign.

\begin{lemma}
    \label{lemma: est of sign}
    Conditional on $\D_1$, $\sign(W_i)$ are independent random variables. Furthermore, 
    for $i\in \{\widehat{\sign}_i \neq \sign(\gamma^*_i)\}$, there holds
    \begin{align*}
        \Prob[W_i<0] \ge \max\left(\frac{1}{2}, f(\nu)\right),
    \end{align*}
    where $f(\nu)$ is an increasing function of $\nu$ defined in Equation \eqref{def: f(nu)} s.t. $\lim_{\nu\to\infty}f(\nu) = 1$.
\end{lemma}

Lemma \ref{lemma: est of sign} summarizes the benefits brought by the sample splitting scheme and the variable splitting scheme. In particular, it brings the following two properties:
\begin{enumerate}
    \item independence among $\sign(W_i)$ conditional on $\D_1$ (which determines $|W|$, $\widehat{\sign}$) that enables a supermartingale inequality as in Lemma 1 in \cite{barber2019knockoff}; %to give a rigorous formalization for Equation \eqref{eq: bound by martingale};
    \item the desired lower bound on $\Prob[W_i<0]$ for $i\in \{\widehat{\sign}_i \neq \sign(\gamma^*_i)\}$ which enables the upper bound in Equation \eqref{eq: target} and \eqref{eq: bound by martingale}.
\end{enumerate}

With the properties above, the following supermartingale inequality (whose detailed proof is provided in Section \ref{sec: proof thm}) can be achieved,
\begin{align}
    \Expect\left.\left[\frac{\sum_{i}1\{W_i\ge T_q, \widehat{\sign}_i\neq\sign(\gamma^*_i)\}}{1+\sum_{i}1\{W_i\le -T_q, \widehat{\sign}_i\neq\sign(\gamma^*_i)\}}\right|\D_1\right]\le \min(\alpha(\nu), 1),\label{eq: conditional target}
\end{align}
where $T_q$ is the stopping time on the supermartingale structure associated with 
$$\left.\frac{\sum_{i}1\{W_i\ge T, \widehat{\sign}_i\neq\sign(\gamma^*_i)\}}{1+\sum_{i}1\{W_i\le -T, \widehat{\sign}_i\neq\sign(\gamma^*_i)\}}\right|\D_1,$$ for $T>0$. 
Taking expectation over $\D_1$ in Equation \eqref{eq: conditional target} leads to Equation \eqref{eq: target}. 
The detailed proof of Theorem \ref{theorem: directional fdr} is given in Section \ref{sec: proof thm}, while the proof of Lemma \ref{lemma: distribution zeta} and \ref{lemma: est of sign} are provided in Section \ref{sec: proof zeta} and \ref{sec: proof lemma} respectively.

\section{Simulation Experiments}
\label{sec: simu_exp}

In this section, simulation experiments  are conducted to validate the effectiveness of Split Knockoffs in both $\dfdr$ control and selection power, where standard Knockoffs \cite{barber2019knockoff} are implemented for comparisons when both are applicable.

\subsection{Models}
In simulation experiments, the rows of the design matrix $X\in \mathbb R^{n\times p}$ are generated  independent and identically distributed (i.i.d.) from $\mathcal{N}(0_p, \Sigma)$ where $\Sigma_{i,j}=0.5^{|i-j|}$ for all $i,j$. We take $n=500$ and $p=100$ in this section. The regression coefficient  $\beta^*\in\mathbb R^p$ is taken as
\begin{equation*}
    \beta_i^*:=\left\{
    \begin{array}{ccl}
        1   &   & i \le 20,\ i \equiv 0, -1 (\mathrm{mod}\ 3),\\
        0   &   & \mathrm{otherwise}.
    \end{array} \right.
\end{equation*}
Then the response vector $y\in\R^n$ are generated from,
$$y = X \beta^* + \varepsilon,$$
where $\varepsilon\in\R^n$ is generated from $\Nm(0_n, I_n)$. Three types of linear transformations $\{D_i\}_{i=1}^3$ are tested in this section, where $\gamma^*$ is generated by $\gamma^* = D_i\beta^*$ for each $D_i$.
\begin{itemize}
    \item $\beta^*$ is sparse itself, so we take $D_1=I_p\in\R^{p\times p}$. In this case, $m = p$. 
    \item $\beta^*$ is an uni-dimensional piece-wise constant function, so we take $D_2$  as the graph difference operator on a line, i.e. $D_2\in \mathbb R^{(p-1)\times p}$, $D_2(i, j)=1_{\{i=j\}}-1_{\{i=j-1\}}$, for $(i, j) \in \{1, 2, \cdots p-1\}\times \{1, 2, \cdots p\}$. In this case, $m=p-1<p$.
    \item Combining the above two points, $\beta^*$ is a sparse piece-wise constant function, so we take $D_3=\left[\begin{array}{c}
        D_1 \\
        D_2
    \end{array}\right]\in \R^{(2p-1)\times p}$. In this case, $m=2p-1>p$.
\end{itemize}

The first two cases, $D_1$, $D_2$, are two special cases where the linear transformation has full row rank, and Knockoffs can be applicable (details in Section \ref{sec: special cases}). In this regard, the Knockoff method are implemented for comparisons in these cases. Meanwhile, for $D_3$ where $\rank(D_3) = p<2p-1=m$, only the Split Knockoff method will be implemented, as Knockoffs are no longer applicable. 

In addition to the $\dfdr$ control, the selection power defined as
\begin{align*}
    \mathrm{Power} = \Expect\left[\frac{|\{i:i\in\widehat{S}, \widehat{\sign}_i = \sign(\gamma^*_i)\}|}{|\{i: \sign(\gamma^*_i)\neq 0\}|}\right],
\end{align*}
is also presented to evaluate the accuracy of the directional effect estimation.

\subsection{Results}
Figure \ref{fig: comparison} presents the performance of Split Knockoffs and Knockoffs (when applicable) in $\dfdr$ control and selection power, where the nominal $\dfdr$ level is taken to be $q= 0.2$. The performance of Split Knockoffs is presented for a sequence of $\log_{10}\nu$ between 0 and 2 with a step size 0.2, while the effects of other parameters aside from $\nu$ are presented in Section \ref{sec: supp simu}. For Split Knockoffs, we randomly divide the full dataset $\D$ into $\D_1$ and $\D_2$ with sample sizes $n_1=200$ and $n_2 = 300$, while Knockoffs are implemented on the full dataset $\D$ with a sample size $n = 500$. The Knockoff copy is taken as the SDP Knockoffs \citep{barber2015controlling}, while the Split Knockoff copy is constructed in the equi-correlated way for simplicity (i.e. by Equation \eqref{eq: equi corre} in Section \ref{sec: construct split knockoff}).  
The feature importance statistics of Knockoffs are constructed through the emergence point of the LASSO path (the way presented in Section 1.2 in \cite{barber2015controlling}).

\begin{figure}[!ht]
\centering
\subfigure[$\dfdr$ in $D_1$]{
\begin{minipage}[t]{0.33\textwidth}
\centering
\includegraphics[width=\textwidth]{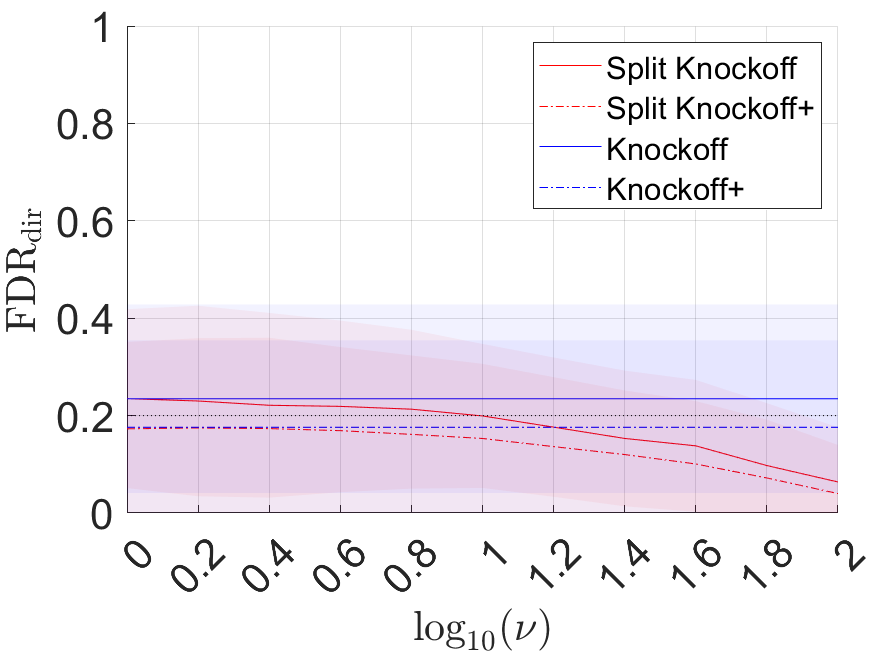}
%\caption{FDR comparison in $D_1$}
\label{fig: d1 fdr}
\end{minipage}%
}%
\subfigure[$\dfdr$ in $D_2$]{
\begin{minipage}[t]{0.33\textwidth}
\centering
\includegraphics[width=\textwidth]{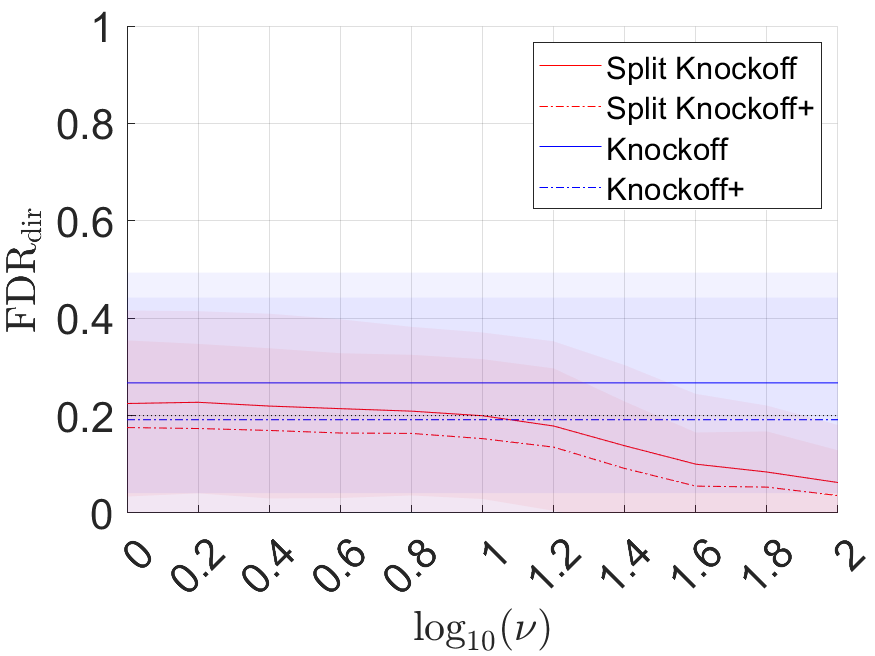}
%\caption{FDR comparison in $D_2$}
\label{fig: d2 fdr}
\end{minipage}%
}%
\subfigure[$\dfdr$ in $D_3$]{
\begin{minipage}[t]{0.33\textwidth}
\centering
\includegraphics[width=\textwidth]{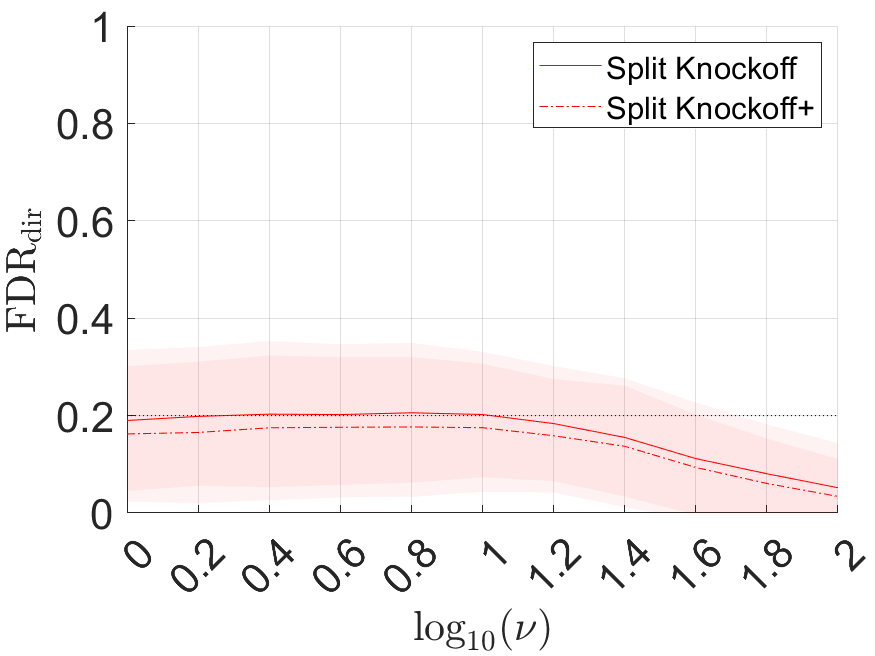}
%\caption{FDR performance in $D_3$}
\label{fig: d3 fdr}
\end{minipage}%
}%

\centering
\subfigure[Power in $D_1$]{
\begin{minipage}[t]{0.33\textwidth}
\centering
\includegraphics[width=\textwidth]{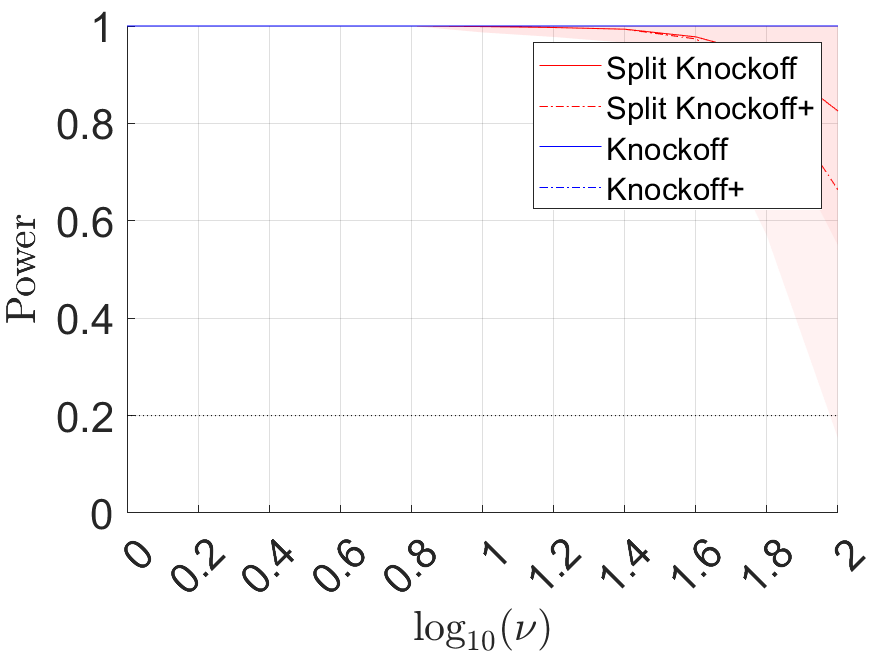}
%\caption{Power comparison in $D_1$}
\label{fig: d1 power}
\end{minipage}%
}%
\subfigure[Power in $D_2$]{
\begin{minipage}[t]{0.33\textwidth}
\centering
\includegraphics[width=\textwidth]{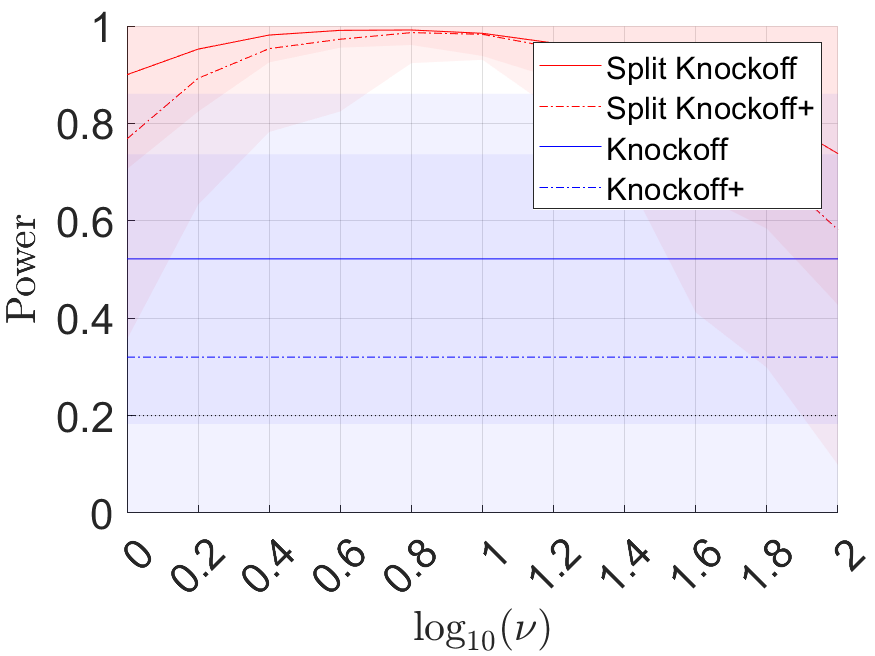}
%\caption{Power comparison in $D_2$}
\label{fig: d2 power}
\end{minipage}%
}%
\subfigure[Power in $D_3$]{
\begin{minipage}[t]{0.33\textwidth}
\centering
\includegraphics[width=\textwidth]{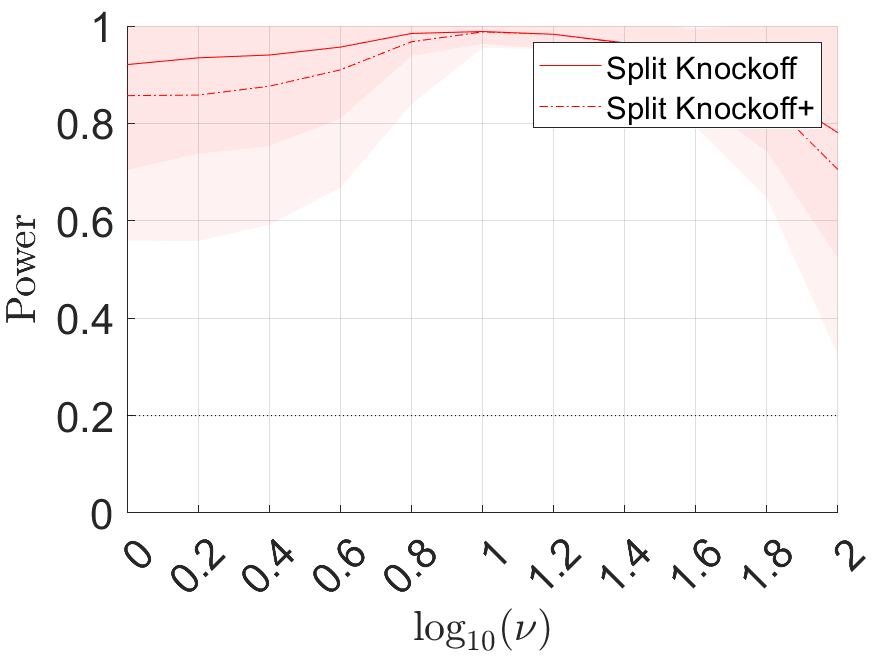}
%\caption{Power performance in $D_3$}
\label{fig: d3 power}
\end{minipage}%
}%

\caption{Performance of Split Knockoffs and Knockoffs (when applicable) in $\dfdr$ and power ($q=0.2$). The curves in the figures represent the average performance in 200 simulation instances, while the shaded areas represent the 80\% confidence intervals truncated to the range $[0, 1]$.}
\label{fig: comparison}
\end{figure}

\noindent{\bf{$\textbf{FDR}_{\textbf{dir}}$}}: Figure \ref{fig: comparison} shows that Split Knockoffs achieve desired $\dfdr$ control in all cases, where Split Knockoff+ improves $\dfdr$ over Split Knockoff slightly. For $D_1$ and $D_2$ where Knockoffs are applicable, Split Knockoffs shows comparable or possibly better $\dfdr$ control compared with Knockoffs. In particular, when $\nu$ is large, the $\dfdr$ for Split Knockoffs decreases  to zero with the increase of $\nu$ in all cases, as predicted by Theorem \ref{theorem: directional fdr} in Section \ref{sec: analysis}.

\begin{table}[!ht]
    \caption{Performance of Split Knockoffs and Knockoffs (when applicable) in $\dfdr$ and power ($q=0.2$), where the means and standard deviations in 200 simulation instances are provided. For Split Knockoffs, $\nu$ is chosen by cross validation with Split LASSO on $\D_1$.}
    \label{tab: simu_compare}
    \centering
    \begin{tabular}{|c|cccc|}
    \hline
        Performance & Knockoff & Split Knockoff & Knockoff+ & Split Knockoff+\\
        \hline
        \multirow{2}*{FDR in $D_1$} & 0.2346 & 0.2348 &  0.1759 &  0.1725 \\
        ~& $\pm$0.1505 & $\pm$0.1428 & $\pm$0.1389 & $\pm$0.1380 \\
        \multirow{2}*{Power in $D_1$} & 1.0000 & 1.0000 & 1.0000 & 1.0000 \\
        ~& $\pm$0.0000 & $\pm$0.0000 & $\pm$0.0000 & $\pm$0.0000 \\
        \hline
        \multirow{2}*{FDR in $D_2$} & 0.2672 & 0.2216 &  0.1916 & 0.1730  \\
        ~& $\pm$0.1760 & $\pm$0.1442 & $\pm$0.1948 & $\pm$0.1363 \\
        \multirow{2}*{Power in $D_2$} & 0.5218 & 0.9139 &  0.3200  & 0.8071\\
        ~& $\pm$0.2637 & $\pm$0.1416 & $\pm$0.3240 & $\pm$0.2929 \\
        \hline
        \multirow{2}*{FDR in $D_3$} & N/A & 0.1920 & N/A &  0.1631 \\
        ~& N/A & $\pm$0.1125 & N/A & $\pm$0.1080 \\
        \multirow{2}*{Power in $D_3$} & N/A & 0.9228  & N/A &  0.8572 \\
        ~& N/A & $\pm$0.1659 & N/A & $\pm$0.2314 \\
    \hline
    \end{tabular}
\end{table}

\noindent{\bf Power}: For a wide range of $\nu$, Split Knockoffs exhibit desired selection power in all cases, where Split Knockoff+ sacrifices a little bit power for improved $\dfdr$ control compared with Split Knockoff. Compared with Knockoffs, Split Knockoffs achieve comparable selection power when the linear transformation is trivial ($D_1$), and higher selection power when the linear transformation is non-trivial ($D_2$). The improved selection power benefits from both the orthogonal design of $\gamma$ introduced in Equation \eqref{eq: reformulated model on D_2} and our newly defined $W$ statistics \eqref{eq: def_w} which induces an inclusion property on the selection set (Proposition \ref{prop: inclusion} in Section \ref{sec: inclusion}).

For Split Knockoffs, 
increasing $\nu$ can improve the power to select strong nonnulls as the incoherence condition gets improved, at a possible risk of losing weak nonnulls. This phenomenon is discussed through the sign consistency of Split LASSO in Section \ref{sec:pathconsistency}. In practice, it is recommended to optimize $\nu$ on $\D_1$ for better selection power without losing the $\dfdr$ control. In particular, Table \ref{tab: simu_compare} presents the performance of Split Knockoffs when $\nu$ is chosen by cross validation with Split LASSO on $\D_1$, with comparisons against Knockoffs when applicable. In Table \ref{tab: simu_compare}, Split Knockoffs achieves desired $\dfdr$ control and enjoys improved selection power compared with Knockoffs when the transformation is nontrivial, validating the effectiveness on this choice of $\nu$.

\section{Applications}

\label{sec: applications}

In this section, we implement Split Knockoffs on two real world applications. In the first application, we select the lesion brain regions and abnormal connections of brain regions with large activation changes in Alzheimer's Disease with structural Magnetic Resonance Imaging (sMRI) data. In the second application, we make pairwise comparisons on human ages based on voluntarily annotated data on face images.

\subsection{Alzheimer's Disease} 
\label{sec: AD-exp}

In this application, we study lesion brain regions as well as  abnormal connections of brain regions with large activation changes in Alzheimer's Disease (AD). The data is obtained from ADNI (\url{http://adni.loni.ucla.edu}) dataset, acquired by structural Magnetic Resonance Imaging (MRI) scan. In total, the dataset contains $n=752$ samples. 
For each image, the Dartel VBM \citep{ashburner2007fast} is implemented for pre-processing, followed by the toolbox \emph{Statistical Parametric Mapping} (SPM) for segmentation of gray matter (GM), white matter (WM) and cerebral spinal fluid (CSF). Then Automatic Anatomical Labeling (AAL) atlas \citep{tzourio2002automated} is applied to partition the whole brain into $p=90$ Cerebrum brain anatomical regions, with the volume of each region (summation of all GMs in the region) provided.

The design matrix $X \in \mathbb{R}^{n \times p}$ consists of the region volumes obtained by  structural MRI scan, whose element $X_{i,j}$ represents the (column normalized) volume of region $j$ in the subject $i$'s brain respectively. The response vector $y \in \mathbb{R}^{n}$ denotes the Alzheimer's Disease Assessment Scale (ADAS), which was originally designed to assess the severity of cognitive dysfunction \citep{rosenwg1984scale} and was later found to be able to clinically distinguish the diagnosed AD from normal controls \citep{zec1992alzheimer}. The following two types of transformations are implemented in this section:
\begin{enumerate}
    \item[(a)] $D=I_{p}$, for selecting the lesion Cerebrum  brain regions, where $m=p=90$;
    \item[(b)] $D$ is the graph difference operator on the connectivity graph of Cerebrum  brain regions, for selecting the abnormal connections of regions with large activation changes accounting for the disease, where $m=463>p=90$.
\end{enumerate}
For the region selection, Split Knockoff is implemented with respect to a sequence of $\log_{10}\nu$ between -2 and 0 with a step size 0.2, while for the connection selection, Split Knockoff is implemented on $\log_{10}\nu$ between 0 and 2. The dataset is split into two with sample sizes $n_1 = 150$ and $n_2 = 602$ respectively. The Knockoff is implemented for comparisons in the region selection on the full dataset with a sample size $n = 752$, under the same settings as in Section \ref{sec: simu_exp}.

\begin{table}[!ht]
    \caption{Selected regions by Knockoff and Split Knockoff in Alzheimer's Disease ($q = 0.2$). The sign $-1$/$1$ marks the estimated atrophied/enlarged. regions accounting for the disease. %, while the positive sign $1$ marks regions with possible enlargement. 
    For Split Knockoff, the optimal choice of $\nu$ selected by Split LASSO  on $\D_1$ is $\log_{10}\nu = -2$.}
    \centering
    \resizebox{\textwidth}{!}{
    \begin{tabular}{|c|c|cccc|}
    \hline
        \multirow{2}*{Region} & \multirow{2}*{Knockoff}  & \multicolumn{4}{c|}{Split Knockoff with $\log_{10}\nu$}\\
        \cline{3-6}
        ~ & ~ & \textbf{\{-2, -1.8\}}& -1.6 & \{-1.4, -1.2\} & \{-1:\ 0.2:\ 0\}\\
        \hline
        Hippocampus (L) &  -1&   \textbf{-1} & -1 & -1 & -1 \\
        Hippocampus (R) & -1&   \textbf{-1} & -1 & -1 & -1 \\
        Middle temporal gyrus (L) & -1  & \textbf{-1}  & -1 & -1 & -1 \\
        Middle temporal gyrus (R) & -1  & \textbf{-1}  &  &  &  \\
        Inferior temporal gyrus (L) & -1  & \textbf{-1}  &  &  &  \\
        Amygdala (L) &   &   \textbf{-1} & -1 & -1 & \\
        Supramarginal gyrus (R) &  &  &  & -1 & -1 \\
        Inferior frontal gyrus, opercular part (L) &  &  &  &  & -1\\
        
        Inferior parietal gyrus (R) &  &  &  & & -1 \\
        
        \hline
      \end{tabular}
      }
      \label{tab: region}
\end{table}

 \paragraph*{Region Selection:} The target $\dfdr$ is set to be $q = 0.2$, and the region selection results are presented in Table \ref{tab: region}. For Split Knockoff, the minimal cross-validation loss happens in the case $\log_{10}\nu =-2$, where the selected six regions are marked in Figure \ref{fig: ad front}. Both Split Knockoff and Knockoff select two-side Hippocampus, two-side Middle Temporal Gyrus, and Inferior parietal gyrus (L) as degenerated regions, which have been found to suffer from atrophy among AD patients with functional deficits in language, memory processing \citep{vemuri2010role,schuff2009mri,visser2002medial}, and sensory interpretation \citep{radua2010neural, greene2010subregions}. Besides, the Split Knockoff additionally selects the Amygdala (L), which is involved in memory, decision making, and emotional responses \citep{gupta2011amygdala}, and has been found to suffer from atrophy in the progression of AD \citep{vereecken1994neuron}.

\paragraph*{Connection Selection:}

In this experiment, we select abnormal connections of neighbouring regions with large activation changes accounting for the disease. 
We set $D$ as the graph gradient (difference) operator on the graph $G=(V,E)$ where $V$ denotes the vertex set of brain regions and $E$ denotes the (oriented) edge set of neighboring region pairs, such that $D\beta\in\R^m$ consists of $\beta_i - \beta_j$ for each pair of $(i, j)\in E$. In this way, each selected connection involves two regions that undergo different degrees of degeneration, and the corresponding direction points from the less atrophied region to the significantly more atrophied region.

\begin{table}[!ht]
    \caption{Selected connections by Split Knockoff on Alzheimer's Disease ($q = 0.2$). The sign -1 suggests that Region 1 suffers more severe atrophy compared with Region 2. For Split Knockoff, the optimal choice of $\nu$ selected by Split LASSO  on $\D_1$ is $\log_{10}\nu = 0.4$.}
    \centering
    \resizebox{\textwidth}{!}{
    \begin{tabular}{|c|c|cccccccc|}
        \hline
        \multicolumn{2}{|c|}{Connection} & \multicolumn{8}{c|}{Split Knockoff with $\log_{10}\nu$}\\
        \hline
        Region 1 & Region 2  & 0 & 0.2 & \textbf{0.4} & 0.6 & 0.8 & \{1.0, 1.2\} & \{1.4,1.6\} & \{1.8, 2.0\}\\
        \hline
        Hippocampus (L) & Posterior cingulate gyrus (L) &-1 & -1 & \textbf{-1} & -1 &-1 &-1 & -1& -1\\
        Hippocampus (L) & Lingual gyrus (L)  & -1&-1 &\textbf{-1} & -1 & -1 &-1 &-1 & -1\\
        Hippocampus (L) & Precuneus (L)  & -1&-1 &\textbf{-1} & -1 & -1 &-1 &-1 & -1\\
        Hippocampus (L) & Putamen (L) & & & \textbf{-1} & -1 &-1 &-1 &-1 & -1\\
        %Hippocampus (L) & Parahippocampal gyrus (L)  & & & & & & & & \\
        Hippocampus (L) & Fusiform gyrus (L)  & & &\textbf{-1} &-1 & -1 &-1 & -1& -1\\
        Hippocampus (L) & Heschl's gyrus (L) & & & &-1 & & & -1& -1\\
        Hippocampus (L) & Insula (L) & & & & -1 &-1 & &-1 & \\

        Hippocampus (L) & Thalamus (L)  & & & & -1 &-1 & & & \\
        Hippocampus (L) & Inferior temporal gyrus (L)  & & & & -1 &-1 & & & \\
        Hippocampus (R) & Lingual gyrus (R)  &-1 & -1& \textbf{-1} & -1 &-1 & & & \\
        Hippocampus (R) & Superior temporal gyrus (R)  & & & & -1 &-1 & & & \\
        Hippocampus (R) & Insula (R)  & & & & -1 & & & & \\
        Amygdala (R) & Putamen (R)  & & & \textbf{-1} & -1 &-1 & & & \\
        %Hippocampus (R) & Precuneus (R)  & & & & & & & & \\
        
        Inferior frontal gyrus, opercular part (L) & Inferior frontal gyrus, triangular part (L)  & & & & -1 & & & & \\
        Middle cingulated gyrus (R) & Medial frontal gyrus (R) & & & & -1& & & & \\
        Insula (L) & Middle frontal gyrus, orbital part (L)  & -1 & & & & & & & \\
        \hline
    \end{tabular}
    }
    \label{tab: connection}
\end{table}

The target $\dfdr$ is set to be $q = 0.2$ and the selected directed connections by Split Knockoff are presented in Table \ref{tab: connection}. For Split Knockoff, the minimal cross-validation loss happens in the case $\log_{10}\nu =0.4$, where the selected connections are marked in Figure \ref{fig: ad front}. As presented in Table \ref{tab: connection}, most of the selected connections are connected to Hippocampus, suggesting that Hippocampus suffers from the most significant degree of atrophy, in comparison with the neighboring regions. This can be supported by existing studies \citep{juottonen1999comparative} that the Hippocampus is one of the earliest and the most degenerated regions for Alzheimer's Disease. It is difficult to say whether Amygdala (R) is more atrophied than Putamen (R), as both regions have been found to suffer from degeneration \citep{de2008strongly, poulin2011amygdala}. The sign  may depend on the severity of AD patients, that according to \cite{poulin2011amygdala}, the degree of atrophy in the Amygdala is severity dependent.

\subsection{Human Age Comparisons}

\label{sec: age}

In this experiment, we apply Split Knockoffs to conduct pairwise comparisons in human ages based on voluntarily annotated data of face images. In particular, this experiments adopts $p = 30$ face images (presented in Figure \ref{fig: human faces}) from the human age dataset FG-Net (\url{http://www.fgnet.rsunit.com/}), whose respective true ages are available for evaluation. The dataset contains $n = 14011$ annotations made by volunteers on the ChinaCrowds platform \citep{xu2021evaluating}. For each annotation, one volunteer is presented with two face images, and the volunteer annotates which one looks older (or it is difficult to distinguish).

\begin{figure}
\centering
\subfigure{
\begin{minipage}{0.09\textwidth}
\centering
\includegraphics[height = 1.8cm]{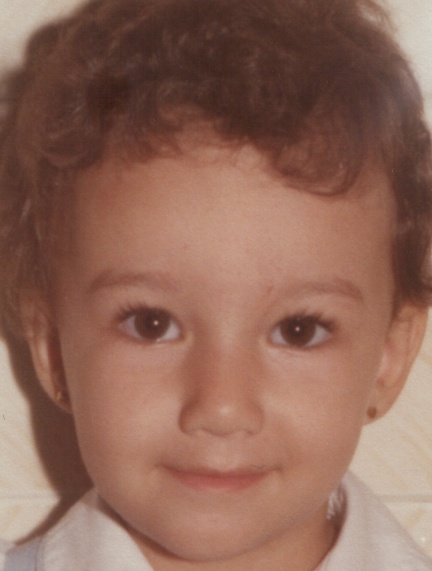}
\end{minipage}%
}%
\subfigure{
\begin{minipage}{0.09\textwidth}
\centering
\includegraphics[height = 1.8cm]{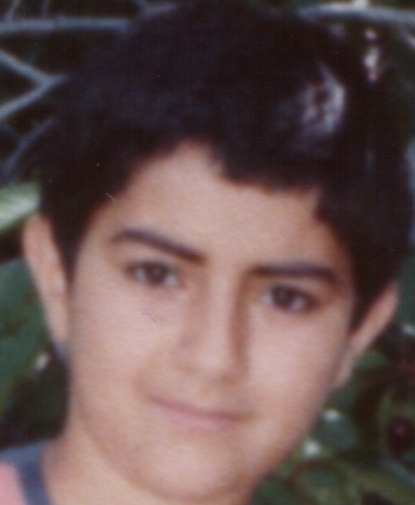}
%\caption{Power comparison in $D_1$}
\end{minipage}%
}%
\subfigure{
\begin{minipage}{0.09\textwidth}
\centering
\includegraphics[height = 1.8cm]{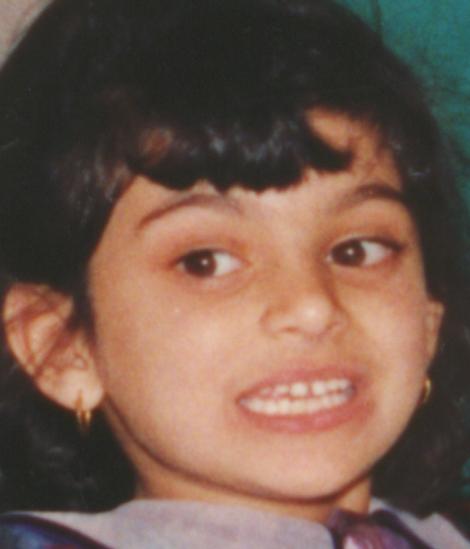}
\end{minipage}%
}%
\subfigure{
\begin{minipage}{0.09\textwidth}
\centering
\includegraphics[height = 1.8cm]{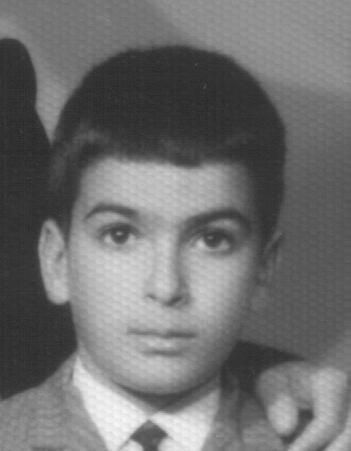}
\end{minipage}%
}%
\subfigure{
\begin{minipage}{0.09\textwidth}
\centering
\includegraphics[height = 1.8cm]{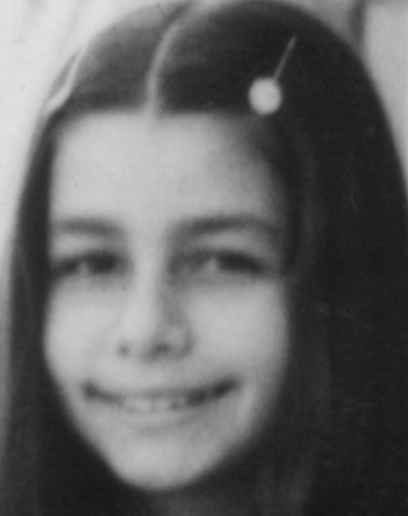}
\end{minipage}%
}%
\subfigure{
\begin{minipage}{0.09\textwidth}
\centering
\includegraphics[height = 1.8cm]{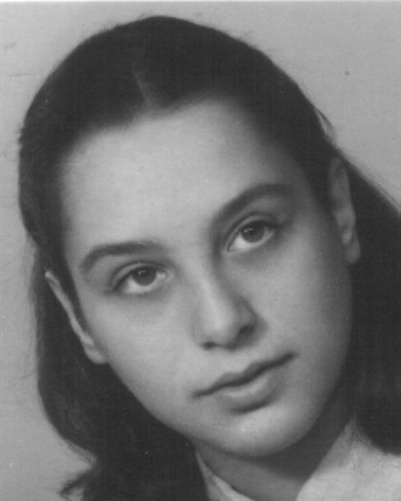}
\end{minipage}%
}%
\subfigure{
\begin{minipage}{0.09\textwidth}
\centering
\includegraphics[height = 1.8cm]{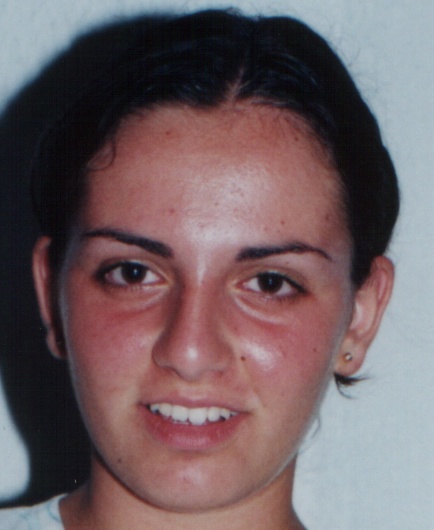}
\end{minipage}%
}%
\subfigure{
\begin{minipage}{0.09\textwidth}
\centering
\includegraphics[height = 1.8cm]{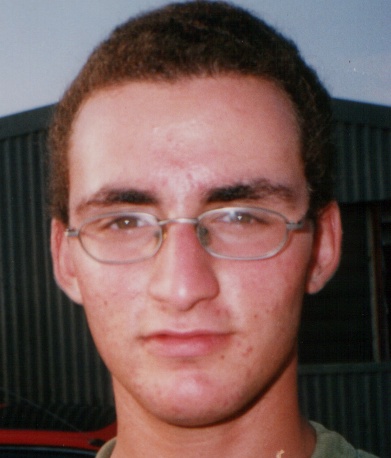}
\end{minipage}%
}%
\subfigure{
\begin{minipage}{0.09\textwidth}
\centering
\includegraphics[height = 1.8cm]{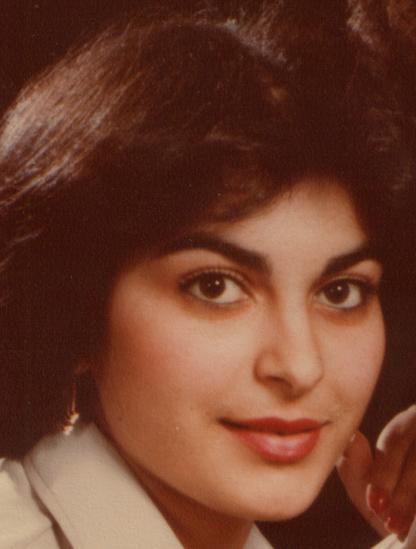}
\end{minipage}%
}%
\subfigure{
\begin{minipage}{0.09\textwidth}
\centering
\includegraphics[height = 1.8cm]{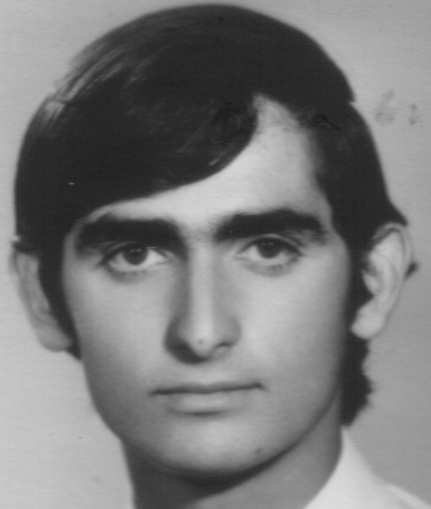}
\end{minipage}%
}%

\subfigure{
\begin{minipage}{0.09\textwidth}
\centering
\includegraphics[height = 1.8cm]{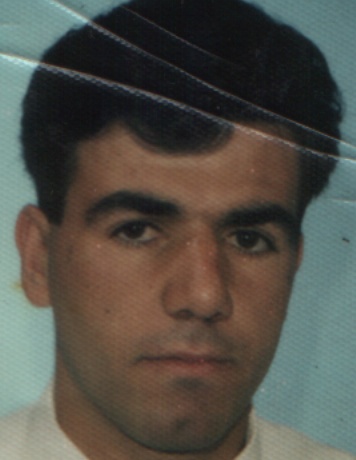}
\end{minipage}%
}%
\subfigure{
\begin{minipage}{0.09\textwidth}
\centering
\includegraphics[height = 1.8cm]{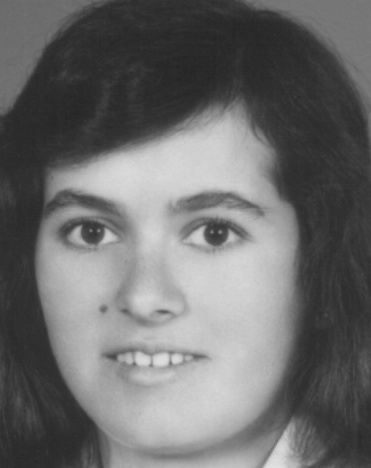}
\end{minipage}%
}%
\subfigure{
\begin{minipage}{0.09\textwidth}
\centering
\includegraphics[height = 1.8cm]{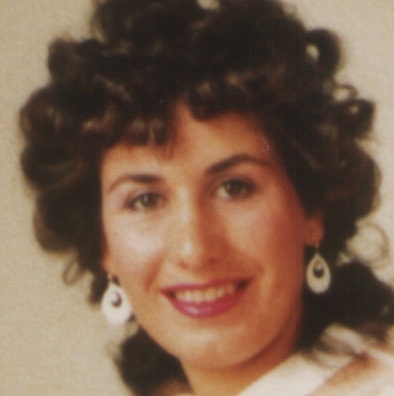}
\end{minipage}%
}%
\subfigure{
\begin{minipage}{0.09\textwidth}
\centering
\includegraphics[height = 1.8cm]{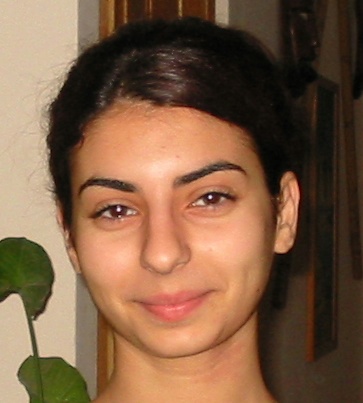}
\end{minipage}%
}%
\subfigure{
\begin{minipage}{0.09\textwidth}
\centering
\includegraphics[height = 1.8cm]{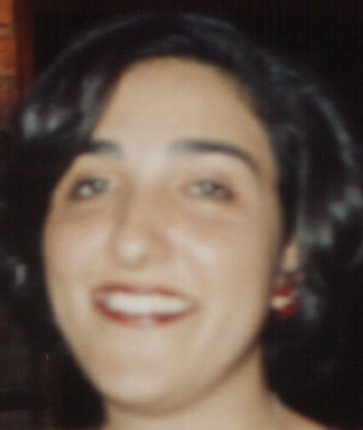}
\end{minipage}%
}%
\subfigure{
\begin{minipage}{0.09\textwidth}
\centering
\includegraphics[height = 1.8cm]{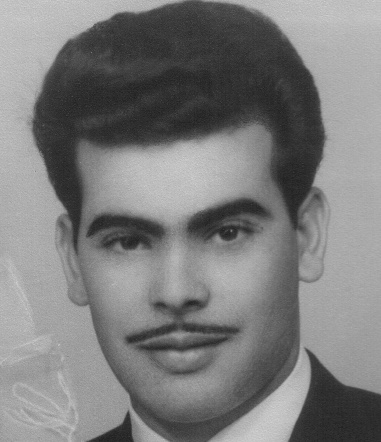}
\end{minipage}%
}%
\subfigure{
\begin{minipage}{0.09\textwidth}
\centering
\includegraphics[height = 1.8cm]{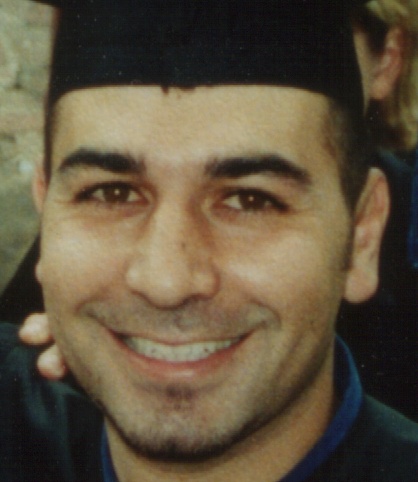}
\end{minipage}%
}%
\subfigure{
\begin{minipage}{0.09\textwidth}
\centering
\includegraphics[height = 1.8cm]{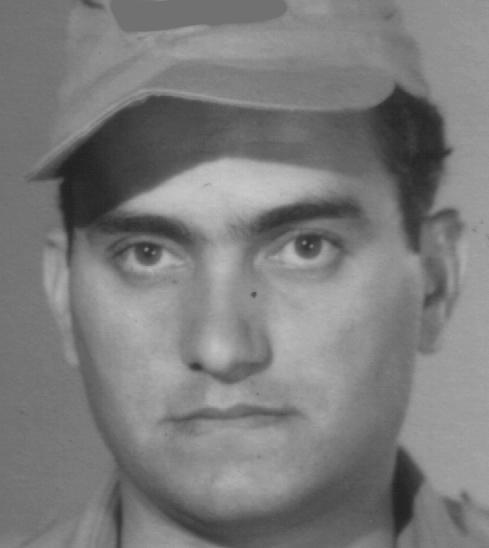}
\end{minipage}%
}%
\subfigure{
\begin{minipage}{0.09\textwidth}
\centering
\includegraphics[height = 1.8cm]{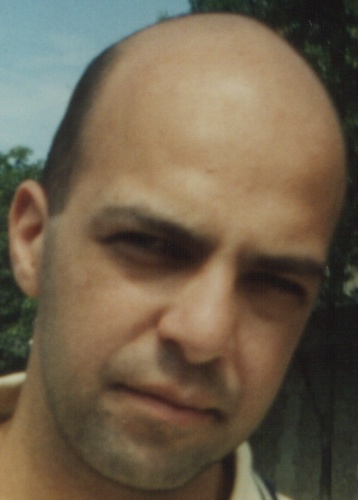}
\end{minipage}%
}%
\subfigure{
\begin{minipage}{0.09\textwidth}
\centering
\includegraphics[height = 1.8cm]{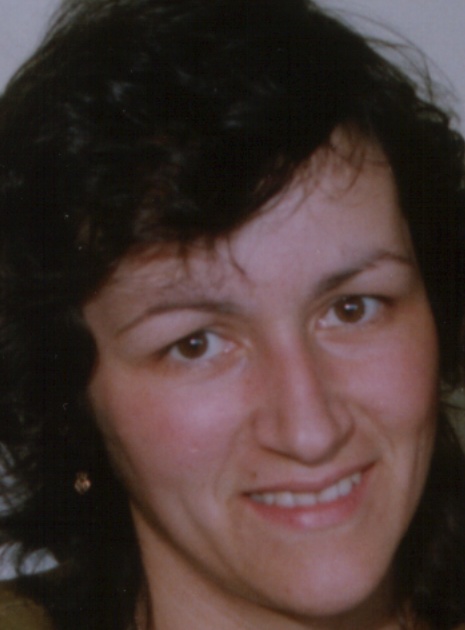}
\end{minipage}%
}%

\subfigure{
\begin{minipage}{0.09\textwidth}
\centering
\includegraphics[height = 1.8cm]{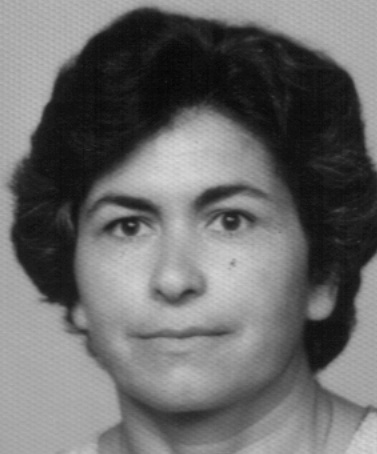}
\end{minipage}%
}%
\subfigure{
\begin{minipage}{0.09\textwidth}
\centering
\includegraphics[height = 1.8cm]{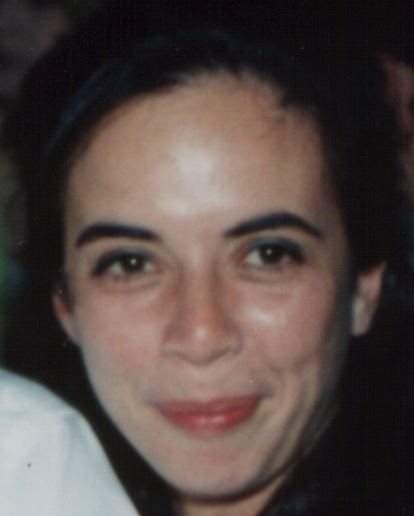}
\end{minipage}%
}%
\subfigure{
\begin{minipage}{0.09\textwidth}
\centering
\includegraphics[height = 1.8cm]{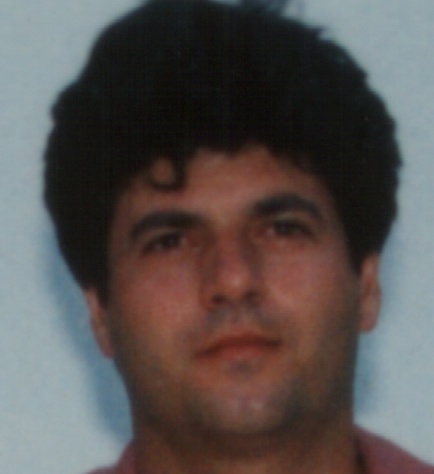}
\end{minipage}%
}%
\subfigure{
\begin{minipage}{0.09\textwidth}
\centering
\includegraphics[height = 1.8cm]{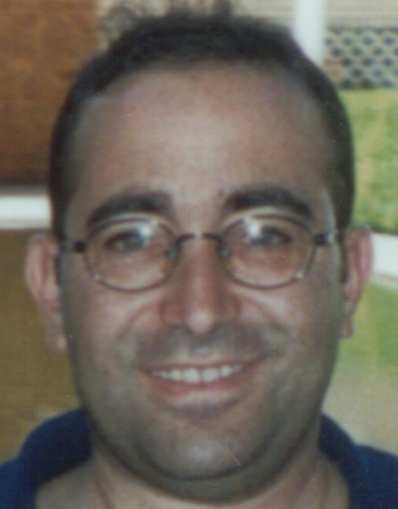}
\end{minipage}%
}%
\subfigure{
\begin{minipage}{0.09\textwidth}
\centering
\includegraphics[height = 1.8cm]{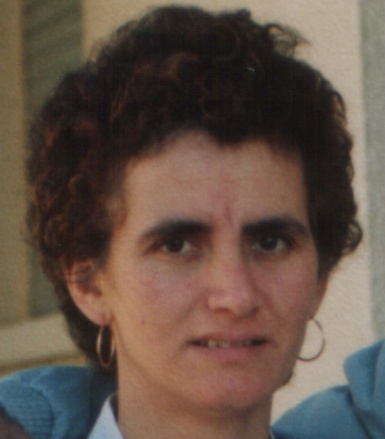}
\end{minipage}%
}%
\subfigure{
\begin{minipage}{0.09\textwidth}
\centering
\includegraphics[height = 1.8cm]{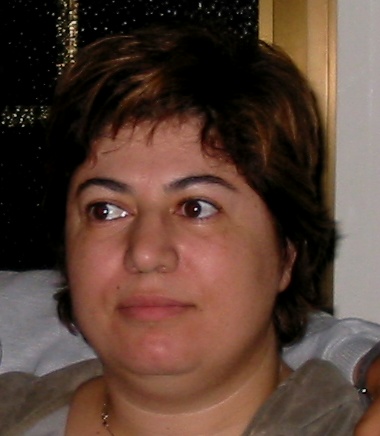}
\end{minipage}%
}%
\subfigure{
\begin{minipage}{0.09\textwidth}
\centering
\includegraphics[height = 1.8cm]{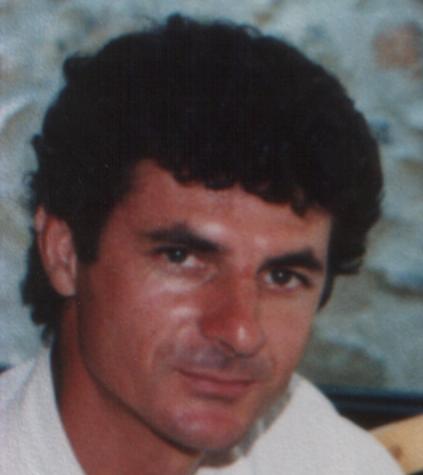}
\end{minipage}%
}%
\subfigure{
\begin{minipage}{0.09\textwidth}
\centering
\includegraphics[height = 1.8cm]{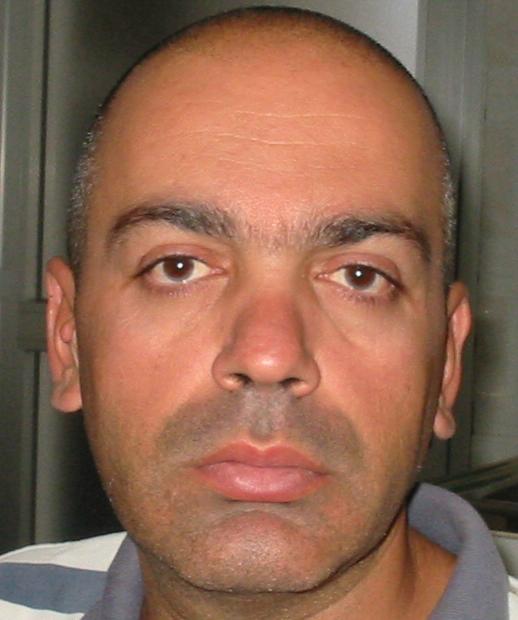}
\end{minipage}%
}%
\subfigure{
\begin{minipage}{0.09\textwidth}
\centering
\includegraphics[height = 1.8cm]{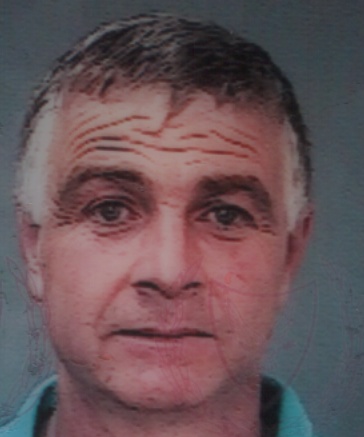}
\end{minipage}%
}%
\subfigure{
\begin{minipage}{0.09\textwidth}
\centering
\includegraphics[height = 1.8cm]{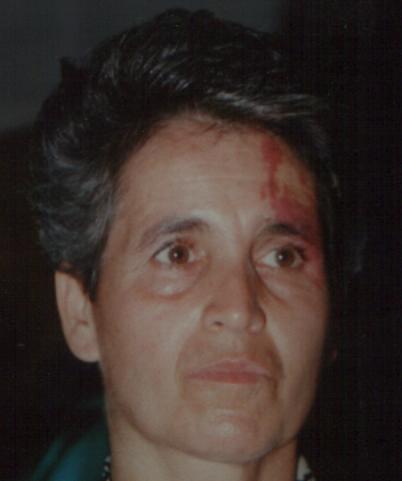}
\end{minipage}%
}%
\caption{30 images of human faces in the non-decreasing order of their respective ages.}
\label{fig: human faces}
\end{figure}

In the experiment, we construct the design matrix $X\in \R^{n\times p}$ and the response vector $y\in \R^n$ in the following way. For each $i\in \{1, 2, \cdots, n\}$, we let $X_{i, i_1}=1$, $X_{i, i_2}=-1$ and $X_{i, j}=0$ for $j\neq i_1, i_2$, if the $i$-th annotation involves images with indexes $i_1<i_2$; we let $y_{i}=\pm 1$ if the annotator thinks image $i_1$ looks older/younger respectively, and $y_{i}=0$ if the annotator feels uncertain to determine.

The linear transformation $D\in \R^{m\times p}$ is taken to be the graph difference operator on a fully connected graph $G = (V, E)$ with $|V|=p=30$ and $|E|=C_p^2 = m = 435$. In this way, we estimate the directional effects of $\gamma_i = [D\beta]_i=\beta_{i_1}-\beta_{i_2}$ for each pair of  $(i_1, i_2)$, i.e. make pairwise comparisons on the respective ages of images in Figure \ref{fig: human faces}.

The target $\dfdr$ is set to be $q = 0.2$. In Figure \ref{fig: age}, we present the performance of Split Knockoffs in $\dfdr$ and selection power for a sequence of $\log_{10}\nu$ between 0 and 4 with a step size 0.2. 
For Split Knockoffs, we random divide the full dataset $\D$ into $\D_1$ and $\D_2$ with sample sizes $n_1=1000$ and $n_2 = 13011$.

\begin{figure}[!ht]
\centering
\includegraphics[width=0.7\textwidth]{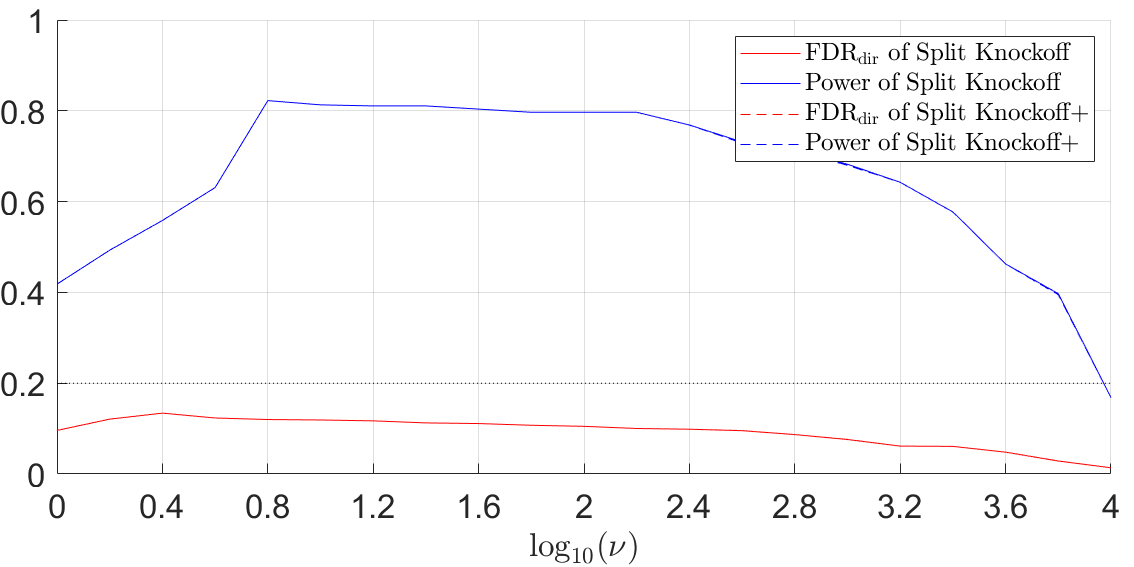}
\caption{Performance of Split Knockoffs in $\dfdr$ and power (q = 0.2) in human age comparisons, where the cross validation optimal choice of $\nu$ on $\D_1$ is $\log_{10}\nu = 2.2$.}
\label{fig: age}
\end{figure}

Figure \ref{fig: age} shows that Split Knockoffs achieve desired $\dfdr$ control for all $\nu$ in this experiment, where the $\dfdr$ of Split Knockoffs goes to zero when $\nu$ is large as predicted by Theorem \ref{theorem: directional fdr}, at the cost of losing the selection power. Meanwhile, the selection power of Split Knockoffs shows a first increase then decrease trend similarly to the simulation experiments in Section \ref{sec: simu_exp}, due to the trade-off in Split LASSO between better incoherence and the risk to lose weak nonnulls as mentioned in Section \ref{sec: simu_exp}. The cross validation optimal choice of $\nu$ selected by Split LASSO on $\D_1$ is $\log_{10}\nu = 2.2$, where relatively high selection power is achieved by Split Knockoffs. Further discussions comparing the pairwise comparisons made by Split Knockoffs with other pairwise comparison methods are provided in Section \ref{sec: supp age}.

\section{Conclusion}

In this paper, Split Knockoff is proposed as a data adaptive selection method to control the directional false discovery rate under linear transformations of regression parameters. Specifically, instead of following the linear manifold constraint, 
we relax the linear constraint to its Euclidean neighbourhood, which leads to an orthogonal design with respect to the transformed parameters. 
Combining the orthogonality with a sample splitting scheme, Split Knockoffs enjoy a desired $\dfdr$ control with a further reduction as the relaxation parameter grows. 
The simulation experiments demonstrate the effective $\dfdr$ control and good power achieved by Split Knockoffs. 
In two real world applications, for the study of Alzheimer's Disease with MRI data, Split Knockoff discovers important lesion brain regions and abnormal region connections in disease progression; for human age comparisons, Split Knockoff successfully recovers a majority of large age differences with the desired $\dfdr$ control.

\bibliography{reference}

\newpage

\begin{appendix}

\bigskip
\begin{center}
{\large\bf SUPPLEMENTARY MATERIAL}
\end{center}

\section{$\mathbf{FDR}_{\mathbf{dir}}$ versus FDR}
\label{sec: dfdr vs. fdr}

The standard false discovery rate (FDR) on $\gamma$, which evaluates the accuracy of $\widehat{S}$ (without evaluating the estimated directional effects), is defined as 
\begin{align*}
    \mathrm{FDR} = \E\left[\frac{|\{i\in S_0: i\in \widehat{S}\}|}{|\widehat{S}|\vee 1}\right].
\end{align*}
where $S_0:=\{i:\gamma^*_i=0\}$ is the null set. The following inequality shows that controlling the $\dfdr$ is stricter than controlling the FDR by showing that $\dfdr\ge \mathrm{FDR}$:
\begin{align*}
    \dfdr & = \E\left[\frac{|\{i\in \widehat{S}, \widehat{\sign}_i\neq \sign(\gamma^*_i)\}|}{|\widehat{S}|\vee 1}\right],\\
    & \ge \E\left[\frac{|\{i\in S_0: i\in \widehat{S}, \widehat{\sign}_i\neq \sign(\gamma^*_i)=0\}|}{|\widehat{S}|\vee 1}\right],\\
    & = \E\left[\frac{|\{i\in S_0: i\in \widehat{S}\}|}{|\widehat{S}|\vee 1}\right]= \mathrm{FDR},
\end{align*}
where the last step is due to the fact that for any selected feature $i\in \widehat{S}$, the estimated directional effect is always nonzero ($\widehat{\sign}_i\neq 0$).

\section{Related Work}
\label{sec: related work}

The control of the false discovery rate, including its directional counterpart, has been a subject of extensive study since the seminal work of \cite{benjamini1995controlling}. Notably, the Knockoff method, introduced by \cite{barber2015controlling}, has achieved theoretical control of the false discovery rate in sparse regression problems (specifically, the special case where $D=I_p$ in equation \eqref{eq: structural sparsity model}). Subsequently, this method has been extended to various settings, encompassing multitask regression models \citep{dai2016knockoff}, Huber's robust regression \citep{xu2016false}, high-dimensional scenarios, and even the control of the directional false discovery rate \citep{barber2019knockoff}.

An important development in this context is the Model-X Knockoff method \citep{candes2016panning}, specifically designed to address random designs, and its robustness against estimation errors of random design distributions was demonstrated by \cite{barber2020robust}. To tackle non-parametric random designs, \cite{romano2019deep} introduced the Deep Knockoff approach. Additionally, \cite{Ren21_jasa, ren2022derandomized} proposed the derandomized Knockoffs technique to generate stable selection sets, while \cite{Ren20_side} suggested leveraging side information to enhance the selection power of Knockoffs.

However, it is important to note that all of the aforementioned approaches are specifically designed for the special case where $D = I_p$ in equation \eqref{eq: structural sparsity model}. To address the more general case presented in equation \eqref{eq: structural sparsity model}, \cite{cao2021controlling} introduced the Split Knockoff method, which allows for effective control of the false discovery rate under general linear transformations. Furthermore, \cite{cao2021controlling} laid the foundation for the Split Knockoff method by proposing a family of diverse variants. In this paper, our focus is on extending a canonical version of Split Knockoffs to specifically target the control of the directional false discovery rate within application scenarios involving multiple comparisons.

To be more specific, \cite{cao2021controlling} extensively investigated various variations of Split Knockoffs and their interrelationships. These variations include the original version proposed by Barber and Cand`{e}s in \cite{barber2015controlling}, a new canonical form for Split Knockoffs, and an enhanced version of Split Knockoffs with truncation. The paper systematically explores their relationships in terms of false discovery rate (FDR) control and power. It demonstrates that all the aforementioned variations of Split Knockoffs achieve the desired FDR control while exhibiting increasing selection power relations among these variations through an inclusion property on their selectors. Among these variations, the canonical version of Split Knockoffs stands out as the most fundamental and straightforward variation.

However, it is important to note that there are critical application scenarios where the transformations involved may not exhibit sparsity. In such cases, controlling the false discovery rate may be less meaningful or reduced in significance. For instance, in scenarios involving multiple comparisons where the majority of object pairs $(\beta_i, \beta_j)$ are inherently distinct, the importance of controlling the false discovery rate for incorrectly asserting $\beta_i \neq \beta_j$ may be diminished. Nevertheless, it remains meaningful and important to control the directional false discovery rate to avoid erroneously claiming $\beta_i > \beta_j$ when $\beta_i < \beta_j$, or vice versa.

To further illustrate this point, let's consider the application of comparing human ages discussed in Section \ref{sec: age}. In this particular scenario, almost every pair of images consists of individuals of different ages. Consequently, investigating the presence of age differences between these image pairs becomes less meaningful as it is expected and obvious. However, it remains significant to control the directional FDR to accurately determine which specific individual in a pair appears older. This allows for reliable conclusions about the relative age appearance within the pair, despite the inherent differences between the individuals.

In this paper, we extend the canonical version of Split Knockoffs, which is the most straightforward variation introduced in \cite{cao2021controlling}, to address the control of directional false discovery rate (FDR) in multiple comparison scenarios. Our work encompasses several notable technical contributions, summarized as follows:
\begin{enumerate}
    \item Theorem \ref{theorem: directional fdr} establishes that the extended Split Knockoffs method ensures directional FDR control.
    \item Building upon Theorem \ref{theorem: directional fdr}, we demonstrate that the directional FDR of the extended Split Knockoffs method diminishes to zero as $\nu$ increases. This finding highlights the phenomenon of diminishing directional FDR, as depicted in Figure \ref{fig: comparison}.
    \item In Section \ref{sec: inflation}, we introduce Theorem \ref{theorem: directional fdr without split}, which emphasizes the necessity of the sample splitting scheme in Split Knockoffs. Without the sample splitting scheme, the Split Knockoff method may lead to an inflation issue in directional FDR control.
\end{enumerate}

In summary, this paper serves as a follow-up to \cite{cao2021controlling}, extending the Split Knockoff method to control the directional false discovery rate in multiple comparison scenarios with a more deliberate analysis.

\section{Knockoffs with Generalized LASSO: Antisymmetry Broken}

\label{sec: genlasso with knockoff}

In this section, we will show that on the problem \eqref{eq: structural sparsity model}, the naive construction of Knockoffs \citep{barber2015controlling,barber2019knockoff}  by ignoring the structural constraint fails the antisymmetry, 
except for a special case when $D$ has full row rank, i.e. $\ker(D^T)=0$.

The canonical statistical method for problem \eqref{eq: structural sparsity model} is the generalized LASSO \citep{tibshirani2011solution}. The generalized LASSO regularization path with respect to the regularization parameter $\lambda>0$ is given by
\begin{align}
    \label{eq: genlasso below}
    (\beta(\lambda), \gamma(\lambda)) :=\frac{1}{2n}\argmin_{\beta, \gamma}\|y-X\beta\|_2^2+\lambda\|\gamma\|_1,\ \mathrm{subject\ to}\ \gamma=D\beta.
\end{align}
Clearly, $X$ is not a proper design matrix for $\gamma$ in constructing Knockoff copies. Therefore, 
Equation \eqref{eq: genlasso below} need to be reformulated for creating the Knockoff copy. 

The naive way of applying Knockoffs to Equation \eqref{eq: genlasso below} is to solve $\beta$ from $\gamma$ through the constraint $\gamma=D\beta$. Suppose that $\beta$ can be solve from $\gamma$, that $\beta = D^\dagger\gamma$ for some $D^\dagger\in \R^{p\times m}$ (e.g. the pseudo inverse of $D$). In this case, 
Equation \eqref{eq: genlasso below} becomes
\begin{align}
    \label{eq: genlasso 2}
    \gamma(\lambda) :=\argmin_{\gamma}\frac{1}{2n}\|y-XD^\dagger \gamma\|_2^2+\lambda\|\gamma\|_1,\ \mathrm{subject\ to}\ \gamma\in \im(D).
\end{align}
Note that $\gamma \in \im(D) \Leftrightarrow \gamma \perp \ker(D^T)$.
Let $A$ be the matrix whose rows span $\ker(D^T)$, and $X_\gamma = XD^\dagger$ be the design matrix for $\gamma$ in Equation \eqref{eq: genlasso 2}. Then Equation \eqref{eq: genlasso 2} becomes the following constrained LASSO problem
\begin{align}
    \label{eq: genlasso 3}
    \gamma(\lambda) :=\argmin_{\gamma}\frac{1}{2n}\|y-X_\gamma \gamma\|_2^2+\lambda\|\gamma\|_1,\ \mathrm{subject\ to}\ A\gamma = 0.
\end{align}
We will show in the following two subsections that:
\begin{enumerate}
    \item[(a)] For the general case that $\ker(D^T)\neq 0$ ($A\neq 0$), the constraint that $A\gamma = 0$ will destroy the antisymmetry property \citep{barber2015controlling} if Knockoffs are constructed naively, ignoring the constraint $A\gamma = 0$ in Equation \eqref{eq: genlasso 3}.
    \item[(b)] For the special case that $\ker(D^T)= 0$, i.e. $D\in \R^{m\times p}$ satisfies $m\le p$ and $D$ has full row rank, the standard knockoffs are indeed applicable.
\end{enumerate}

\subsection{A Counter Example for General Case: $\ker(D^T)\neq 0$}

\label{sec: counter example}

In this section we present a counter example where the antisymmetry property of Knockoffs fails  for the case $\ker(D^T)\neq 0$. 

The Knockoff method \citep{barber2015controlling, barber2019knockoff} constructs the Knockoff copy matrix $\tilde{X}_\gamma\in \R^{n\times p}$ with respect to $X_\gamma$ in the way such that
\begin{align}
    X_\gamma^T X_\gamma = \tilde{X}_\gamma^T\tilde{X}_\gamma,\ \tilde{X}_\gamma^T X_\gamma = X_\gamma^T X_\gamma - \diag(\vecs), \label{eq: knockoff copy}
\end{align}
for some proper non-negative vector $\vecs\in \R^{p}$. After constructing the Knockoff copy matrix, the following optimization problem is commonly used in Knockoffs to determine the feature significance and $W$ statistics in \cite{barber2015controlling, barber2019knockoff},
\begin{align}
    \hat{\beta}(\lambda) = \argmin_{\beta\in \R^{2p}}\frac{1}{2n}\|y - [X_\gamma, \tilde{X}_\gamma]\beta\|_2^2+\lambda\|\beta\|_1,\ \mathrm{s.t.}\ A\beta_1=A\beta_2=0,\ \ \beta = \left[\begin{array}{c}
\beta_1\\
\beta_2
\end{array}
\right],\label{eq:classo knock}
\end{align}
 where the constraint $A\beta_1=A\beta_2=0$ succeeds from Equation \eqref{eq: genlasso 3}. One common way to construct the $W$ statistics is to record the points $\lambda$ in Equation \eqref{eq:classo knock} where the feature $X_{\gamma,i}$ or the copy $\tilde{X}_{\gamma, i}$ first enters the model as $Z_i$ or $\tilde{Z}_i$ for all $i$, i.e.
\begin{align*}
    Z_i = & \sup\{\lambda:\ \hat{\beta_1}_i(\lambda)\neq 0\},\\
    \tilde{Z}_i = & \sup\{\lambda:\ \hat{\beta_2}_i(\lambda)\neq 0\},
\end{align*}
and define $W = \max(Z, \tilde{Z})\odot\sign(Z-\tilde{Z})$.

However, we will show by the following simple example that such a treatment dissatisfies the antisymmetry property \citep{barber2015controlling}. This consequently fails the provable FDR control for Knockoffs. 

\paragraph{} Let $A = [1, 1]$, $X_\gamma=\left[
\begin{array}{c c}
1 & 0\\
0 & 1\\
0 & 0\\
0 & 0
\end{array}
\right]$, $\beta^*=\left[
\begin{array}{c}
1\\
-1
\end{array}
\right]$, and $y = \left[
\begin{array}{c c}
1\\
-1\\
0\\
0
\end{array}
\right]$. Then $\left[
\begin{array}{c c}
0 & 0\\
0 & 0\\
1 & 0\\
0 & 1
\end{array}
\right]:=\tilde{X}_\gamma$ is a valid Knockoff copy matrix satisfying $X^T_\gamma X_\gamma=I=\tilde{X}_\gamma^T \tilde{X}_\gamma$, $\tilde{X}_\gamma^T X_\gamma =0$ by taking $\vecs^T=[1,1]$ in Equation \eqref{eq: knockoff copy}.

\paragraph{} In this example, it will be shown below that without swapping the columns of $X_\gamma$ and $\tilde{X}_\gamma$, the Knockoff statistics are $Z_1=Z_2=\frac{1}{4}$, $\tilde{Z}_1=\tilde{Z}_2=0$, and $W_1=W_2=\frac{1}{4}$. 

To see this point, we check the solution of the constrained LASSO problem \eqref{eq:classo knock}. %Here, the design matrix is $[X_\gamma, \tilde{X}_\gamma] = I_4$. 
Due to the constraint $A\beta_1=A\beta_2=0$, we re-parameterize $\beta$ as $\beta=\left[
\begin{array}{c c}
a\\
-a\\
b\\
-b
\end{array}
\right]$. Then, the problem \eqref{eq:classo knock} can be written as
\begin{align*}
    & (\hat{a}(\lambda), \hat{b}(\lambda)) = \argmin_{a, b\in \R}   \frac{1}{8}(a-1)^2+\frac{1}{8}(-a+1)^2+\frac{1}{8}b^2+\frac{1}{8}b^2+2\lambda|a|+2\lambda|b|, \\
\Leftrightarrow & (\hat{a}(\lambda), \hat{b}(\lambda)) = \argmin_{a, b\in \R} \frac{1}{4}(a^2+1)+(2\lambda|a|-\frac{1}{2}a)+\frac{1}{4}b^2+2\lambda|b|, \\
    \Leftrightarrow & \hat{a}(\lambda) = \argmin_{a \in \R} \frac{1}{4}(a^2+1)+(2\lambda|a|-\frac{1}{2}a) \mbox{ and }  \hat{b}(\lambda) = \argmin_{b\in \R}\frac{1}{4}b^2+2\lambda|b|.
\end{align*}

Clearly, the optimizer $\hat{a}(\lambda)$ can be nonzero if and only if $2\lambda<\frac{1}{2}$, therefore, $Z_1=Z_2=\frac{1}{4}$. Meanwhile, for any $\lambda>0$, the optimizer $\hat{b}(\lambda)$is always zero, which means $\tilde{Z}_1=\tilde{Z}_2=0$. Then by definition, $W_1=W_2=\frac{1}{4}$.

\paragraph{} However, if the first column of $X_\gamma$ and $\tilde{X}_\gamma$ is swapped, we show in the following that the Knockoff statistics become $Z_1=Z_2=\frac{1}{8}$, $\tilde{Z}_1=\tilde{Z}_2=\frac{1}{8}$, and $W_1=W_2=0$.

To see this point, consider the constrained LASSO problem \eqref{eq:classo knock} after swapping the column. We re-parameterize $\beta$ by $(a, b)$ in the same way as above, then the
\iffalse
Similarly, from the linear constraint, we can re-parameterize $\beta$ as $\beta=\left[
\begin{array}{c c}
a\\
-a\\
b\\
-b
\end{array}
\right]$. After swapping the column, the design matrix becomes $[X_\gamma, \tilde{X}_\gamma]_{\mathrm{swap}} = \left[
\begin{array}{cccc}
0 & 0 & 1 & 0\\
0 & 1 & 0 & 0\\
1 & 0 & 0 & 0\\
0 & 0 & 0 & 1
\end{array}
\right]$. Then
\fi
problem \eqref{eq:classo knock} after swapping the first column of $X_\gamma$ and $\tilde{X}_\gamma$ becomes
\begin{align*}
    & (\hat{a}(\lambda), \hat{b}(\lambda)) =  \argmin_{a, b\in \R}  \frac{1}{8}(b-1)^2+\frac{1}{8}(-a+1)^2+\frac{1}{8}a^2+\frac{1}{8}b^2+2\lambda|a|+2\lambda|b|, \\
    \Leftrightarrow & (\hat{a}(\lambda), \hat{b}(\lambda)) = \argmin_{a, b\in \R}  \frac{1}{8}(2a^2+1)+(2\lambda|a|-\frac{1}{4}a)+\frac{1}{8}(2b^2+1)+(2\lambda|b|-\frac{1}{4}b), \\
    \Leftrightarrow &  \hat{a}(\lambda) = \argmin_{a \in \R} \frac{1}{8}(2a^2+1)+(2\lambda|a|-\frac{1}{4}a) \mbox{ and } \hat{b}(\lambda) = \argmin_{b\in \R}  \frac{1}{8}(2b^2+1)+(2\lambda|b|-\frac{1}{4}b).
\end{align*}
Clearly, the optimizer $\hat{a}(\lambda)$ and $\hat{b}(\lambda)$ can be nonzero if and only if $2\lambda<\frac{1}{4}$, therefore, $Z_1=Z_2=\tilde{Z}_1=\tilde{Z}_2=\frac{1}{8}$. Then by definition, $W_1=W_2=0$.

\paragraph{} The computation above shows that swapping the first column of $X_\gamma$ and $\tilde{X}_\gamma$ in this example changes the value of the $W$ statistics for the second variable due to the linear constraint. This violates the antisymmetry property, which states that swapping the column will only lead to an opposite sign of $W$ statistics for its respective variable while keeping invariant the $W$ statistics for other variables. 

\subsection{Special Case: $\ker(D^T)=0$}

\label{sec: special cases}

In the special case that $\ker(D^T)=0$ ($\rank(D) = m \le p$), the problem \eqref{eq: structural sparsity model} can be reduced to the standard sparse linear regression problem where the classical Knockoff method is applicable. In this case, we write
\begin{align*}
    \beta^* = D^\dagger\gamma^*+\beta_0,
\end{align*}
for some $\beta_0\in\R^p$ lies in the null space of $D$. Then the problem \eqref{eq: structural sparsity model} can be written as
\begin{align}
\label{eq: tran1}
    y = XD^\dagger\gamma^* + X\beta_0+\varepsilon.
\end{align}
Let $D_0\in \R^{p\times (p-m)}$ be the matrix whose columns span the null space of $D$, and take $U$ to be the orthogonal complement of $X D_0$, multiple $U$ on both sides of Equation \eqref{eq: tran1}, there holds
\begin{align*}
    Uy = UXD^\dagger\gamma^*+U\varepsilon.
\end{align*}
Viewing $Uy$ as the response vector, $UXD^\dagger$ as the design matrix for $\gamma$, and $U\varepsilon$ as the i.i.d. Gaussian noise, standard Knockoffs are applicable, at the cost of a worse incoherence condition brought by $UXD^\dagger$ when $D$ is not the identity matrix. This may cause a loss in the selection power. 

\section{Construction of Split Knockoff Copy Matrix}

\label{sec: construct split knockoff}

In this section, we will show the details on the construction of the Split Knockoff copy. For shorthand notations, define $\Sigma_{\beta,\beta}:=A_{\beta_2}^TA_{\beta_2}$, $\Sigma_{\beta,\gamma}=\Sigma_{\gamma,\beta}^T:=A_{\beta_2}^TA_{\gamma_2}$, and $\Sigma_{\gamma,\gamma}:=A_{\gamma_2}^TA_{\gamma_2}$. The necessary and sufficient condition for the existence of $\tilde{A}_{\gamma_2}$ satisfying Equation \eqref{eq: copy} is
\begin{align*}
    G: = 
    \begin{bmatrix}
        \Sigma_{\beta,\beta} & \Sigma_{\beta,\gamma} & \Sigma_{\beta,\gamma}\\
        \Sigma_{\beta,\gamma} & \Sigma_{\gamma,\gamma} & \Sigma_{\gamma,\gamma}-\diag(\vecs)\\
        \Sigma_{\beta,\gamma} & \Sigma_{\gamma,\gamma}-\diag(\vecs) & \Sigma_{\gamma,\gamma}
    \end{bmatrix}
    \succeq 0.
\end{align*}
This holds if and only if the Schur complement of $\Sigma_{\beta,\beta}$ is positive semi-definite, i.e.
\begin{align}
    \begin{bmatrix}
        C_\nu & C_\nu-\diag(\vecs)\\
        C_\nu-\diag(\vecs) & C_\nu
    \end{bmatrix}
    \succeq 0, \mbox{ where }C_\nu :=\Sigma_{\gamma,\gamma}-\Sigma_{\gamma,\beta}\Sigma_{\beta,\beta}^{-1}\Sigma_{\beta,\gamma},\nonumber
\end{align}
which holds if and only if $C_\nu$ and its Schur complement are positive semi-definite, i.e.
\begin{subequations}
\begin{align}
    &C_\nu  \succeq 0,\nonumber\\
    &C_\nu-(C_\nu-\diag(\vecs))C_\nu^{-1}(C_\nu-\diag(\vecs))=2\diag(\vecs)-\diag(\vecs)C_\nu^{-1}\diag(\vecs)
    \succeq 0.\nonumber
\end{align}
\end{subequations}
Therefore the non-negative vector $\vecs\in\R^m$ should satisfy
\begin{align}
    \label{eq: vecs}
    \diag(\vecs)  \succeq 0,\ 2C_\nu-\diag(\vecs) \succeq 0.
\end{align}

There are various choices for $\vecs=(\vecs_i)_{i=1}^m$ satisfying \eqref{eq: vecs} for the construction of Split Knockoff copy. Below we give two typical examples. 

\begin{itemize}
\item[(a)] (Equi-correlation) Take 
\begin{align}
    \vecs_i = 2\lambda_{\mathrm{min}}(C_\nu)\land \frac{1}{\nu}, \label{eq: equi corre}
\end{align}
for all $i\in \{1, 2, \cdots, m\}$. This is the default choice in this paper.
    \item[(b)] (SDP for discrepancy maximization) One can maximize the discrepancy between Knockoffs and its corresponding features by solving the following SDP
\begin{eqnarray*} 
\mbox{maximize} & & \sum_i \vecs_i , \\
\mbox{subject to} & & 0\leq \vecs_i \leq \frac{1}{\nu} \mbox{ and } \diag(s) \preceq  2C_\nu. 
\end{eqnarray*}

\end{itemize}

For each suitable vector $\vecs$, the construction of the Split Knockoff copy is given by
\begin{align}
    \tilde{A}_{\gamma_2}=A_{\gamma_2}(I_{m}-C_\nu^{-1}\diag(\vecs))+A_{\beta_2} \Sigma_{\beta,\beta}^{-1}\Sigma_{\beta,\gamma}C_{\nu}^{-1}\diag(\vecs)+\tilde{U}K,
\end{align}
where $\tilde{U}\in \R^{(n_2+m)\times m}$ is the orthogonal complement of $[A_{\beta_2}, A_{\gamma_2}]\in\R^{(n_2+m)\times(m+p)}$ (which requires $n_2+m\ge m+m+p$, i.e. $n_2\ge m+p$) and $K\in \R^{m\times m}$ satisfies $K^TK=2\diag(\vecs)-\diag(\vecs)C_\nu^{-1}\diag(\vecs)$.

\section{Failure of Exchangeability in Split Knockoffs}

\label{sec: fail exchange}

In this section, we show that the exchangeability property no longer holds for Split Knockoffs. In particular, we show that the pairwise exchangeability on the response fails for Split Knockoffs through the following proposition.

\begin{proposition}[Failure of Pairwise Exchangeability on the response]
    \label{prop: exchange}
    For any $i$,
    \begin{align*}
        [A_{\gamma_2}, \tilde{A}_{\gamma_2}]_{\mathrm{swap}\{i\}}^T\tilde y_2\overset{d}{\neq} &  [A_{\gamma_2}, \tilde{A}_{\gamma_2}]^T\tilde y_2,
    \end{align*}
where $[A_{\gamma_2}, \tilde{A}_{\gamma_2}]_{\mathrm{swap}\{i\}}$ denotes a swap of the $i$-th column of $A_{\gamma_2}$ and $\tilde{A}_{\gamma_2}$ in $[A_{\gamma_2}, \tilde{A}_{\gamma_2}]$.
\end{proposition}

\begin{proof}
    Following the definition of $A_{\gamma_2}$ and $\tilde y_2$ in Equation \eqref{eq: new design matrix}, $A_{\gamma_2}^T\tilde y_2 = 0_m$. By Lemma \ref{lemma: distribution zeta} on the distribution of $\zeta = \tilde{A}_{\gamma_2}^T\tilde y_2$, there further holds
    \begin{align*}
    [A_{\gamma_2}, \tilde{A}_{\gamma_2}]^T\tilde y_2\sim\mathcal{N}\left(
    \begin{bmatrix}
        0_m\\
        -\diag(\vecs)\gamma^*
    \end{bmatrix},
    \begin{bmatrix}
        0_m & 0_m\\
        0_m & \frac{1}{n_2}\diag(\vecs)(2I_m-\diag(\vecs)\nu)\sigma^2
    \end{bmatrix}
    \right),
    \end{align*}
    where swapping any $i$ in the first and second block leads to different distributions.
\end{proof}

The failure of exchangeability is a great hurdle for the beautiful supermartingale structure \citep{barber2015controlling, barber2019knockoff} which is crucial for the provable $\dfdr$ control in Knockoff-based methods. Fortunately, with the orthogonal design deducted from the variable splitting scheme and a sample splitting scheme presented in Section \ref{sec: methodology}, we overcome such a hurdle and achieve desired $\dfdr$ control in Theorem \ref{theorem: directional fdr}.

\section{Split Knockoffs in High Dimensional Settings}

\label{sec: screen for hd}

In this section, we introduce the way to perform Split Knockoffs in high dimensional settings where $n_2< m+p$ or even $n_2<p$. In short, we first screen for subsets $\hat{S}_\beta$, $\hat{S}_\gamma$ of the features in $\beta$, $\gamma$ sequentially on $\D_1$, and then perform Split Knockoffs on the selected subset of features.

To screen for a subset $\hat{S}_\beta$ of features in $\beta$, we conduct LASSO on $\D_1 = (X_1, y_1)$ with respect to suitable $\lambda_\beta>0$ (e.g. selected by cross validation),
\begin{align*}
    \beta(\lambda_\beta) = \argmin_{\beta}\frac{1}{2n_1}\|y_1-X_1\beta\|_2^2+\lambda_\beta\|\beta\|_1,
\end{align*}
and let $\hat{S}_\beta:=\{i: \beta(\lambda_\beta)_i \neq 0\}$. Then we proceed to screen for a subset  $\hat{S}_\gamma$ of features in $\gamma$. Let $X_{1, \hat{S}_\beta}$, $D_{\hat{S}_\beta}$ be the submatrix of $X_1$, $D$ respectively containing the columns indicated by $\hat{S}_\beta$, we conduct Split LASSO on $\D_1$ with respect to suitable $\lambda_\gamma>0$ (e.g. selected by cross validation),
\begin{align*}
    \gamma(\lambda_\gamma) = \argmin_{\gamma}\min_\beta\frac{1}{2n_1}\|y_1-X_{1, \hat{S}_\beta}\beta\|_2^2+\frac{1}{2\nu}\|D_{\hat{S}_\beta}\beta - \gamma\|_2^2+\lambda_\gamma\|\gamma\|_1,
\end{align*}
and let $\hat{S}_\gamma:=\{i: \gamma(\lambda_\gamma)_i \neq 0\}$.

Let $X_{2, \hat{S}_\beta}$ be the submatrix of $X_2$ containing the columns indicated by $\hat{S}_\beta$, and $D_{\hat{S}_\beta, \hat{S}_\gamma}$ be the submatrix of $D$ containing the columns indicated by $\hat{S}_\beta$ and rows indicated by $\hat{S}_\gamma$. When the conditions $X_{2, \hat{S}_\beta}^TX_{2, \hat{S}_\beta}$ is invertible and $n_2\ge |\hat{S}_\beta|+ |\hat{S}_\gamma|$ are satisfied, we can now conduct Split Knockoffs with respect to $\D_1' = (X_{1, \hat{S}_\beta}, y_1)$, $\D_2' = (X_{1, \hat{S}_\beta}, y_2)$ and the linear transformation $D_{\hat{S}_\beta, \hat{S}_\gamma}$ as in Section \ref{sec: methodology}. Proposition \ref{prop: hd} shows that the above procedure will not lose the $\dfdr$ control if $\{i:\beta^*_i\neq 0\}\subseteq\hat{S}_\beta$, which is known as the sure screening event \citep{fan2008sure} in literature.

\begin{proposition}
    \label{prop: hd}
    Let $\Gamma$ be the event that $\{i:\beta^*_i\neq 0\}\subseteq\hat{S}_\beta$, then for Split Knockoffs with the above feature screening procedure, there holds
    \begin{itemize}
        \item[(a)] ($\mdfdr$ of Split Knockoff)
    \begin{align*}
        \E\left[\left.\frac{|\{i\in \widehat{S}: \widehat{\sign}_i \neq \sign(\gamma^*_i)\}|}{|\widehat{S}|+q^{-1}}\right|\Gamma\right]\le \min(\alpha(\nu), 1) q,
    \end{align*}
    \item[(b)] ($\dfdr$ of Split Knockoff+)
    \begin{align*}
        \E\left[\left.\frac{|\{i\in \widehat{S}: \widehat{\sign}_i \neq \sign(\gamma^*_i)\}|}{|\widehat{S}|\vee 1}\right|\Gamma\right]\le \min(\alpha(\nu), 1) q,
    \end{align*}
    \end{itemize}
    where $\alpha(\nu)$ is defined in Theorem \ref{theorem: directional fdr}.
\end{proposition}
\begin{proof}
In the case that the event $\Gamma$ occurs, there holds
\begin{align}
    y = X_{\hat{S}_\beta}\beta^*_{\hat{S}_\beta} + \varepsilon, \gamma^*_{\hat{S}_\gamma} = D_{\hat{S}_\beta, \hat{S}_\gamma}\beta^*_{\hat{S}_\beta},\label{eq: restricted model}
\end{align}
where $X_{\hat{S}_\beta}$ is the submatrix of $X$ containing the columns indicated by $\hat{S}_\beta$, and $\beta^*_{\hat{S}_\beta}$, $\gamma^*_{\hat{S}_\gamma}$ are subvectors of $\beta^*$, $\gamma^*$ containing the elements indicated by $\hat{S}_\beta$, $\hat{S}_\gamma$ respectively. Therefore, Equation \eqref{eq: restricted model} is a reduced model of Equation \eqref{eq: structural sparsity model} restricted on $\hat{S}_\beta$, $\hat{S}_\gamma$, and the same procedure as in Theorem \ref{theorem: directional fdr} can be applied to achieve Proposition \ref{prop: hd}.
\end{proof}

Proposition \ref{prop: hd} suggests that for the $\dfdr$ control, we need to be conservative when screening off features in $\beta$, while it is fine to be aggressive when screening off features in $\gamma$ to make $n_2\ge |\hat{S}_\beta|+ |\hat{S}_\gamma|$. We conduct simulation experiments to validate the effectiveness of our procedure in high dimensional settings. In particular, we succeed all the settings in Section \ref{sec: simu_exp} except for taking $p = 1000>n =500$.

\begin{figure}[!ht]
\centering
\subfigure[Performance in $D_1$]{
\begin{minipage}[t]{0.33\textwidth}
\centering
\includegraphics[width=\textwidth]{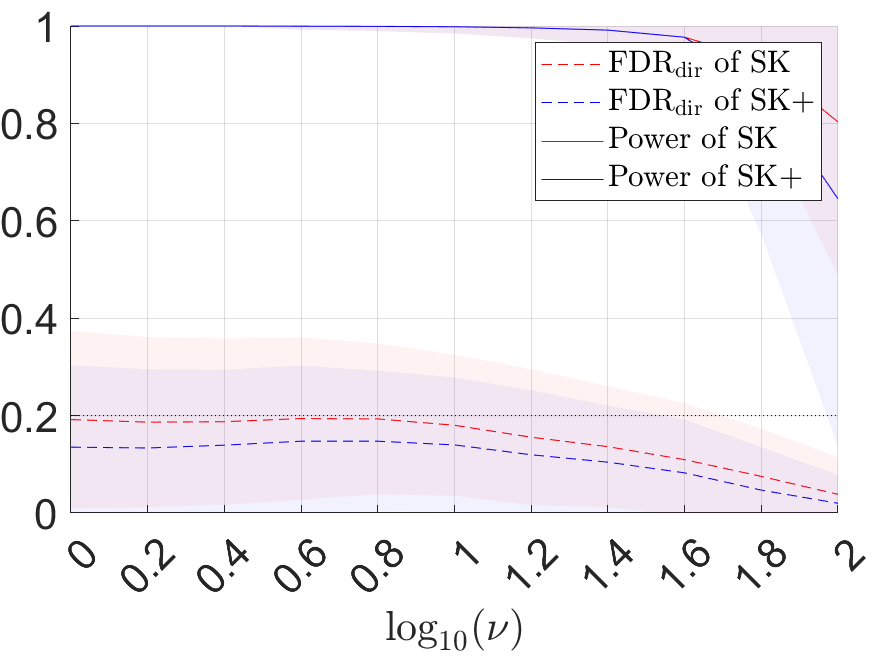}
%\caption{FDR comparison in $D_1$}
\end{minipage}%
}%
\subfigure[Performance in $D_2$]{
\begin{minipage}[t]{0.33\textwidth}
\centering
\includegraphics[width=\textwidth]{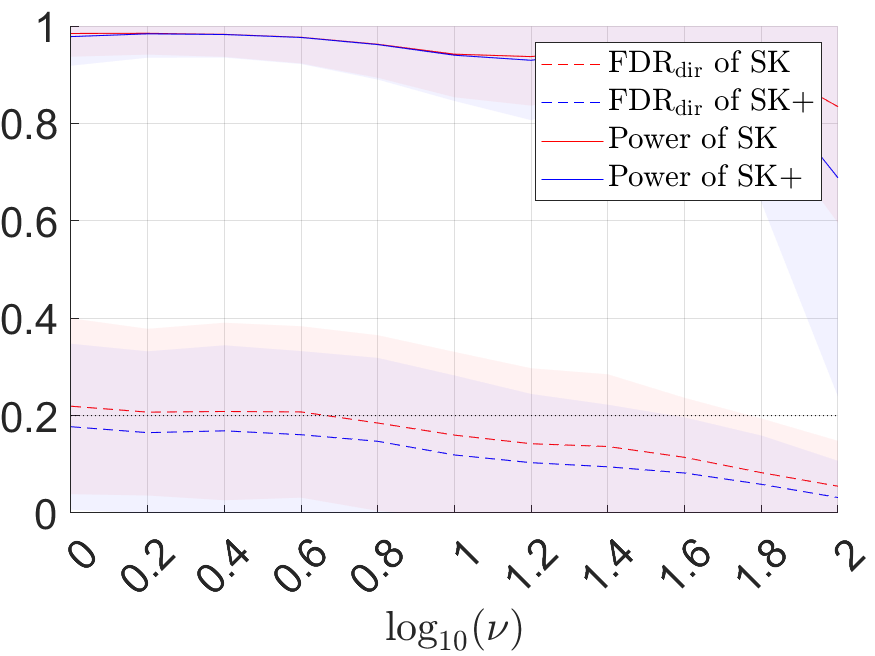}
%\caption{FDR comparison in $D_2$}
\end{minipage}%
}%
\subfigure[Performance in $D_3$]{
\begin{minipage}[t]{0.33\textwidth}
\centering
\includegraphics[width=\textwidth]{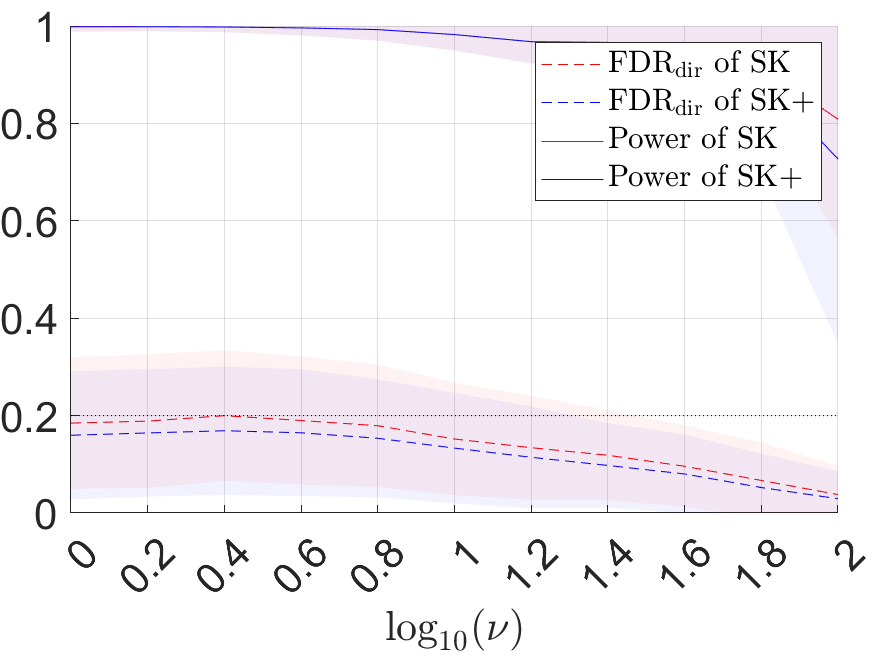}
%\caption{FDR performance in $D_3$}
\end{minipage}%
}%

\caption{Performance of Split Knockoffs in high dimensional settings: Power, $\dfdr$ for $q=0.2$. The curves in the figures represent the average performance in 200 simulation instances, while the shaded areas represent the 80\% confidence intervals truncated to the range $[0, 1]$. We use ``SK(+)'' to represent ``Split Knockoff(+)'' for shorthand notations.}
\label{fig: hd simulation}
\end{figure}

Figure \ref{fig: hd simulation} shows that our proposed procedure in high dimensional settings achieve desired performance in both $\dfdr$ control and selection power in the simulation experiments. Moreover, the $\dfdr$ of Split Knockoffs still goes to zero when $\nu$ is large, while good selection power is achieved for a wide range of $\nu$.

\section{Power Improvement of the Modified $W$ Statistics}

\label{sec: inclusion}

In this section, we show that our definition of $W$ statistics \eqref{eq: def_w}, a slightly modified version compared with that in \cite{barber2015controlling, barber2019knockoff}, improves the selection power of Split Knockoffs through an inclusion property between the selection set given by the $W$ statistics defined in \eqref{eq: def_w} and that given by its original definition in \cite{barber2015controlling, barber2019knockoff}, denoted as $W' := \max(Z, \tilde{Z})\odot\sign(Z-\tilde{Z})$ in the following.

For $W'$, \cite{barber2015controlling, barber2019knockoff} define the data dependent threshold $T_q'$ based on a pre-set nominal $\dfdr$ level $q$ as 
\begin{align*}
  &\mbox{(Split Knockoff)}\ \ \ \ \  T_q'=\min\left\{\lambda:\frac{|\{i:W_i'\le-\lambda\}|}{1\vee|\{i:W_i'\ge \lambda\}|}\le q\right\},\\
  &\mbox{(Split Knockoff+)}\ \ \ T_q'=\min\left\{\lambda:\frac{1+|\{i:W_i'\le-\lambda\}|}{1\vee|\{i:W_i'\ge \lambda\}|}\le q\right\},
\end{align*}
Then the respective selection set is defined as $\widehat{S}' = \{i: W_i'\ge T_q'\}$. In Proposition \ref{prop: inclusion}, we directly show that $\widehat{S}'\subseteq\widehat{S}$, which suggests our definition of $W$ statistics \eqref{eq: def_w} improves the selection power.

\begin{proposition}
\label{prop: inclusion}
For any value of $Z$ and $\tilde{Z}$, there always holds
\begin{equation*}
    \widehat{S}'\subseteq\widehat{S}.
\end{equation*}
\end{proposition}

\begin{proof}
    Firstly, by the definition of $W$ and $W'$, for all $i$, there holds
\begin{itemize}
    \item if $Z_i>\tilde{Z}_i$, $W'_i = Z_i = W_i$;
    \item if $Z_i<\tilde{Z}_i$, $W'_i = -\tilde{Z}_i<-Z_i = W_i$.
\end{itemize}
Therefore, there holds $W'_i\le W_i$ for all $i$. With this property, for all $\lambda>0$, there holds
\begin{align*}
    \{i: W_i\le -\lambda\}\subseteq \{i: W'_i\le -\lambda\},\ \{i: W'_i\ge \lambda\}\subseteq \{i: W_i\ge \lambda\},
\end{align*}
which further suggest that
\begin{align*}
    &\frac{|\{i:W_i\le-\lambda\}|}{1\vee|\{i:W_i\ge \lambda\}|}\le \frac{|\{i:W'_i\le-\lambda\}|}{1\vee|\{i:W'_i\ge \lambda\}|};  \frac{1+|\{i:W_i\le-\lambda\}|}{1\vee|\{i:W_i\ge \lambda\}|}\le \frac{1+|\{i:W'_i\le-\lambda\}|}{1\vee|\{i:W'_i\ge \lambda\}|}.
\end{align*} 
Thus by the definition of the thresholds $T_q'$ and $T_q$, there holds $T_q\le T_q'$. Combining this fact with the property that $W'_i\le W_i$ for all $i$, there further holds
    \begin{align*}
        \widehat{S}' = \{i: W'_i\ge T_q'\}\subseteq \{i: W_i\ge T_q'\}\subseteq \{i: W_i\ge T_q\} = \widehat{S}.
    \end{align*}
    This ends the proof.
\end{proof}
Proposition \ref{prop: inclusion} shows that the selection set given by $W$ is always a super set of that given by $W'$. Therefore, our currently choice of $W$ statistics \eqref{eq: def_w} always enjoys better selection power compared with that in \cite{barber2015controlling, barber2019knockoff}.

\section{Necessity of Sample Splitting}

\label{sec: inflation}

In this section, we show the necessity of the sample splitting scheme for Split Knockoffs. As mentioned in Section \ref{sec: analysis}, the sample splitting scheme brings conditional independence between the sign and magnitude of the $W$ statistics, an important property for Split Knockoffs to achieve provable $\dfdr$ control. In this section, 
we show that dropping the sample splitting scheme will lead to an inflation in the $\dfdr$ control of Split Knockoffs both theoretically and experimentally. 

The procedure of Split Knockoffs without the sample splitting scheme is presented in the following. The main difference is that the construction will be built on the full dataset $\D = (X, y)$.

\begin{mdframed}[roundcorner=10pt,frametitle=Stage 1: Compute the intercept term ($\beta(\lambda)$) on $\D$.]%\small
Compute the Split LASSO regularization path for $\beta$,
\begin{align*}
    \beta(\lambda):=\arg\min_{\beta}\min_\gamma \frac12\|\tilde y-A_{\beta}\beta-A_{\gamma}\gamma\|_2^2+\lambda\|\gamma\|_1, \ \ \ \lambda>0.
\end{align*}
where
\begin{equation*}
\tilde{y}=
\left(
\begin{array}{c}
\frac{y}{\sqrt{n}} \\
0_m
\end{array}
\right),
A_{\beta}=
\left(
\begin{array}{c}
\frac{X}{\sqrt{n}} \\
\frac{D}{\sqrt{\nu}} 
\end{array}
\right),
A_{\gamma}=
\left(
\begin{array}{c}
0_{n\times m} \\
-\frac{I_m}{\sqrt{\nu}} 
\end{array}
\right).
\end{equation*}
\end{mdframed}

\begin{mdframed}[roundcorner=10pt,frametitle=Stage 2: Compute the significance level $Z$ on $\D$.]%\small

\begin{enumerate}
\item Compute the Split LASSO regularization path for $\gamma$,
\begin{align*}
    \gamma(\lambda):=\arg\min_{\gamma} \frac12\|\tilde y-A_{\beta}\beta(\lambda)-A_{\gamma}\gamma\|_2^2+\lambda\|\gamma\|_1, \ \ \ \lambda>0.
\end{align*}
\item For all $i$, define $Z_i$ and record the sign of $\gamma_i(\lambda)$ upon being nonzero as:
\begin{align*}
    Z_i=  \sup\left\{\lambda:\gamma_i(\lambda)\neq0\right\},\ r_i=  \lim_{\lambda\to Z_i-}\sign{\gamma_i(\lambda)}.
\end{align*}
\end{enumerate}
\end{mdframed}

\begin{mdframed}[roundcorner=10pt,frametitle=Stage 3: Compute the significance level $\tilde{Z}$ on $\D$.]%\small

\begin{enumerate}
\item Compute the Split LASSO regularization path for $\tilde{\gamma}$, %using the solution of $\beta(\lambda)$ from Stage 1
\begin{align*}
    \tilde{\gamma}(\lambda) :=\arg \min_{\tilde{\gamma}} \frac12\|\tilde y-A_{\beta}\beta(\lambda)-\tilde{A}_{\gamma}\tilde{\gamma}\|_2^2+\lambda\|\tilde{\gamma}\|_1, \ \ \ \lambda>0,
\end{align*}
where for some proper nonnegative vector $\vecs\in\R^m$, $\tilde{A}_{\gamma}$ satisfies
\begin{align*}
    \tilde{A}_{\gamma}^T\tilde{A}_{\gamma} = A_{\gamma}^TA_{\gamma},\ 
    A_{\beta}^T\tilde{A}_{\gamma} = A_{\beta}^TA_{\gamma},\ A_{\gamma}^T\tilde{A}_{\gamma} =  A_{\gamma}^TA_{\gamma}-\diag(\vecs).
\end{align*}
$\tilde{A}_{\gamma}$ can be constructed in the same way as $\tilde{A}_{\gamma_2}$ in Section \ref{sec: construct split knockoff}.
\item For all $i$, define $\tilde{Z}_i$ as:
\begin{align*}
    \tilde{Z}_i=  \sup\left\{\lambda:\tilde\gamma_i(\lambda)\neq0\right\}.
\end{align*}
\end{enumerate}
\end{mdframed}

The $W$ statistics, the data dependent threshold $T_q$ with respect to the preset nominal $\dfdr$ level $q$, the selector $\widehat{S}$ and the sign estimator $\widehat{\sign}$ are defined in the same way as in Section \ref{sec: methodology}. Theorem \ref{theorem: directional fdr without split} presents the $\dfdr$ control of the above Split Knockoff procedure without sample splitting.

\begin{theorem}
    \label{theorem: directional fdr without split}
    For any linear transformation $D$ and any $0<q<1$, there exists a decreasing function $g(\nu)$ with respect to $\nu$, such that the following holds:
    \begin{itemize}
        \item[(a)] ($\mdfdr$ of Split Knockoff without sample splitting)
    \begin{align*}
        \E\left[\frac{|\{i\in \widehat{S}: \widehat{\sign}_i \neq \sign(\gamma^*_i)\}|}{|\widehat{S}|+q^{-1}}\right]\le g(\nu) q,
    \end{align*}
    \item[(b)] ($\dfdr$ of Split Knockoff+ without sample splitting)
    \begin{align*}
        \E\left[\frac{|\{i\in \widehat{S}: \widehat{\sign}_i \neq \sign(\gamma^*_i)\}|}{|\widehat{S}|\vee 1}\right]\le g(\nu) q,
    \end{align*}
    \end{itemize}
    where $g(\nu)$ goes to zero when $\nu$ goes to infinity. The detailed form of $g(\nu)$ is given by Equation \eqref{def: g(nu)} in the proof of Theorem \ref{theorem: directional fdr without split} in Section \ref{sec: proof thm late}.
\end{theorem}

The main difference between Theorem \ref{theorem: directional fdr without split} and Theorem \ref{theorem: directional fdr} is that Theorem \ref{theorem: directional fdr without split} no longer guarantees the $\dfdr$ control for all $\nu>0$. Instead, it states that the $\dfdr$ could be under control when $\nu$ is sufficiently large. In other words, there could be an inflation in the $\dfdr$ control when $\nu$ is small. In particular, we apply Split Knockoffs without sample splitting on the simulation examples in Section \ref{sec: simu_exp} for illustration.

\begin{figure}[!ht]
\centering
\subfigure[Performance in $D_1$]{
\begin{minipage}[t]{0.33\textwidth}
\centering
\includegraphics[width=\textwidth]{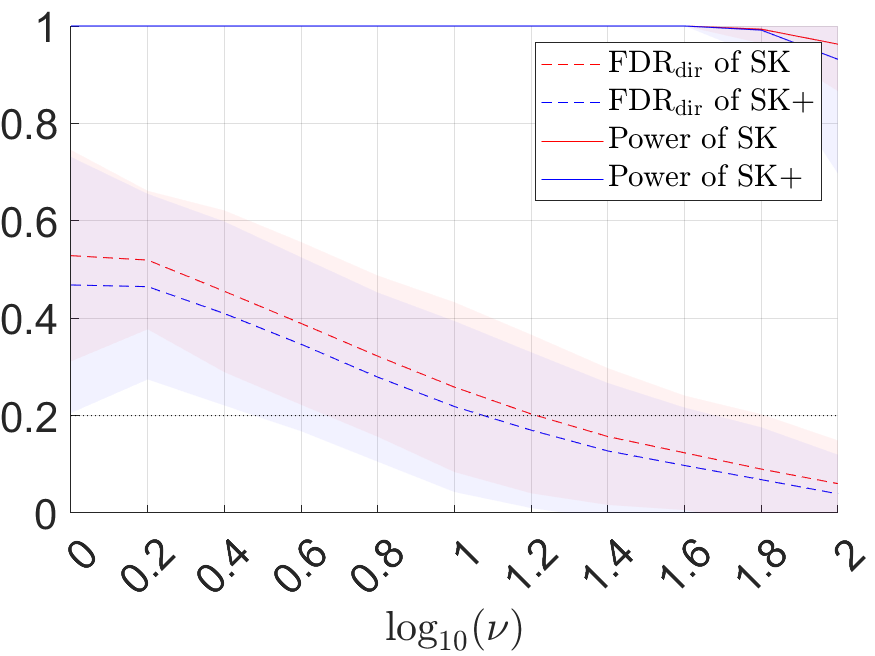}
%\caption{FDR comparison in $D_1$}
\end{minipage}%
}%
\subfigure[Performance in $D_2$]{
\begin{minipage}[t]{0.33\textwidth}
\centering
\includegraphics[width=\textwidth]{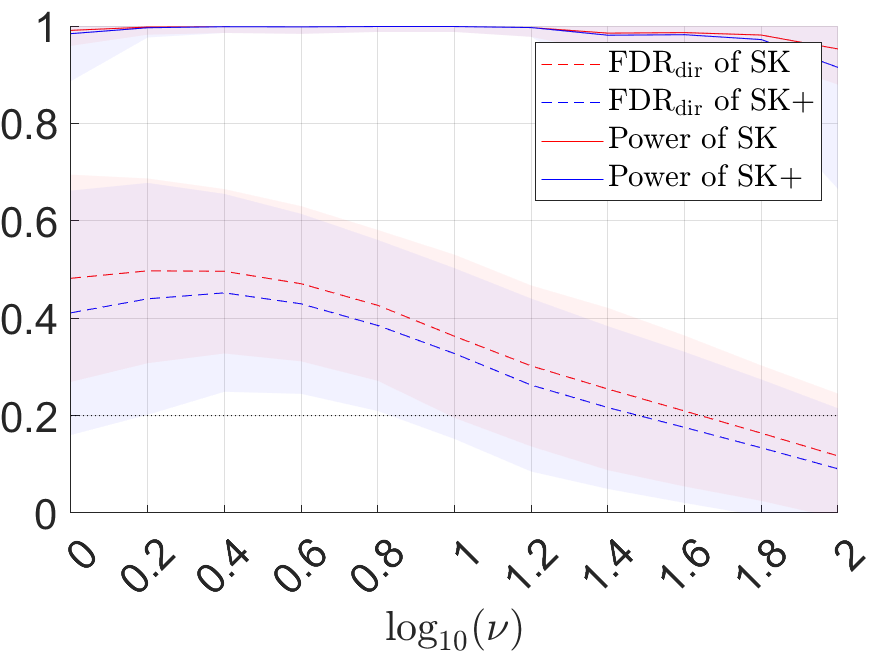}
%\caption{FDR comparison in $D_2$}
\end{minipage}%
}%
\subfigure[Performance in $D_3$]{
\begin{minipage}[t]{0.33\textwidth}
\centering
\includegraphics[width=\textwidth]{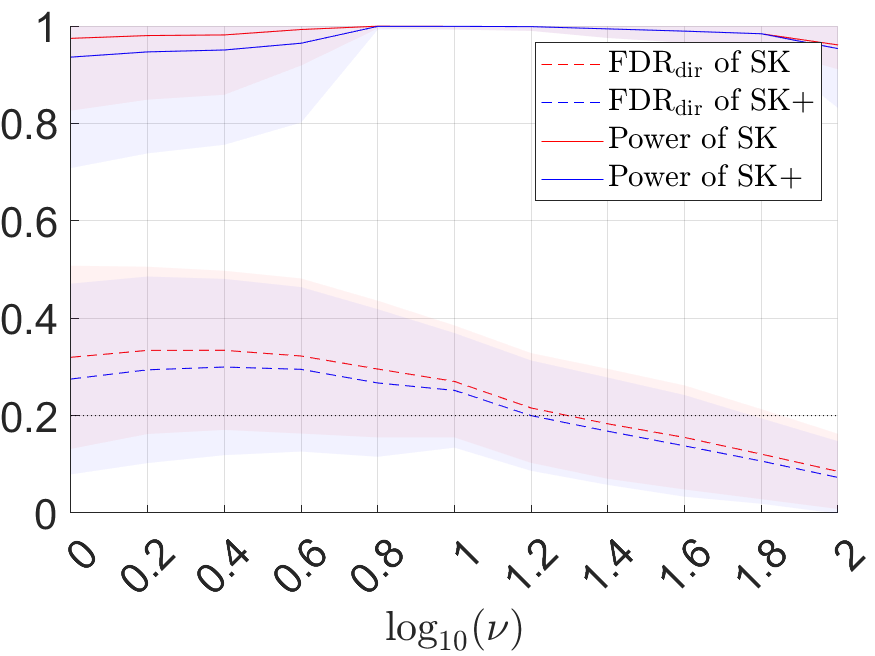}
%\caption{FDR performance in $D_3$}
\end{minipage}%
}%

\caption{Performance of Split Knockoffs without sample splitting: Power, $\dfdr$ for $q=0.2$. The curves in the figures represent the average performance in 200 simulation instances, while the shaded areas represent the 80\% confidence intervals truncated to the range $[0, 1]$. We use``SK(+)'' to represent ``Split Knockoff(+)'' for shorthand notations.}
\label{fig: comparison without split}
\end{figure}

In particular, Figure \ref{fig: comparison without split} presents the performance of Split Knockoffs without sample splitting. As one can observe in Figure \ref{fig: comparison without split}, on the one hand, Split Knockoffs without sample splitting exhibits higher selection power compared with Split Knockoffs due to an enlarged sample size. On the other hand, although Split Knockoffs without sample splitting can achieve desired $\dfdr$ control when $\nu$ is large, there could be an inflation in the $\dfdr$ when $\nu$ is small. This is because Theorem \ref{theorem: directional fdr without split} no longer ensures that the $\dfdr$  is upper bounded by $q$, different from Theorem \ref{theorem: directional fdr}. This validates the necessity of the sample splitting scheme experimentally.

\section{Sign Consistency of Split LASSO} 
\label{sec:pathconsistency}

In this section, we present the sign consistency of Split LASSO to discuss the effects of $\nu$ on the selection power of Split Knockoffs. 
Specially, we present that Split LASSO enjoys weaker incoherence conditions with the increase of $\nu$ which benefits discovering strong nonnull features, at a possible cost of losing weak ones when $\nu$ is huge and the gap between $\gamma$ and $D\beta$ is barely penalized. 

For shorthand notations, define $H_\nu := I_m - \frac{D[\Sigma_X+L_D]^{-1}D^T}{\nu}$, where $\Sigma_X = \frac{X^TX}{n}$ and $L_D = \frac{D^TD}{\nu}$. Further denote $H_\nu^{11}$, $H_\nu^{00}$ to be the gram matrix for the nonnull set $S_1:=\{i: \gamma^*_i\neq0\}$, the null set $S_0$ respectively, and $H_\nu^{10}$, $H_\nu^{01}$ to be the correlation matrices between $\S_1$ and $\S_0$. First of all, we assume the following restricted-strongly-convex condition to ensure the identifiability of the problem.   
\paragraph*{Restricted-Strongly-Convex Condition} There exists some $C_\mathrm{min}>0$, such that the smallest eigenvalue of $H_\nu^{11}$ is bounded below by $C_\mathrm{min}$, i.e.
\begin{align}\label{min_eigen}
    \lambda_\mathrm{min} (H_\nu^{11})>C_\mathrm{min}.
\end{align}

\subsection{$\nu$-Incoherence Condition}

In this section, we present the following $\nu$-incoherence condition which is crucial for the sign consistency of Split LASSO, which becomes weaker (easier to be satisfied) with the increase of $\nu$.

\paragraph*{$\nu$-Incoherence Condition} There exists a parameter $\chi_\nu\in (0, 1]$, such that
\begin{align}
    \|H_\nu^{01} [H_\nu^{11}]^{-1}\|_\infty \le 1-\chi_\nu.\label{incoherence}
\end{align}

As $\nu \to \infty$, $H_\nu = I_m - \frac{1}{\nu}D[\Sigma_X+L_D]^{-1}D^T \to I_m$, therefore $H_\nu^{01}\to 0_{|\S_0|\times|\S_1|}$, while $H_\nu^{11}\to I_{|S_1|}$. Therefore, with the increase of $\nu$, the left hand side of \eqref{incoherence} drops to zero, and the $\nu$-incoherence condition is satisfied with $\chi_\nu\to 1$.  In this regard, increasing $\nu$ helps meet the $\nu$-incoherence condition above.

\subsection{Path Consistency and Power}

Now we are ready to state Theorem \ref{thm:pathconsistency} on the sign consistency of Split LASSO under the restricted strongly convex condition and $\nu$-incoherence conditions. 
\begin{theorem} \label{thm:pathconsistency}
    Assume that the design matrix $X$ and $D$ satisfy the restricted-strongly-convex condition \eqref{min_eigen} and $\nu$-incoherence condition \eqref{incoherence}. Let the columns of $X$ be normalized as $\max_{i\in [1:p]}\frac{\|x_i\|_2}{\sqrt{n}}\le 1$. There exists $c>0$, such that for the sequence of $\{\lambda_n\}$ satisfying 
    \begin{align} \label{eq:earlypath}
        \lambda_n > \frac{c}{\chi_\nu}\sqrt{\frac{\sigma^2\ln m}{n}},
    \end{align}
     the following properties hold with probability larger than $1-4e^{-c n\lambda_n^2}$.
    \begin{enumerate}
        \item (No-false-positive) The solution $(\hat\beta, \hat\gamma)\in \R^p \times\R^m$ of Split LASSO on $\lambda_n$ does not have false positives with respect to $\gamma$. 
        \item (Sign-consistency) In addition, $\hat\gamma$ recovers the sign of $\gamma^*$, if there further holds 
        \begin{align}
            \min_{i\in \S_1}{\gamma^*_i}> \lambda_n \nu \left[\frac{\sigma}{2C_\mathrm{min}}+\| [H_\nu^{11}]^{-1}\|_\infty\right].\label{eq: snr requirement}
        \end{align}
        %where $C_{\min}$ is the minimal eigenvalue of $H_\nu^{11}$, in addition, $\hat\gamma$ recovers the sign of $\gamma^*$.
    \end{enumerate}
\end{theorem}

From Theorem \ref{thm:pathconsistency}, the influence of $\nu$ on the selection power of Split LASSO and Split Knockoffs can be understood as follows: (a) for the early stage of the Split LASSO path characterized by \eqref{eq:earlypath}, there is no false positive and only nonnull features are selected; (b) all the strong nonnull features whose magnitudes are larger than $O(\nu \sigma \chi_\nu^{-1}\sqrt{\ln m/n} )$ could be selected on the path with sign consistency. Hence a sufficiently large $\nu$ will ensure the incoherence condition for sign consistency such that strong nonnull features will be selected earlier on the Split LASSO path than the nulls, at the cost of possibly losing weak nonnull features below $O(\nu \sigma \chi_\nu^{-1}\sqrt{\ln m/n} )$. Therefore, good selection power relies on a proper choice of $\nu$ for the trade-off.  

The proof of this theorem will be provided in Section \ref{sec:proof_path_consistency}.

\section{Experimental Supplementary Material}

In this section, we provide supplementary material for simulation experiments and two real world applications on the Alzheimer's Disease and human age comparisons.

\subsection{Supplementary Material for Simulation Experiments}

\label{sec: supp simu}

In this section, we discuss the effects of various parameters/procedures in Split Knockoffs by simulation experiments. In particular, we study the effects of the signal noise ratio (SNR) and the sample splitting fraction $\frac{n_1}{n_1+n_2}$ on the performance of Split Knockoffs. Moreover, we discuss the effects of random sample splits, which leads to random selection sets $\widehat{S}$, while the random selection sets include the nonnull features with much higher frequencies compared with the null features.

\subsubsection{Effects of the Signal Noise Ratio}

\label{sec: snr}

In this section, we present the performance of Split Knockoffs with respect to different signal noise ratios, compared with Knockoffs when applicable. In this section, all the simulation settings are succeeded from Section \ref{sec: simu_exp}, except that we take $\beta^*\in \R^p$ as
\begin{equation*}
    \beta_i^*:=\left\{
    \begin{array}{ccl}
        A   &   & i \le 20,\ i \equiv 0, -1 (\mathrm{mod}\ 3),\\
        0   &   & \mathrm{otherwise},
    \end{array} \right.
\end{equation*}
where we test $\log_{10}(A)$ in the range of -0.5 to 0.5 with a step size 0.2. For Split Knockoffs, the $\nu$ is chosen by cross validation with Split LASSO on $\D_1$.

\begin{figure}[!ht]
\centering
\subfigure[$\dfdr$ in $D_1$]{
\begin{minipage}[t]{0.33\textwidth}
\centering
\includegraphics[width=\textwidth]{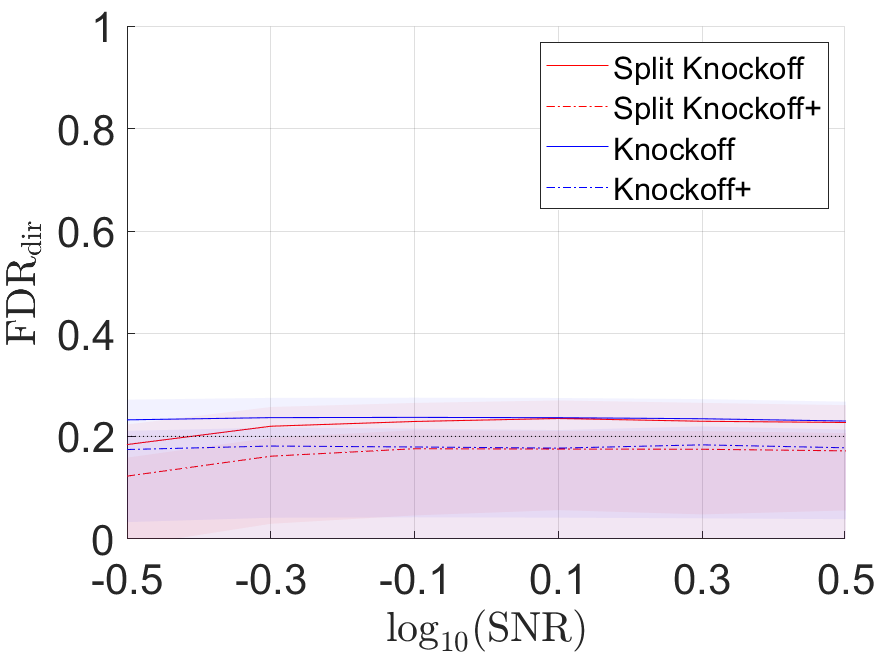}
%\caption{FDR comparison in $D_1$}
\end{minipage}%
}%
\subfigure[$\dfdr$ in $D_2$]{
\begin{minipage}[t]{0.33\textwidth}
\centering
\includegraphics[width=\textwidth]{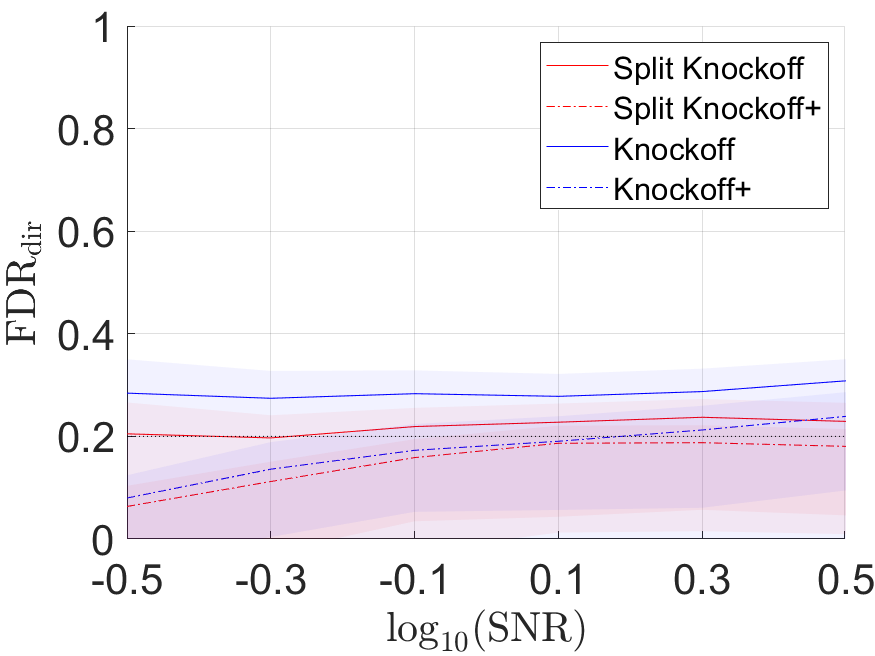}
%\caption{FDR comparison in $D_2$}
\end{minipage}%
}%
\subfigure[$\dfdr$ in $D_3$]{
\begin{minipage}[t]{0.33\textwidth}
\centering
\includegraphics[width=\textwidth]{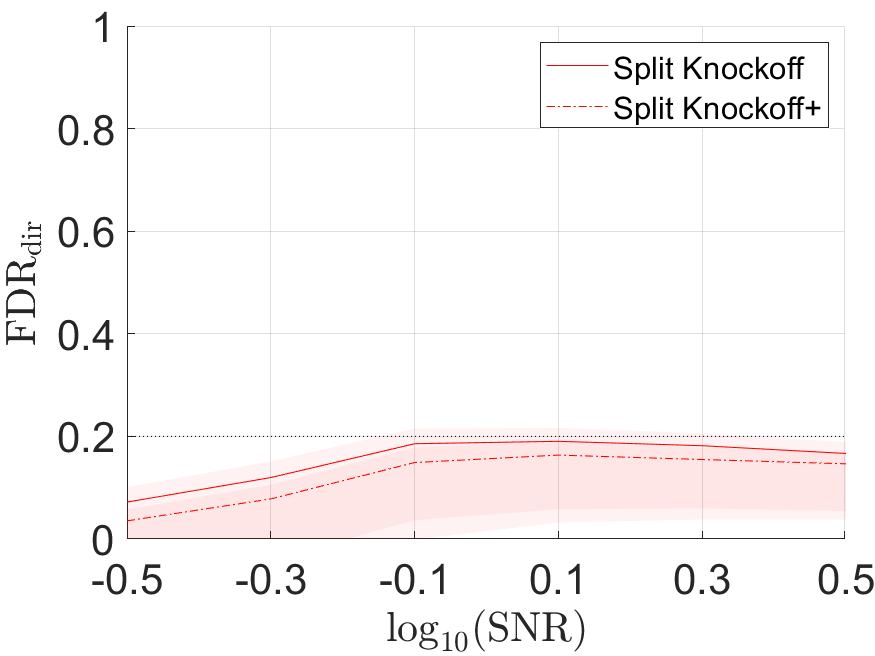}
%\caption{FDR performance in $D_3$}
\end{minipage}%
}%

\centering
\subfigure[Power in $D_1$]{
\begin{minipage}[t]{0.33\textwidth}
\centering
\includegraphics[width=\textwidth]{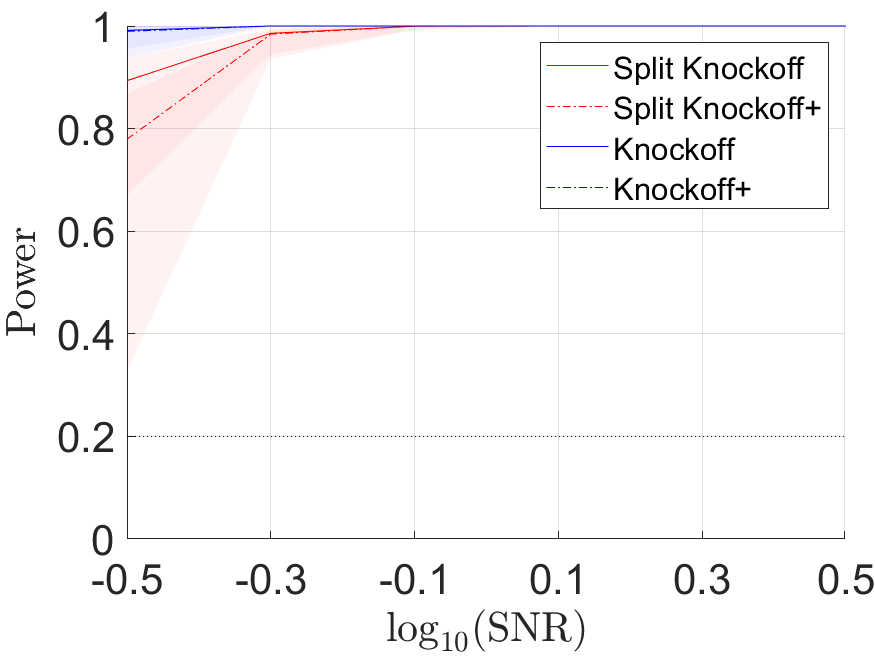}
%\caption{FDR comparison in $D_1$}
\end{minipage}%
}%
\subfigure[Power in $D_2$]{
\begin{minipage}[t]{0.33\textwidth}
\centering
\includegraphics[width=\textwidth]{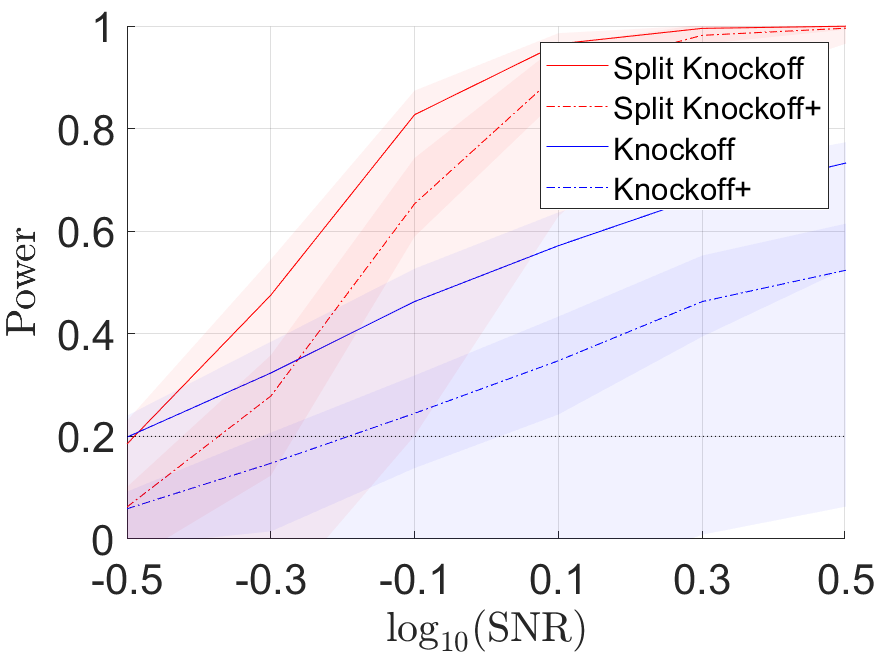}
%\caption{FDR comparison in $D_2$}
\end{minipage}%
}%
\subfigure[Power in $D_3$]{
\begin{minipage}[t]{0.33\textwidth}
\centering
\includegraphics[width=\textwidth]{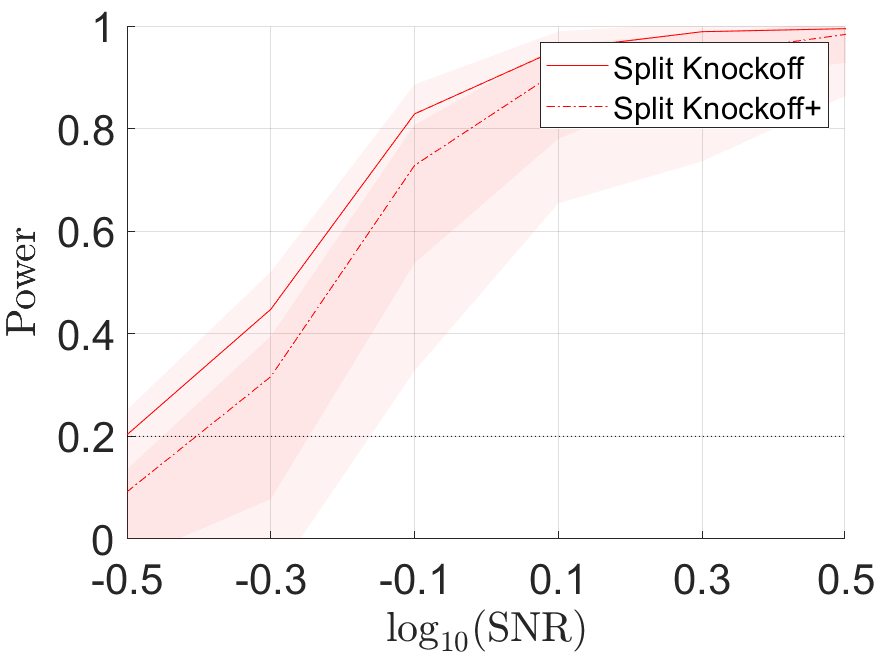}
%\caption{FDR performance in $D_3$}
\end{minipage}%
}%

\caption{Performance of Split Knockoffs and Knockoffs on various signal noise ratios: Power, $\dfdr$ on $q=0.2$. The curves in the figures represent the average performance in 200 simulation instances, while the shaded areas represent the 80\% confidence intervals truncated to the range $[0, 1]$.}
\label{fig: snr}
\end{figure}

As presented in Figure \ref{fig: snr}, Split Knockoffs achieves desired $\dfdr$ control for all signal noise ratios. For the selection power, when the linear transformation is trivial (e.g. $D_1$), Split Knockoffs drop more the selection power compared with Knockoffs when the signal noise ratio is low, as the minimal signal strength requirement in Theorem \ref{thm:pathconsistency} for Split LASSO may no longer be satisfied, while the improvement in selection power by the weaker $\nu$-incoherence conditions \eqref{incoherence} is not significant. 

However, when the linear transformation is non-trivial (e.g. $D_2$), Split Knockoffs exhibit higher selection power compared with Knockoffs when the signal noise ratio is reasonably large such that the selection power is not completely lost. In these cases, the improved $\nu$-incoherence conditions \eqref{incoherence} brought by the orthogonal design \eqref{eq: reformulated model on D_2} in variable splitting take advantage over the  minimal signal strength requirement in Theorem \ref{thm:pathconsistency}, and results in higher selection power for Split Knockoffs.

\subsubsection{Effects of the Sample Splitting Fraction}

In this section, we study the effects of the sample splitting fraction $\frac{n_1}{n_1+n_2}$ to the performance of Split Knockoffs. In this section, all the simulation settings are succeeded from Section \ref{sec: simu_exp}, except that the sample splitting fraction $\frac{n_1}{n_1+n_2}$ is tested between 0.1 to 0.8 with a step size 0.1.

\begin{figure}[!ht]
\centering
\subfigure[$\dfdr$ in $D_1$]{
\begin{minipage}[t]{0.33\textwidth}
\centering
\includegraphics[width=\textwidth]{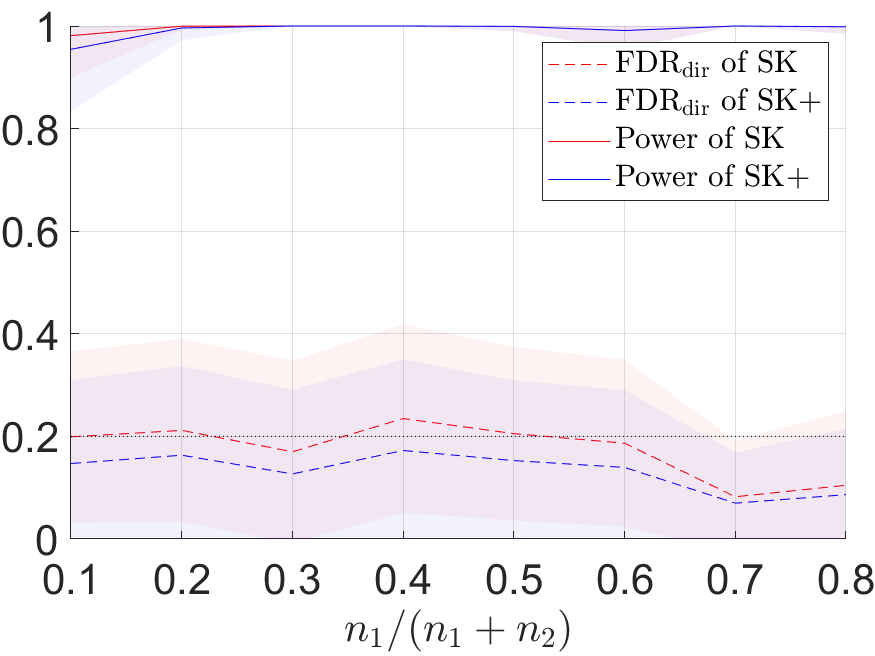}
%\caption{FDR comparison in $D_1$}
\end{minipage}%
}%
\subfigure[$\dfdr$ in $D_2$]{
\begin{minipage}[t]{0.33\textwidth}
\centering
\includegraphics[width=\textwidth]{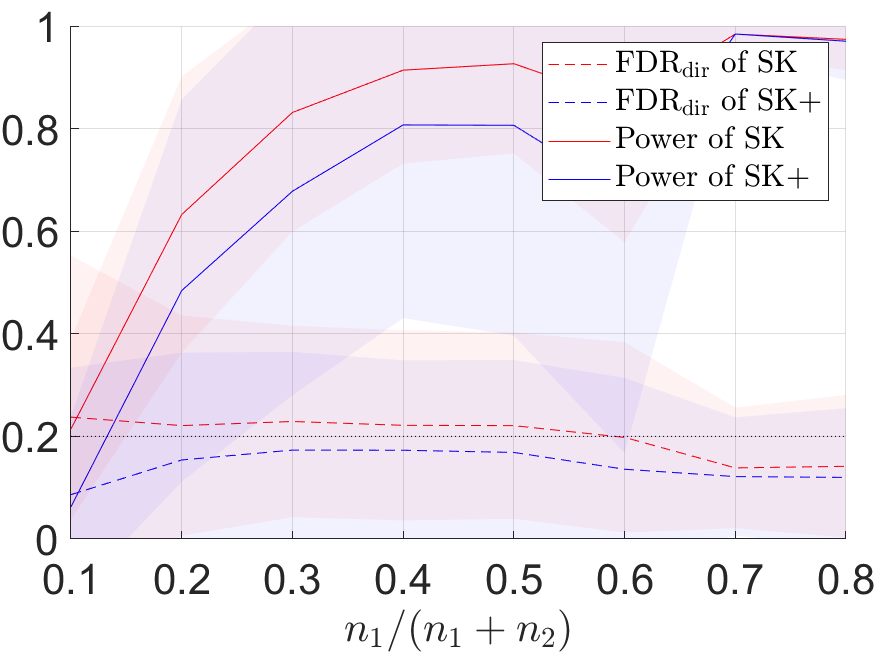}
%\caption{FDR comparison in $D_2$}
\end{minipage}%
}%
\subfigure[$\dfdr$ in $D_3$]{
\begin{minipage}[t]{0.33\textwidth}
\centering
\includegraphics[width=\textwidth]{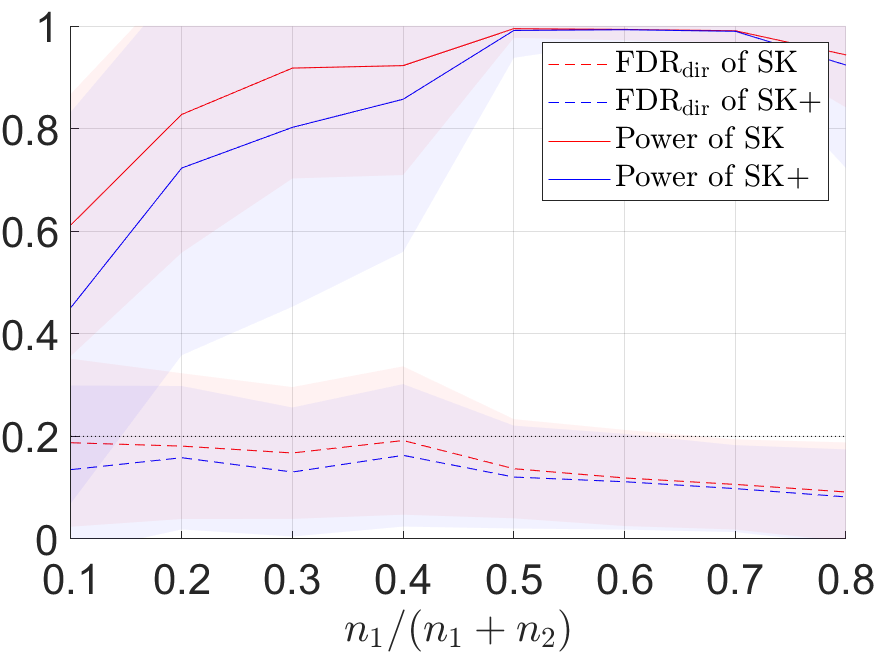}
%\caption{FDR performance in $D_3$}
\end{minipage}%
}%

\caption{Performance of Split Knockoffs on various sample split fractions: Power, $\dfdr$ on $q=0.2$. The curves in the figures represent the average performance in 200 simulation instances, while the shaded areas represent the 80\% confidence intervals truncated to the range $[0, 1]$.}
\label{fig: fractions}
\end{figure}

As presented in Figure \ref{fig: fractions}, in all cases, the Split Knockoffs achieve desired $\dfdr$ control. Meanwhile, for the cases where the linear transformation is nontrivial ($D_2$ or $D_3$), the selection power of Split Knockoffs exhibits an increasing trend when $\frac{n_1}{n_1+n_2}$ is low, and decreases a little bit when $\frac{n_1}{n_1+n_2}$ is very high. On the other hand, for the case that the linear transformation is trivial ($D_1$), where the recovery of \eqref{eq: structural sparsity model} becomes easier, the selection power of Split Knockoffs does not vary a lot when $\frac{n_1}{n_1+n_2}$ changes.

Therefore, in the cases where the recovery of \eqref{eq: structural sparsity model} is hard ($D$ is non-trivial, the sample size is limited, the signal noise ratio is low, etc.), it is favorable to take reasonably large $\frac{n_1}{n_1+n_2}$ to improve the selection power. Meanwhile, when the recovery of \eqref{eq: structural sparsity model} is easy, the sample splitting fraction $\frac{n_1}{n_1+n_2}$ does not make a big difference in the selection power.

\subsubsection{Effects of Random Sample Splits}

The procedure of Split Knockoffs involves a sample splitting scheme to ensure the $\dfdr$ control as discussed in Section \ref{sec: inflation}. Meanwhile, the random sample splits consequently lead to random selection sets $\widehat{S}$. In this section, we show by simulation experiments that the random selection sets include the nonnull features with much higher frequencies compared with the null features.
\begin{figure}[!ht]
\centering
\subfigure[Selection frequencies in $D_1$]{
\begin{minipage}[t]{0.3\textwidth}
\centering
\includegraphics[width=\textwidth]{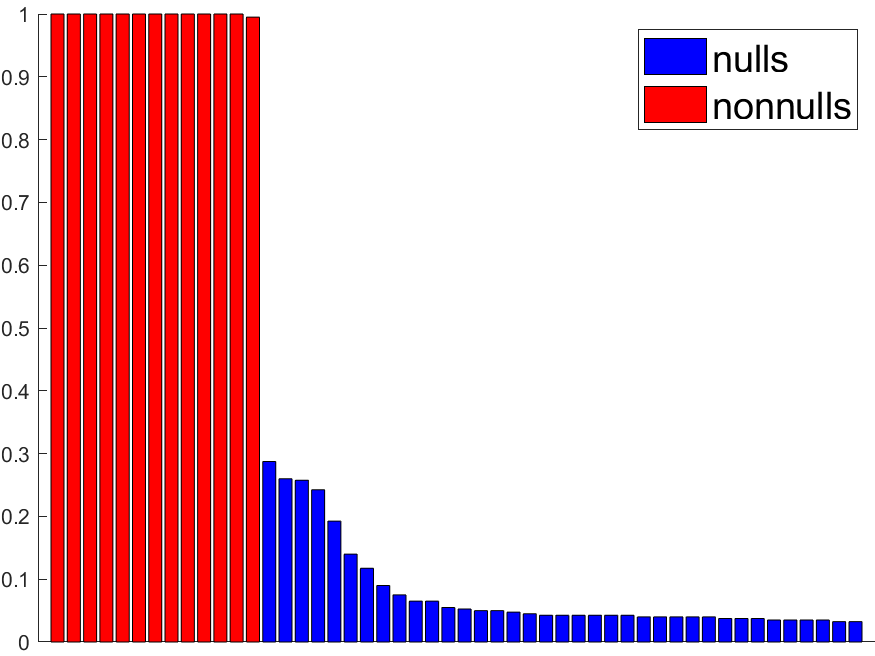}
%\caption{FDR comparison in $D_1$}
\end{minipage}%
}%
\subfigure[Selection frequencies in $D_2$]{
\begin{minipage}[t]{0.3\textwidth}
\centering
\includegraphics[width=\textwidth]{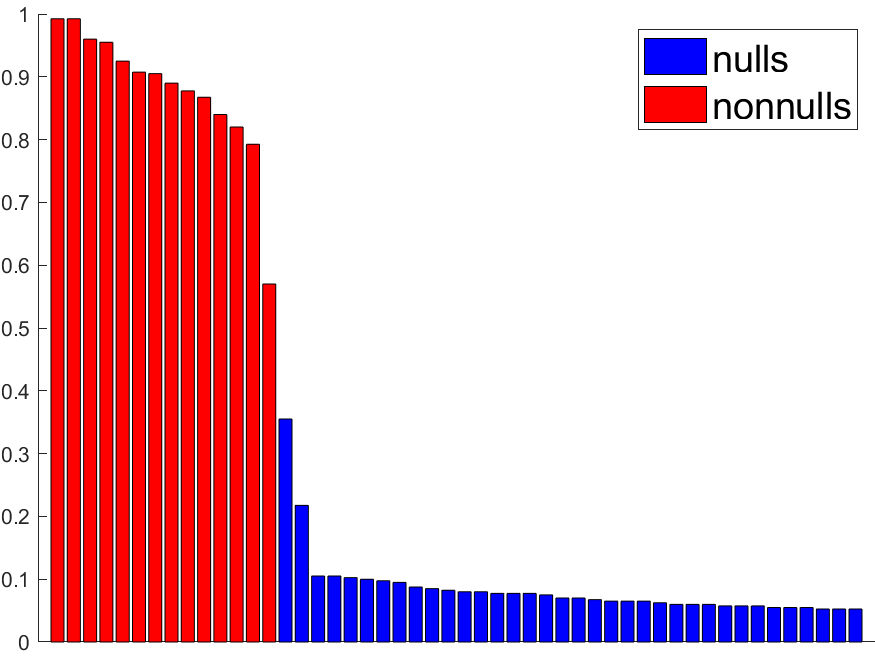}
%\caption{Power comparison in $D_1$}
\end{minipage}%
}%
\subfigure[Selection frequencies in $D_3$]{
\begin{minipage}[t]{0.3\textwidth}
\centering
\includegraphics[width=\textwidth]{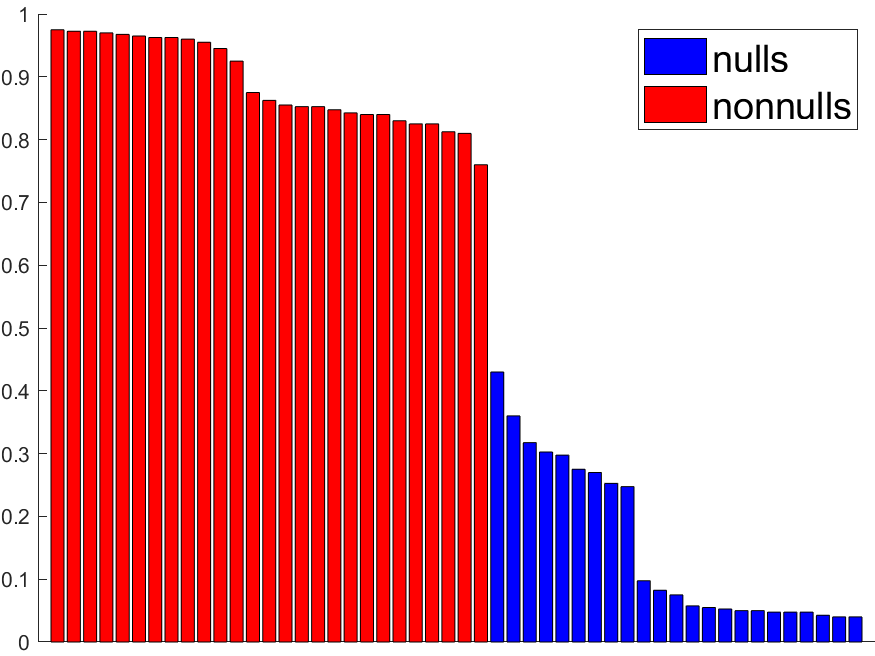}
\end{minipage}%
}%

\caption{Top 50 most frequently selected features of Split Knockoff  in simulation experiments among 20 different random sample splits ($q = 0.2$). Each bar in the figure represents the selection frequency of a particular feature.}
\label{fig: robust}
\end{figure}

In this section, all the simulation settings are succeeded from Section \ref{sec: simu_exp}. Moreover, Split Knockoffs are performed with respect to 20 different random sample splits, whose respective selection sets $\widehat{S}$ are all recorded. In Figure \ref{fig: robust}, we present the top 50 most frequently selected features by Split Knockoff among these random sample splits.

As presented in Figure \ref{fig: robust}, most of the nonnull features (in red) are selected with much higher  frequencies compared with the null features (in blue). This suggests that Split Knockoffs generate selection sets which robustly include the nonnull features.

\subsection{Supplementary Material for Alzheimer's Disease}

Figure \ref{fig: ad front} illustrates the lesion regions and high contrast connections selected by Split Knockoff. In the graph, each vertex represents a Cerebrum brain region in Automatic Anatomical Labeling (AAL) atlas \citep{tzourio2002automated}, with the abbreviation of each region marked in the vertex. In this section, we provide the comparison table between the full region names and their abbreviations in Table \ref{tab:name of region}. 
%Specially, vertices with a circle shape represent the left brain regions, while the ones with a square shape represent the right brain regions. There will be an edge connecting two vertices if and only if the respective two brain regions are adjacent.

\begin{center}
    \begin{longtable}{cc}
    \caption{Names and Abbreviations for Cerebrum Brain Anatomical Regions}
    \label{tab:name of region}\\
    \toprule
    Region Name & Abbreviation\\
\midrule
Precental gyrus & PreCG\\
 Superior frontal gyrus,  dorsolateral  & SFGdor\\
 Superior frontal gyrus,  orbital part  & ORBsup\\
Middle frontal gyrus & MFG\\
 Middle frontal gyrus,  orbital part  & ORBmid\\
 Inferior frontal gyrus,  opercular part  & IFGoperc\\
 Inferior frontal gyrus,  triangular part  & IFGtriang\\
 Inferior frontal gyrus,  orbital part  & ORBinf\\
Rolandic operculum & ROL\\
Supplementary motor area & SMA\\
Olfactory cortex & OLF\\
 Superior frontal gyrus,  medial  & SFGmed\\
 Superior frontal gyrus,  medial orbital  & ORBsupmed\\
Gyrus rectus & REC\\
Insula & INS\\
Anterior cingulate and paracingulate gyri & ACG\\
Median cingulate and paracingulate gyri & MCG\\
Posterior cingulate gyrus & PCG\\
Hippocampus & HIP\\
Parahippocampal gyrus & PHG\\
Amygdala & AMYG\\
Calcarine fissure and surrounding cortex & CAL\\
Cuneus & CUN\\
Lingual gyrus & LING\\
Superior occipital gyrus & SOG\\
Middle occipital gyrus & MOG\\
Inferior occipital gyrus & IOG\\
Fusiform gyrus & FFG\\
Postcentral gyrus & PoCG\\
Superior parietal gyrus & SPG\\
 Inferior parietal,  but supramarginal and angular gyri  & IPL\\
Supramarginal gyrus & SMG\\
Angular gyrus & ANG\\
Precuneus & PCUN\\
Paracentral lobule & PCL\\
Caudate nucleus & CAU\\
Lenticular nucleus putamen & PUT\\
 Lenticular nucleus,  pallidum  & PAL\\
Thalamus & THA\\
Heschl gyrus & HES\\
Superior temporal gyrus & STG\\
Temporal pole: superior temporal gyrus & TPOsup\\
Middle temporal gyrus & MTG\\
Temporal pole: middle temporal gyrus & TPOmid\\
Inferior temporal gyrus & ITG\\
\bottomrule
    \end{longtable}
\end{center}

\subsection{Supplementary Material for Human Age Comparisons}

\label{sec: supp age}

In this section, we compare the performance of Split Knockoffs against classical pairwise comparison models, the Bradley-Terry model and the Thurstone-Mosteller model. In particular, we first show that in the human age comparisons problem, the linear model \eqref{eq: structural sparsity model}, where Split Knockoffs rely on, is consistent with the Bradley-Terry model and the Thurstone-Mosteller model. Then we further show that the pairwise comparisons made by Split Knockoffs are consistent with the Bradley-Terry model and the Thurstone-Mosteller model in the human age comparisons problem.

We first briefly introduce how to apply the linear model, the Bradley-Terry model and the Thurstone-Mosteller model in the human age comparisons problem.

For the linear model \eqref{eq: structural sparsity model}, let $X\in\R^{n\times p}$, $y\in\R^n$, $D\in\R^{m\times p}$ be defined in the same way as in Section \ref{sec: age}. Then the comparison for each face image  pair $(i, j)$ are made based on the signs of $\hat\beta_i-\hat\beta_j$, where $\hat\beta_i$ represents the estimated regression coefficient (namely the ``score'' in Figure \ref{fig: scores}) of the face image $i$. In this section, consider the simplest way to solve $\hat\beta\in\R^p$, i.e. solve $\hat\beta$ by minimizing the $\ell_2$ regression loss,
\begin{align*}
    \hat\beta := \argmin_{t\in\R^p} \|y-Xt\|_2^2.
\end{align*}

The Bradley-Terry model models the probability that the face image $i$ looks older than the face image $j$ for each pair of $(i, j)$ by
\begin{align*}
    \Prob[\mbox{$i$ looks older than $j$}] = \frac{e^{\beta_i}}{e^{\beta_i}+e^{\beta_j}},
\end{align*}
where $\beta_i$ is the score of the face image $i$, and $\beta\in\R^p$ is solved by maximizing the log-likelihood. In fact, maximizing the log-likelihood in the Bradley-Terry  model is equivalent with minimizing the logistic regression loss with respect to $X\in\R^{n\times p},\ y\in\R^n$ defined in Section \ref{sec: age}, since
\begin{align}
    \beta & = \argmax_{t\in\R^p} \sum_{k} \left\{z_{(i, j)}^k\ln\frac{e^{t_i}}{e^{t_i}+e^{t_j}}+(1-z_{(i, j)}^k)\ln\frac{e^{t_j}}{e^{t_i}+e^{t_j}}\right\},\nonumber\\
    &= \argmin_{t\in\R^p} \sum_k \left\{-y_k\ln\frac{1}{1+e^{-X_k\cdot t}}-(1-y_k)\ln\frac{e^{-X_k\cdot t}}{1+e^{-X_k\cdot t}}\right\},\label{eq: equivalence}
\end{align}
where $z_{(i, j)}^k = 1\{\mbox{the volunteer thinks $i$ looks older than $j$ in the $k$-th annotation}\}$, and $X_k\in\R^p$ represents the $k$-th row of $X$.

The Thurstone-Mosteller model models the probability that the face image $i$ looks older than the face image $j$ for each pair of $(i, j)$ by
\begin{align*}
    \Prob[\mbox{$i$ looks older than $j$}] = \Prob[A_i>A_j],
\end{align*}
where random variables $A_i$ for all $i$ are drafted independently from the Gaussian distribution $\mathcal{N}(\mu_i, \sigma^2)$, and $\mu_i$ stands for the score of the face image $i$. The scores $\{\mu_i\}_{i=1}^p$ are solved by maximizing the log-likelihood function.

\begin{figure}[!ht]
\centering
\includegraphics[width=0.5\textwidth]{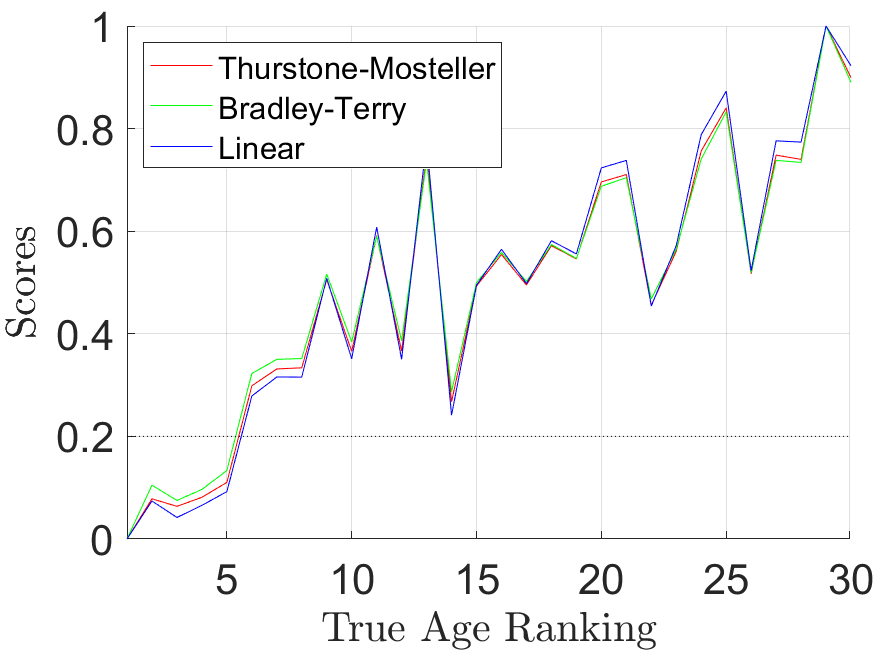}
\caption{Scores of face images given by the basic linear model, the Bradley-Terry model and the Thurstone-Mosteller model in the ascending order of true ages.}
\label{fig: scores}
\end{figure}

In Figure \ref{fig: scores}, we compare the scores of face images given by various models in the ascending order of true ages, where all scores are normalized to the interval $[0, 1]$. As presented in Figure \ref{fig: scores}, the scores given by various models, including the linear model, are largely consistent with each other. This provides evidence Split Knockoffs are applicable for this particular problem.

Moreover, Figure \ref{fig: heat map} compares the pairwise comparisons made by Split Knockoff (with the cross validation optimal choice of $\log_{10}\nu = 2.2$) against the ascending orders of face images given by the true ages as well as scores of the Bradley-Terry model and the Thurstone-Mosteller model.

\begin{figure}[!ht]
\centering
\subfigure[True ages]{
\begin{minipage}[t]{0.3\textwidth}
\centering
\includegraphics[width=\textwidth]{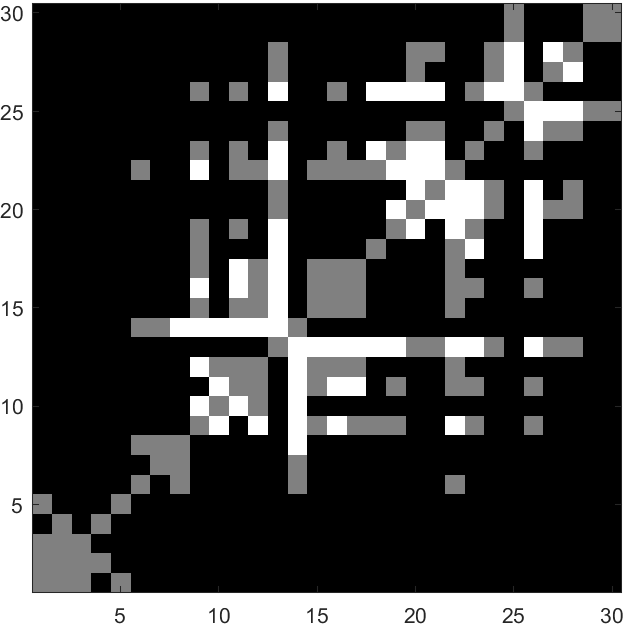}
%\caption{FDR comparison in $D_1$}
\end{minipage}%
}%
\iffalse
\subfigure[Basic linear]{
\begin{minipage}[t]{0.25\textwidth}
\centering
\includegraphics[width=\textwidth]{Fig/figure_rankd.png}
%\caption{Power comparison in $D_1$}
\end{minipage}%
}%
\fi
\subfigure[Bradley-Terry]{
\begin{minipage}[t]{0.3\textwidth}
\centering
\includegraphics[width=\textwidth]{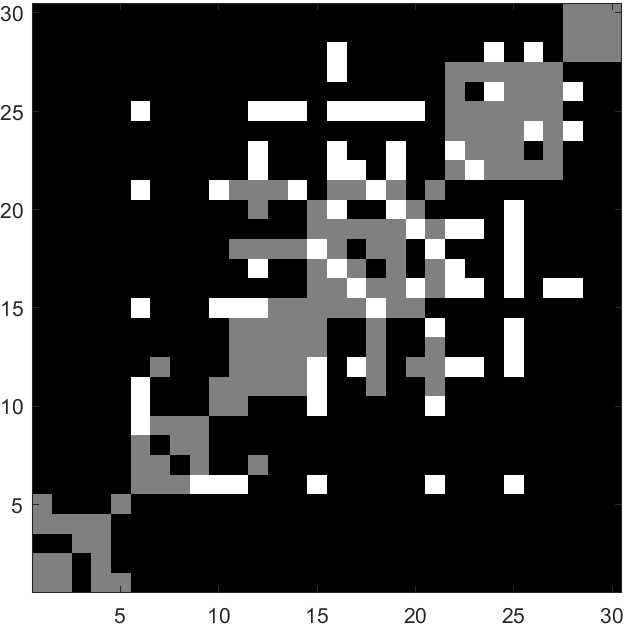}
\end{minipage}%
}%
\subfigure[Thurstone-Mosteller]{
\begin{minipage}[t]{0.3\textwidth}
\centering
\includegraphics[width=\textwidth]{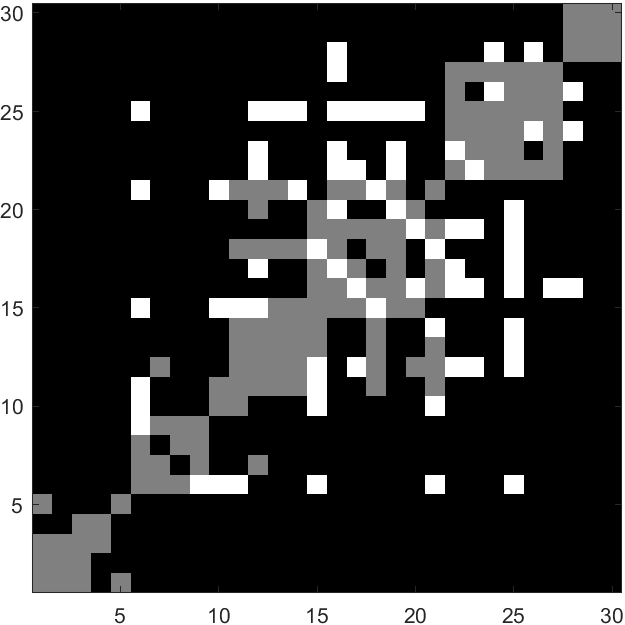}
\end{minipage}%
}%

\caption{Heat maps comparing the pairwise comparisons made by Split Knockoff ($\log_{10}\nu = 2.2$) against the ascending orders of face images given by the true ages as well as scores of the Bradley-Terry model and the Thurstone-Mosteller model. Each square in Figure \ref{fig: heat map} represents a comparison of two images labeled by the coordinates of x-axis and the y-axis. The squares in black, white, or grey represent the comparisons where Split Knockoff correctly estimates, wrongly estimates, or fail to estimate the age differences respectively.}
\label{fig: heat map}
\end{figure}

As presented in Figure \ref{fig: heat map}, for each order of face images, the majority of the false discoveries (white) and non-discoveries (grey) concentrates along the diagonal of the squares, while the top-left corner and bottom-right corner of the squares mainly consist of true discoveries (black). This shows that Split Knockoff works well in discovering strong signals where the age difference is large, at a cost of possible loss of weak signals where the age difference is small. Moreover, it demonstrate that Split Knockoff with linear model is largely consistent with other models such as the Bradley-Terry and the Thurstone-Mosteller models.

\section{Proof of Proposition \ref{prop: structure of tagamma}}

\label{sec: structure of tagamma}

\begin{proof}
    From $\tilde{A}^T_{\gamma_2}A_{\gamma_2} = A^T_{\gamma_2}A_{\gamma_2} - \diag(\vecs)$ in  Equation \eqref{eq: copy}, there holds
    \begin{align*}
        -\frac{\tilde{A}^T_{\gamma_2,2}}{\sqrt{\nu}}= \tilde{A}^T_{\gamma_2}A_{\gamma_2} & = A^T_{\gamma_2}A_{\gamma_2} - \diag(\vecs) = \frac{I_m}{\nu}-\diag(\vecs).
    \end{align*}
    Therefore $\tilde{A}_{\gamma_2,2}= -\frac{I_m}{\sqrt{\nu}}+\sqrt{\nu}\diag(\vecs)$. 
    
    Moreover, by $ \tilde{A}^T_{\gamma_2}A_{\beta_2}  = A^T_{\gamma_2}A_{\beta_2}$ in Equation \eqref{eq: copy}, there holds
    \begin{align*}
        \tilde{A}^T_{\gamma_2,1}\frac{X_2}{\sqrt{n_2}}+\tilde{A}^T_{\gamma_2,2}\frac{D}{\sqrt{\nu}}= \tilde{A}^T_{\gamma_2}A_{\beta_2} & = A^T_{\gamma_2}A_{\beta_2} = -\frac{D}{\nu}.
    \end{align*}
    Plug $\tilde{A}_{\gamma_2,2}= -\frac{I_m}{\sqrt{\nu}}+\sqrt{\nu}\diag(\vecs)$ into the above equation, there holds $\tilde{A}^T_{\gamma_2,1}{X_2} = -{\sqrt{n_2}}\diag(\vecs)D$. 
    
    Lastly, by $\tilde{A}_{\gamma_2}^T\tilde{A}_{\gamma_2} =  A_{\gamma_2}^TA_{\gamma_2}$ in Equation \eqref{eq: copy}, there holds
    \begin{align*}
    \tilde{A}^T_{\gamma_2,1}\tilde{A}_{\gamma_2,1}+\tilde{A}^T_{\gamma_2,2}\tilde{A}_{\gamma_2,2}&=\tilde{A}_{\gamma_2}^T\tilde{A}_{\gamma_2} =  A_{\gamma_2}^TA_{\gamma_2} = \frac{I_m}{\nu}.\nonumber
\end{align*}
Plug $\tilde{A}_{\gamma_2,2}= -\frac{I_m}{\sqrt{\nu}}+\sqrt{\nu}\diag(\vecs)$ into the above equation, there holds $
    \tilde{A}_{\gamma_2,1}^T\tilde{A}_{\gamma_2,1}=\diag(\vecs)(2I_m-\diag(\vecs)\nu)$.
\end{proof}

\section{Proof of Theorem \ref{theorem: directional fdr}}
\label{sec: proof thm}

\begin{proof}

In the proof, we first show that Theorem \ref{theorem: directional fdr} can be deducted from Equation \eqref{eq: target}, following the standard procedure of Knockoffs in \cite{barber2015controlling, barber2019knockoff}. Then we show how Lemma \ref{lemma: est of sign} will lead to Equation \eqref{eq: target} using a supermartingale inequality as in \cite{barber2019knockoff}.

Now, we will present the first point for $\mdfdr$ of Split Knockoff and $\dfdr$ of Split Knockoff+ separately.

1)$\mdfdr$ of Split Knockoff: The $\mdfdr$ can be written as:
\begin{align}
    &\ \ \ \ \Expect\left[\frac{\sum_{i}1\{W_i\ge T_q, \widehat{\sign}_i\neq\sign(\gamma_i^*)\}}{\sum_{i}1\{W_i\ge T_q\}+q^{-1}}\right]\nonumber\\
    &= \Expect\left[\frac{1+\sum_{i}1\{W_i\le -T_q\}}{\sum_{i}1\{W_i\ge T_q\}+q^{-1}}\cdot \frac{\sum_{i}1\{W_i\ge T_q, \widehat{\sign}_i\neq\sign(\gamma^*_i)\}}{1+\sum_{i}1\{W_i\le -T_q\}}\right]\label{eq: fdr decompose}.
\end{align}
From the definition of the Split Knockoff threshold, there holds
\begin{equation}
    \frac{\sum_{i}1\{W_i\le -T_q\}}{1\vee\sum_{i}1\{W_i\ge T_q\}}\le q<1,\nonumber
\end{equation}
which implies
\begin{equation}
    \sum_{i}1\{W_i\le -T_q\}\le q\sum_{i}1\{W_i\ge T_q\}.\nonumber
\end{equation}
Consequently, there holds
\begin{align}
    \frac{1+\sum_{i}1\{W_i\le -T_q\}}{\sum_{i}1\{W_i\ge T_q\}+q^{-1}}\le& \frac{1+q[\sum_{i}1\{W_i\ge T_q\}]}{\sum_{i}1\{W_i\ge T_q\}+q^{-1}}=q.\nonumber
\end{align}
Combined with Equation \eqref{eq: fdr decompose}, there holds
\begin{align}
    &\ \ \ \ \Expect\left[\frac{\sum_{i}1\{W_i\ge T_q, \widehat{\sign}_i\neq\sign(\gamma_i^*)\}}{\sum_{i}1\{W_i\ge T_q\}+q^{-1}}\right] \\
    &\le q\Expect\left[\frac{\sum_{i}1\{W_i\ge T_q, \widehat{\sign}_i\neq\sign(\gamma^*_i)\}}{1+\sum_{i}1\{W_i\le -T_q\}}\right],\nonumber\\
    &\le q\Expect\left[\frac{\sum_{i}1\{W_i\ge T_q, \widehat{\sign}_i\neq\sign(\gamma^*_i)\}}{1+\sum_{i}1\{W_i\le -T_q, \widehat{\sign}_i\neq\sign(\gamma^*_i)\}}\right].\label{eq: standard}
\end{align}

2)$\dfdr$ of Split Knockoff+: The $\dfdr$ satisfies
\begin{align}
    &\ \ \ \ \Expect\left[\frac{\sum_{i}1\{W_i\ge T_q, \widehat{\sign}_i\neq\sign(\gamma_i^*)\}}{1\vee\sum_{i}1\{W_i\ge T_q\}}\right]\nonumber\\
    &= \Expect\left[\frac{1+\sum_{i}1\{W_i\le -T_q\}}{1\vee\sum_{i}1\{W_i\ge T_q\}}\frac{\sum_{i}1\{W_i\ge T_q, \widehat{\sign}_i\neq\sign(\gamma^*_i)\}}{1+\sum_{i}1\{W_i\le -T_q\}}\right], \nonumber\\
    &\le q\Expect\left[\frac{\sum_{i}1\{W_i\ge T_q, \widehat{\sign}_i\neq\sign(\gamma^*_i)\}}{1+\sum_{i}1\{W_i\le -T_q, \widehat{\sign}_i\neq\sign(\gamma^*_i)\}}\right]. \label{eq: standard +}
\end{align}

Equation \eqref{eq: standard} and \eqref{eq: standard +} together show that Theorem \ref{theorem: directional fdr} can be deducted from Equation \eqref{eq: target}. 

Now we proceed to prove Equation \eqref{eq: target} using Lemma \ref{lemma: est of sign}. In particular, we will show that for any $\D_1$, Equation \eqref{eq: conditional target} holds. 
%there holds
%\begin{align}
    %\Expect\left.\left[\frac{\sum_{i}1\{W_i\ge T_q, \widehat{\sign}_i\neq\sign(\gamma^*_i)\}}{1+\sum_{i}1\{W_i\le -T_q, \widehat{\sign}_i\neq\sign(\gamma^*_i)\}}\right|\D_1\right]\le \min(\alpha(\nu), 1). \label{eq: conditional target}
%\end{align}
Equation \eqref{eq: conditional target} is sufficient for Equation \eqref{eq: target} in the sense that taking expectation over $\D_1$ on Equation \eqref{eq: conditional target} gives us the desired result.

For shorthand notations, conditional on $\D_1$, we rearrange the subscripts of $\{W_i\}_{i=1}^n$ as $\{W_{(i)}\}_{i=1}^n$, such that $\{(i)\}_{i=1}^{m^*} = \{\widehat{\sign}_i \neq \sign(\gamma^*_i)\}$, and $|W_{(1)}|\ge |W_{(2)}|\ge\cdots\ge |W_{(m^*)}|$, where $m^*:=|\{\widehat{\sign}_i \neq \sign(\gamma^*_i)\}|$. In other words, we rearrange the subscripts of $W$, such that the features whose signs are wrongly estimated appear first and are in a decreasing order of their absolute values. Further define $B_{(i)} = 1\{W_{(i)}<0\}$, then from Lemma \ref{lemma: est of sign} on the independence of $\sign(W_i)$ conditional on $\D_1$, $B_{(i)}$ are independent from each other conditional on $\D_1$. Moreover, there holds
\begin{align}
    \left.\frac{\sum_{i}1\{W_i\ge T_q, \widehat{\sign}_i\neq\sign(\gamma^*_i)\}}{1+\sum_{i}1\{W_i\le -T_q, \widehat{\sign}_i\neq\sign(\gamma^*_i)\}}\right|\D_1 & =  \left.\frac{1+|\{i:\widehat{\sign}_i\neq\sign(\gamma^*_i), |W_i|\ge T_q\}|}{1+\sum_{i:\widehat{\sign}_i\neq\sign(\gamma^*_i), |W_i|\ge T_q}1\{W_i<0\}}-1\right|\D_1,\nonumber\\
    & = \left.\frac{1+J}{1+B_{(1)}+B_{(2)}+\cdots+B_{(J)}}-1\right|\D_1,\label{eq: tran bound}
\end{align}
for $0\le J\le m^*$ satisfying
\begin{align*}
    |W_{(1)}|\ge|W_{(2)}|\ge\cdots\ge|W_{(J)}|\ge T_q>|W_{(J+1)}|\ge\cdots\ge|W_{(m^*)}|.
\end{align*}
In other words, $J = \argmax_{k\le m^*}\{|W_{(k)}|\ge T_q\}$ is the maximum subscript of $k$ satisfying $|W_{(k)}|\ge T_q$. With the above change of notations, Lemma 1 in \cite{barber2019knockoff} can be applied directly to Equation \eqref{eq: tran bound}.

\begin{lemma}[Lemma 1 in \cite{barber2019knockoff}]
    \label{lemma: barber} 
    Let $\{B_{(i)}\}_{i=1}^{m^*}$ be independent Bernoulli random variables with $\Prob[B_{(i)}=1]=\rho_i$ for each $i$. Let $\rho>0$ satisfy $\rho\le \min_{i}\{\rho_i\}$. Let $J$ be a stopping time in inverse time on the filtration $\{\mathcal{F}_j\}_{j=1}^{m^*}$  defined as
    \begin{align*}
        \mathcal{F}_j = \sigma\left(\left\{\sum_{i=1}^jB_{(j)}, B_{(j+1)}, \cdots, B_{(m^*)}\right\}\right).
    \end{align*}
    Then
    \begin{align*}
        \Expect\left[\frac{1+J}{1+B_{(1)}+B_{(2)}+\cdots+B_{(J)}}\right]\le \rho^{-1}.
    \end{align*}
\end{lemma}

\begin{remark}
    The requirement that $\{B_{(i)}\}_{i=1}^{m^*}$ are independent random variables in applying this lemma to Equation \eqref{eq: tran bound} is ensured by Lemma \ref{lemma: est of sign} that $\sign(W_i)$ are independent from each other conditional on $\D_1$ as mentioned above.
\end{remark}

Applying Lemma \ref{lemma: barber} and Lemma \ref{lemma: est of sign} to Equation \eqref{eq: tran bound}, there holds
\begin{align*}
    \Expect\left[\left.\frac{\sum_{i}1\{W_i\ge T_q, \widehat{\sign}_i\neq\sign(\gamma^*_i)\}}{1+\sum_{i}1\{W_i\le -T_q, \widehat{\sign}_i\neq\sign(\gamma^*_i)\}}\right|\D_1\right] & = \Expect\left[\left.\frac{1+J}{1+B_{(1)}+B_{(2)}+\cdots+B_{(J)}}-1\right|\D_1\right],\\
    &\le \min\left(2, f(\nu)^{-1}\right) - 1 = \min\left(1, f(\nu)^{-1}-1\right).
\end{align*}
Take 
\begin{align}
    \alpha(\nu) = f(\nu)^{-1}-1,\label{def: alpha(nu)}
\end{align}
where $f(\nu)$ first appears in Lemma \ref{lemma: est of sign} and we achieve the desired result.
\end{proof}

\section{Proof of Lemma \ref{lemma: distribution zeta}}
\label{sec: proof zeta}

\begin{proof}
With direct calculation, it can be shown that the KKT conditions that Equation \eqref{eq:stage1} and \eqref{eq:stage2} should satisfy is
\begin{subequations}
        \begin{align*}
        %&0 = - \left(\frac{X_1^TX_1}{n} + \frac{D^TD}{\nu}\right) \beta(\lambda) + \frac{D^T}{\nu} \gamma(\lambda) + \left\{\frac{X_1^TX_1}{n} \beta^* + \frac{X_1^T}{n} \varepsilon_1 \right\}, \label{eq: beta kkt} \\
        &\lambda\rho(\lambda) + \frac{\gamma(\lambda)}{\nu}= \frac{D\beta(\lambda)}{\nu},\\
        &\lambda\tilde{\rho}(\lambda) + \frac{\tilde{\gamma}(\lambda)}{\nu} = \frac{D\beta(\lambda)}{\nu} +\tilde{A}_{\gamma_2}^T\tilde{y}_2,
    \end{align*}
\end{subequations}
where $\rho(\lambda) \in \partial \|\gamma(\lambda)\|_1$, $\tilde{\rho}(\lambda) \in \partial \|\tilde{\gamma}(\lambda)\|_1$. In the following, we will focus on the distribution of $\zeta:=\tilde{A}_{\gamma_2}^T\tilde{y}_2\in\R^m$ and prove Equation \eqref{eq: zeta dis}.
 
By Proposition \ref{prop: structure of tagamma}, there holds
     $\tilde{A}^T_{\gamma_2,1}\frac{X_2}{\sqrt{n_2}} = \diag(\vecs)D$. Therefore, there further holds
    \begin{align*}
    \zeta=\tilde{A}_{\gamma_2}^T\tilde y_2=
        -\diag(\vecs)\gamma^*+\frac{\tilde{A}^T_{\gamma_2, 1}}{\sqrt{n_2}}\varepsilon_2.
    \end{align*}
    Therefore, the mean of $\zeta$ is $ -\diag(\vecs)\gamma^*$, and the covariance of $\zeta$ is $\frac{\sigma^2}{n_2}\tilde{A}^T_{\gamma_2,1}\tilde{A}_{\gamma_2,1}$. By Proposition \ref{prop: structure of tagamma}, there holds $\tilde{A}^T_{\gamma_2,1}\tilde{A}_{\gamma_2,1}=\diag(\vecs)(2I_m-\diag(\vecs)\nu)$. 
\end{proof}

\section{Proof of Lemma \ref{lemma: est of sign}}
\label{sec: proof lemma}

In this section, we will decouple Lemma \ref{lemma: est of sign} into the following two lemmas and prove the lemmas separately.

\begin{lemma}
    \label{lemma: est of sign easy}
    Conditional on $\D_1$, $\sign(W_i)$ are independent random variables. Furthermore, 
    for $i\in \{\widehat{\sign}_i \neq \sign(\gamma^*_i)\}$, there holds
    \begin{align*}
        \Prob[W_i<0] \ge \frac{1}{2}.
    \end{align*}
\end{lemma}

\begin{lemma}
    \label{lemma: est of sign hard}
    Conditional on $\D_1$,
    for $i\in \{\widehat{\sign}_i \neq \sign(\gamma^*_i)\}$, there holds
    \begin{align*}
        \Prob[W_i<0] \ge  f(\nu),
    \end{align*}
    where $f(\nu)$ is an increasing function of $\nu$ defined in Equation \eqref{def: f(nu)} s.t. $\lim_{\nu\to\infty}f(\nu) = 1$.
\end{lemma}

The above two lemmas together clearly gives Lemma \ref{lemma: est of sign}.

\subsection{Proof of Lemma \ref{lemma: est of sign easy}}

\begin{proof}

    We will first show that conditional on $\D_1$, $\sign(W_i)$ are independent random variables. By Lemma \ref{lemma: distribution zeta}, $\zeta$ consists of independent random variables. Since $\beta(\lambda)$ is determined by $\D_1$, conditional on $\D_1$, $\sign(W_i) = \sign(Z_i-\tilde{Z}_i)$ is determined by $\zeta_i$ for all $i$, as shown in Equation \eqref{eq: kkts}. Therefore, conditional on $\D_1$, $\sign(W_i)$ are independent random variables.

    We will then show that $\Prob[W_i<0] \ge \frac{1}{2}$ for $i\in \{\widehat{\sign}_i \neq \sign(\gamma^*_i)\}$ conditional on $\D_1$. The following statements will all be conditional on $\D_1$, and the term ``conditional on $\D_1$'' will be omitted for simplicity. Our main focus below is showing that $\{W_i\ge 0\}\subseteq\{r_i\zeta_i\le 0\}$ for $i\in\{\widehat{\sign}_i \neq \sign(\gamma^*_i)\}$ where $r = \widehat{\sign}$ is defined in Equation \eqref{def: r}.

    We first present why the above statement is sufficient for Lemma \ref{lemma: est of sign easy}. With the above statement, there holds for $i\in\{\widehat{\sign}_i \neq \sign(\gamma^*_i)\}$ that
    \begin{align*}
        \Prob[W_i\ge 0]\le \Prob[r_i\zeta_i\le 0] = \Prob[r_i\zeta_i+\vecs_ir_i\gamma_i^*\le\vecs_ir_i\gamma_i^*],
    \end{align*}
    where by Lemma \ref{lemma: distribution zeta}, $r_i\zeta_i+\vecs_ir_i\gamma_i^*$ follows the normal distribution with mean zero. Moreover, for $i\in\{\widehat{\sign}_i \neq \sign(\gamma^*_i)\}$, there holds $\vecs_ir_i \gamma^*_i\le 0$. Therefore $\Prob[W_i\ge 0]\le \frac{1}{2}$ and $\Prob[W_i<0]\ge \frac{1}{2}$.

    Now we start to prove the statement $\{W_i\ge 0\}\subseteq\{r_i\zeta_i\le 0\}$  by contradiction. Suppose there holds $W_i\ge 0$ and $r_i\zeta_i> 0$, then by definition of $W$ in Equation \eqref{eq: def_w}, $Z_i\ge \tilde{Z}_i$. Therefore, by definitions of $Z_i$, $\tilde{Z}_i$ and continuity of $\beta(\lambda)$, there holds $\gamma_i(Z_i) = \tilde{\gamma}_i(Z_i) = 0$ and $r_i\rho_i(Z_i) = 1$. From $r_i\times$\eqref{eq: feature sig} $-$ $r_i\times$\eqref{eq: knockoff sig}, there further holds
    \begin{align*}
        Z_i r_i\tilde{\rho}_i(Z_i) = Z_i r_i\rho_i(Z_i) + r_i\zeta_i > Z_i r_i\rho_i(Z_i) = Z_i,
    \end{align*}
    which leads to contradiction, as $|r_i\tilde{\rho}_i(Z_i)|$ should not be larger than 1.
\end{proof}

\subsection{Proof of Lemma \ref{lemma: est of sign hard}}

In this section, we will prove Lemma \ref{lemma: est of sign hard}. As a reminder, in Section \ref{sec: analysis}, we show that $\sign(W)_i$ is determined by $\zeta_i$ conditional on $\D_1$. In this section, we will first present how $\zeta_i$ determines $\sign(W_i)$ conditional on $\D_1$ precisely, that $\sign(W_i)$ is determined by whether $\zeta_i$ falls in a particular interval. Then we give an estimation on the boundary of such a interval to prove Lemma \ref{lemma: est of sign hard}.

To be specific, we first present the following lemma on how $\zeta_i$ determines $\sign(W_i)$ conditional on $\D_1$.

\begin{lemma}
    \label{lemma: interval}
    Conditional on $\D_1$, there holds
    \begin{align*}
        \{W_i\ge 0\}\subseteq \{-c(\nu)_i\le r_i \zeta_i\le 0\}, \{W_i\le 0\}\subseteq \{r_i \zeta_i\ge 0\}\cup\{r_i \zeta_i\le -c(\nu)_i\}, 
    \end{align*}
    where $c(\nu)_i = \inf_{\lambda\ge Z_i}(r_i[D\beta(\lambda)]_i\nu^{-1}+\lambda)$.
\end{lemma}

Then we give the following partial upper bound on $c(\nu)$. 
\begin{lemma}
    \label{lemma: est c(nu)}
    Conditional on $\D_1$ such that $\xi:= \frac{X_1^T\varepsilon_1}{n_1}$ satisfies 
    $\|\xi\|_\infty<\kappa\sigma\frac{1}{\sqrt{n_1}}$ for a manually controlled parameter $\kappa>0$ to be determined later. Further suppose that $\nu>4C_XC_D^2$,
    where 
    \begin{align*}
        C_X:=\left\|\left(\frac{X_1^T X_1}{n_1}\right)^{-1}\right\|_\infty,\ C_D := \max\{\|D\|_\infty, \|D^T\|_\infty\}.
    \end{align*}
   Then for any $i\in \{\widehat{\sign}_i \neq \sign(\gamma^*_i)\}$, there holds
   \begin{align*}
        c(\nu)_i\le \frac{8}{\nu}\kappa C_XC_D\sigma\frac{1}{\sqrt{n_1}}.
   \end{align*}
\end{lemma}
\begin{remark}
    The decreasing upper bound with the increase of $\nu$ in this lemma can be understood as a result of the improving $\nu$-incoherence conditions \eqref{incoherence} of Split LASSO.% as presented in the proof.
\end{remark}

With Lemma \ref{lemma: est c(nu)}, we can finally prove Lemma \ref{lemma: est of sign hard}, with $f(\nu)$ defined as
\begin{align}
f(\nu) = \left\{
    \begin{array}{ccl}
        \max\left(0, 1-\frac{1}{\sqrt{\nu}}\frac{\sqrt{n_2}}{\sqrt{n_1}}\frac{8\kappa C_XC_D}{\sqrt{2\pi \vecs_i\nu(2-\vecs_i\nu)}}\right)   &   & \mbox{if }\|\xi\|_\infty<\kappa\sigma\frac{1}{\sqrt{n_1}},\ \ \nu>4C_XC_D^2,\\
        0   &   & \mbox{otherwise},\label{def: f(nu)}
    \end{array} \right.
\end{align}
where the respective notations are defined in Lemma \ref{lemma: est c(nu)}. In Equation \eqref{def: f(nu)}, take $\kappa\to\infty$ and $\frac{\kappa}{\sqrt{\nu}}\to 0$ (which means $\nu\to\infty$), there holds $f(\nu)\to 1$.

\subsubsection{Proof of Lemma \ref{lemma: interval}}
\begin{proof}
    We first show that the infimum in the definition of $c(\nu)_i$ can be achieved for some $\lambda\ge Z_i$. This is due to the fact that $r_i[D\beta(\lambda)]_i\nu^{-1}+\lambda$ is continuous with respect to $\lambda$, and the fact that the infimum is not reached when $\lambda\to \infty$, since $r_i [D\beta(\lambda)]_i\nu^{-1}+\lambda\to\infty$ for any $\nu$ as $\lambda\to \infty$. Then by the extreme value theorem, the infimum can be reached for some $\lambda\ge Z_i$. Therefore, it is proper to write $c(\nu)_i = \min_{\lambda\ge Z_i}(r_i[D\beta(\lambda)]_i\nu^{-1}+\lambda)$.
    
    %\paragraph*{2) $\{r_i \zeta_i<-c(\nu)_i\}\cup\{r_i\zeta_i>0\}\subseteq \{\Ws_i>0\}$.} ~ \\
    Then we will prove the lemma by showing that $\{r_i\zeta_i>0\}\cap  \{W_i\ge 0\} = \emptyset$, $\{r_i\zeta_i<-c(\nu)_i\}\cap  \{W_i\ge 0\} = \emptyset$ and $\{-c(\nu)_i<r_i\zeta_i<0\}\cap  \{W_i\le 0\} = \emptyset$.
    
    1) $\{r_i\zeta_i>0\}\cap  \{W_i\ge 0\} = \emptyset$: Suppose $r_i\zeta_i>0$. Then for $\lambda=Z_i$, by definition, there holds $r_i\rho_i(\lambda)=1$ and $r_i\gamma(\lambda)_i= 0$. From \eqref{eq: feature sig} - \eqref{eq: knockoff sig} and the fact that $|r_i\tilde{\rho}_i(\lambda)|\le 1$, there further holds
    \begin{align}
        0 \le \lambda r_i\rho_i(\lambda)-\lambda r_i\tilde{\rho}_i(\lambda) = & r_i\frac{\tilde{\gamma}_i(\lambda)}{\nu} -  r_i\frac{\gamma_i(\lambda)}{\nu} - r_i\zeta_i,\nonumber\\
        = & r_i\frac{\tilde{\gamma}_i(\lambda)}{\nu} - r_i\zeta_i.\nonumber
    \end{align}
    Therefore $r_i\tilde{\gamma}_i(\lambda)\ge r_i\zeta_i\nu>0$. By the continuity of $\tilde{\gamma}_i(\lambda)$, there further holds $\tilde{Z}_i>Z_i$ and therefore $W_i<0$.
    
    2) $\{r_i\zeta_i<-c(\nu)_i\}\cap  \{W_i\ge 0\} = \emptyset$: Suppose $r_i\zeta_i<-c(\nu)_i$. Take $t=\argmin_{\lambda\ge Z_i}(r_i[D\beta(\lambda)]_i\nu^{-1}+\lambda)$, then from Equation \eqref{eq: knockoff sig} there holds
    \begin{align*}
        r_it\tilde{\rho}_i(t)+r_i\frac{\tilde{\gamma}_i(t)}{\nu} = & r_i[D\beta(t)]_i\nu^{-1}+r_i\zeta_i ,\\
        = &  -t + c(\nu)_i + r_i\zeta_i< -t.
    \end{align*}
    Therefore $r_i\tilde{\gamma}_i(t) = \nu(c(\nu)_i + r_i\zeta_i)<0$. By the continuity of $\tilde{\gamma}_i(\lambda)$, $\tilde{Z}_i>t\ge Z_i$, which means $W_i<0$.

    3) $\{-c(\nu)_i<r_i\zeta_i<0\}\cap  \{W_i\le 0\} = \emptyset$: Suppose $-c(\nu)_i<r_i\zeta_i<0$ and $W\le 0$, then further holds $\tilde{Z}_i\ge Z_i$. Take  $t = \tilde{Z}_i \ge Z_i$. If $r_i \neq \sign(\tilde{\rho}_i(t))$, then $r_i\tilde\rho_i(t)=-1$ and $r_i\tilde\gamma(t)_i= 0$. 
    From  Equation \eqref{eq: knockoff sig}, there further holds
    \begin{align}
        -t= r_it\tilde{\rho}_i(t)+r_i\frac{\tilde{\gamma}_i(t)}{\nu} = & r_i[D\beta(t)]_i\nu^{-1}+r_i\zeta_i ,\\
        > & r_i[D\beta(t)]_i\nu^{-1} - \min_{\lambda\ge Z_i}(r_i[D\beta(\lambda)]_i\nu^{-1}+\lambda)\ge -t,
    \end{align}
    which is a contradiction. Therefore  $r_i = \sign(\tilde{\rho}_i(t))$ and consequently $r_i\tilde\rho_i(t)=1$, $r_i\tilde\gamma(t)_i= 0$. From \eqref{eq: feature sig} - \eqref{eq: knockoff sig}, there holds 
    \begin{align}
        r_it\rho_i(t)+ r_i\frac{\gamma_i(t)}{\nu}= & r_it \tilde{\rho}_i(t) +r_i \frac{\tilde\gamma_i(t)}{\nu}  - r_i\zeta_i> t.\nonumber
    \end{align}
    Therefore $r_i\gamma_i(t) \ge  - r_i\zeta_i>0$. By the continuity of $\gamma_i(t)$, there further holds $\tilde{Z}_i<Z_i$, which is a contradiction. This ends the proof.
\end{proof}

\subsubsection{Proof of Lemma \ref{lemma: est c(nu)}}

\begin{proof}
    We prove Lemma \ref{lemma: est c(nu)} using the following property induced from the incoherence conditions of Split LASSO, that for any $\xi$ satisfying $\|\xi\|_\infty<\kappa\sigma\frac{1}{\sqrt{n_1}}$ and $\nu>4C_XC_D^2$, for any $\lambda> \frac{4}{\nu}\kappa C_XC_D\sigma\frac{1}{\sqrt{n_1}}$, there holds if $\sign(\gamma_i(\lambda))\neq0$, then $\sign(\gamma_i(\lambda)) = \sign(\gamma^*_i)$.

    We first show that such a result is sufficient for Lemma \ref{lemma: est c(nu)}. Suppose such a result holds, then for conditional on $\D_1$ where $\xi$ satisfies $\|\xi\|_\infty<\kappa\sigma\frac{1}{\sqrt{n_1}}$ and $\nu>4C_XC_D^2$, there holds for $i\in \{i: \widehat\sign_i \neq \sign(\gamma^*_i)\}$ that $Z_i\le \frac{4}{\nu}\kappa C_XC_D\sigma\frac{1}{\sqrt{n_1}}$. Therefore, by \eqref{eq: feature sig} - \eqref{eq: knockoff sig}, for $r_i\zeta_i<-\frac{8}{\nu}\kappa C_XC_D\sigma\frac{1}{\sqrt{n_1}}$, there holds
    \begin{align*}
        r_iZ_i\tilde{\rho}_i(Z_i)+ r_i\frac{\tilde{\gamma}_i(Z_i)}{\nu}&=r_iZ_i\rho_i(Z_i)+r_i\zeta_i<-Z_i.
    \end{align*}
    Therefore, $\tilde{\gamma}_i(Z_i)= r_i\zeta_i+2Z_i<0$. By the continuity of $\tilde{\gamma}_i(\lambda)$, $\tilde{Z}_i>Z_i$, which means $W_i<0$. Therefore $c(\nu)_i\le \frac{8}{\nu}\kappa C_XC_D\sigma\frac{1}{\sqrt{n_1}}$.

    Now we proceed to show the property induced from the incoherence conditions of Split LASSO, that for any $\xi$ satisfying $\|\xi\|_\infty<\kappa\sigma\frac{1}{\sqrt{n_1}}$ and $\nu>4C_XC_D^2$, for any $\lambda> \frac{4}{\nu}\kappa C_XC_D\sigma\frac{1}{\sqrt{n_1}}$, there holds if $\sign(\gamma_i(\lambda))\neq0$, then $\sign(\gamma_i(\lambda)) = \sign(\gamma^*_i)$. The KKT condition that $\beta(\lambda)$ from Equation \eqref{eq: beta(lambda)} should satisfy is: 
    \begin{align}
        0 &= - \left(\frac{X_1^T X_1}{n_1} + \frac{D^T D}{\nu}\right) \beta(\lambda) + \frac{D^T}{\nu} \gamma(\lambda) + \left\{\frac{X_1^T X_1}{n_1} \beta^* + \frac{X_1^T}{n_1} \varepsilon_1 \right\}. \label{eq: beta kkt}
    \end{align}
    By $D\left(\frac{X_1^T X_1}{n_1}+\frac{D^T D}{\nu}\right)^{-1}\times$\eqref{eq: beta kkt}$+\nu\times$\eqref{eq: feature sig}, and the fact that $\gamma^* = D\beta^*$, there holds
    \begin{align*}
        \lambda\nu\rho(\lambda)=&-\gamma(\lambda)+D\left(\frac{X_1^T X_1}{n_1}+\frac{D^T D}{\nu}\right)^{-1}D^T\frac{\gamma(\lambda)}{\nu}+\cdots\\
        &\cdots+D\left(\frac{X_1^T X_1}{n_1}+\frac{D^T D}{\nu}\right)^{-1}\left(\frac{X_1^T X_1}{n_1}+\frac{D^T D}{\nu}-\frac{D^T D}{\nu}\right)\beta^*+D\left(\frac{X_1^T X_1}{n_1}+\frac{D^T D}{\nu}\right)^{-1}\xi,\\
        &=-H_{\nu}(\gamma(\lambda)-\gamma^*)+\omega,
    \end{align*}
    where 
    \begin{align*}
        H_\nu&=I_m-\frac{1}{\nu}D\left(\frac{X_1^T X_1}{n_1}+\frac{D^T D}{\nu}\right)^{-1}D^T,\\
        \omega&=D\left(\frac{X_1^T X_1}{n_1}+\frac{D^T D}{\nu}\right)^{-1}\xi.
    \end{align*}
    Let $H_{\nu}^{i, i}$ be the $(i, i)-$th element of $H_\nu$, $H_{\nu}^{i, -i}$ be the $i-$th row of $H_\nu$ without the $i-$th element, $H_{\nu}^{-i, i}$ be the $i-$th column of $H_\nu$ without the $i-$th element, and $H_{\nu}^{-i, -i}$ be the submatrix of $H_\nu$ obtained by deleting the $i-$th row and $i-$th column. There holds
    \begin{align}
    	\lambda \nu
    	\begin{bmatrix}
        	\rho_{-i}(\lambda)\\
        	\rho_{i}(\lambda)
    	\end{bmatrix}
    	=-
    	\begin{bmatrix}
    	    H_\nu^{-i, -i} & H_\nu^{-i, i}\\
    	    H_\nu^{i, -i} & H_\nu^{i, i}
    	\end{bmatrix}
    	\begin{bmatrix}
    	    \gamma_{-i}(\lambda) - \gamma^*_{-i}\\
    	    \gamma_i(\lambda) - \gamma^*_{i}
    	\end{bmatrix}
    	+
    	\begin{bmatrix}
    	    \omega_{-i}\\
    	    \omega_{i}
  	  	\end{bmatrix},\nonumber
\end{align}
which means 
\begin{subequations}
    \begin{align}
        \lambda\nu \rho_{-i}(\lambda) = & - H_\nu^{-i, -i}(\gamma_{-i}(\lambda) - \gamma^*_{-i})-H_\nu^{-i, i}(\gamma_i(\lambda) - \gamma^*_{i}) + \omega_{-i}\label{eq1: irr_b},\\
        \lambda\nu \rho_{i}(\lambda) = & -  H_\nu^{i, -i}(\gamma_{-i}(\lambda) - \gamma^*_{-i}) -H_\nu^{i, i}(\gamma_i(\lambda) - \gamma^*_{i})  + \omega_{i}\label{eq2: irr_b}.
    \end{align}
\end{subequations}
Solve $\gamma_{-i}(\lambda) - \gamma^*_{-i}$ from Equation \eqref{eq1: irr_b}, and plug the solution into Equation \eqref{eq2: irr_b}, there holds
\begin{align}
	\lambda\nu\rho_{i}(\lambda)&=H_\nu^{-i, i}[H_\nu^{-i, -i}]^{-1}\lambda\nu\rho_{-i}(\lambda)-H_\nu^{-i, i}[H_\nu^{-i, -i}]^{-1}\omega_{-i}+\omega_i+\cdots\nonumber\\
	&\ \ \ \ \cdots-(H_{i, i}-H_\nu^{-i, i}[H_\nu^{-i, -i}]^{-1}H_\nu^{i, -i})(\gamma_i(\lambda) - \gamma^*_{i}).\label{eq: key for estimation}
\end{align}
Note that $(H_{i, i}-H_\nu^{-i, i}[H_\nu^{-i, -i}]^{-1}H_\nu^{i, -i})$ is the Schur complement of $H_\nu^{-i, -i}$ in the positive semi-definite matrix $H_{\nu}$, therefore
$(H_{i, i}-H_\nu^{-i, i}[H_\nu^{-i, -i}]^{-1}H_\nu^{i, -i})\ge0$.

Now we start to prove the result by contradiction. If there exists $i$ and $\lambda> \frac{4}{\nu}\kappa C_XC_D\sigma\frac{1}{\sqrt{n_1}}$, such that $\sign(\gamma_i(\lambda))\neq 0$ and $\sign(\gamma_i(\lambda)) \neq \sign(\gamma^*_i)$, then there holds $\sign(\gamma_i(\lambda) - \gamma^*_{i}) = \sign(\rho_{i}(\lambda))$. Combining with Equation \eqref{eq: key for estimation}, there further holds
\begin{align}
	1&<\left|\rho_{i}(\lambda)+\frac{1}{\lambda\nu}(H_{i, i}-H_\nu^{-i, i}[H_\nu^{-i, -i}]^{-1}H_\nu^{i, -i})(\gamma_i(\lambda) - \gamma^*_{i})\right|,\nonumber\\
	&=\left|H_\nu^{-i, i}[H_\nu^{-i, -i}]^{-1}\rho_{-i}(\lambda)-\frac{1}{\lambda\nu}(H_\nu^{-i, i}[H_\nu^{-i, -i}]^{-1}\omega_{-i}+\omega_i)\right|.\label{ineq: tag1}
\end{align}
Now we start to estimate $\|H_\nu^{-i, i}[H_\nu^{-i, -i}]^{-1}\|_\infty$. For $\nu>4C_XC_D^2$
\begin{align}
    \left\|\left(\frac{X_1^T X_1}{n_1}+\frac{D^T D}{\nu}\right)^{-1}\right\|_\infty & \le \left\|\left(\frac{X_1^T X_1}{n_1}\right)^{-1}\right\|_\infty\left\|\left[I_m+\left(\frac{X_1^T X_1}{n_1}\right)^{-1}\frac{D^T D}{\nu}\right]^{-1}\right\|_\infty,\nonumber\\
    &\le C_X\sum_{i=0}^\infty\left\|\left[\left(\frac{X_1^T X_1}{n_1}\right)^{-1}\frac{D^T D}{\nu}\right]\right\|^k_\infty,\nonumber\\
    &\le C_X\sum_{i=0}^\infty\left[\frac{C_XC_D^2}{\nu}\right]^k\le \frac{4}{3}C_X.\label{eq: est3}
\end{align}
Let $R(\nu) := D\left(\frac{X_1^T X_1}{n_1}+\frac{D^T D}{\nu}\right)^{-1}D^T$. Then there holds
\begin{align}
	\|H_\nu^{-i, i}\|_\infty=\frac{1}{\nu}\|R(\nu)^{-i, i}\|_\infty\le  \frac{1}{\nu}\|R(\nu)\|_\infty\le \frac{4}{3\nu}C_XC_D^2.\label{eq: est1}
\end{align}
Moreover, for $\nu>4C_XC_D^2$, there holds
\begin{align}
    \left\|\left[H_\nu^{-i, -i}\right]^{-1}\right\|_\infty & = \left\|\left[I_{m-1}-\nu^{-1}R(\nu)^{-i, -i}\right]^{-1}\right\|_\infty,\nonumber\\
    & = \left\|\sum_{k=0}^\infty\nu^{-k}\left[R(\nu)^{-i, -i}\right]^k\right\|_\infty,\nonumber\\
    &\le \sum_{k=0}^\infty\nu^{-k}\left\|\left[R(\nu)^{-i, -i}\right]\right\|_\infty^k\le \frac{1}{1-\frac{4C_X C_D^2}{3\nu}}\le \frac{3}{2}.\label{eq: est2}
\end{align}
Combining Equation \eqref{eq: est1} and \eqref{eq: est2} together, for $\nu>4C_XC_D^2$, there holds
\begin{align}
	\|H_\nu^{-i, i}[H_\nu^{-i, -i}]^{-1}\|_\infty\le \frac{2}{\nu}C_XC_D^2<\frac{1}{2}\label{eq: est4}.
\end{align}
Therefore, there holds
\begin{align}
	|H_\nu^{-i, i}[H_\nu^{-i, -i}]^{-1}\rho_{-i}(\lambda)|\le \|H_\nu^{-i, i}[H_\nu^{-i, -i}]^{-1}\|_\infty\|\rho_{-i}(\lambda)\|_\infty<\frac{1}{2}.\label{ineq: tag2}
\end{align}
Moreover, by Equation \eqref{eq: est3}, there holds
\begin{align*}
	\|\omega\|_\infty &= \left\|D\left(\frac{X_1^T X_1}{n_1}+\frac{D^T D}{\nu}\right)^{-1}\xi\right\|_\infty,\\
	&\le \frac{4}{3}C_XC_D\left\|\xi\right\|_\infty\le \frac{4}{3}\kappa C_XC_D\sigma\frac{1}{\sqrt{n_1}}.
\end{align*}
Combining with Equation \eqref{eq: est4}, there holds
\begin{align}
	\|H_\nu^{-i, i}[H_\nu^{-i, -i}]^{-1}\omega_{-i}+\omega_i\|_\infty
	&<2\kappa C_XC_D\sigma\frac{1}{\sqrt{n_1}}.\label{ineq: tag3}
\end{align}
Combining Equation \eqref{ineq: tag1}, \eqref{ineq: tag2} and \eqref{ineq: tag3} together, there holds
\begin{align*}
	\frac{1}{2}<\frac{1}{\lambda\nu} 2\kappa C_X^2C_D\sigma\frac{1}{\sqrt{n_1}}.
\end{align*}
This contradicts $\lambda> \frac{4}{\nu}\kappa C_XC_D\sigma\frac{1}{\sqrt{n_1}}$ and ends the proof.
\end{proof}

\subsubsection{Proof of Lemma \ref{lemma: est of sign hard}}

\begin{proof}
    First of all, clearly there holds
    $
        \Prob[W_i<0]\ge 0
    $
     under any situation.
    
    From Lemma \ref{lemma: interval} and \ref{lemma: est c(nu)}, conditional on $\D_1$ such that 
    $\|\xi\|_\infty<\kappa\sigma\frac{1}{\sqrt{n_1}}$, for $\nu>4C_XC_D^2$, there holds $$\{W_i\ge 0\}\subseteq \left\{- \frac{8}{\nu}\kappa C_XC_D\sigma\frac{1}{\sqrt{n_1}}\le r_i\zeta_i\le 0\right\}.$$
    Note that $-a<r_i \zeta_i<0\iff -a+\vecs_ir_i\gamma^*_i<r_i( \zeta_i+\vecs_i\gamma^*_i)<\vecs_ir_i\gamma^*_i$ for $a>0$, where
    \begin{align*}
        \zeta+\vecs\gamma^*\sim \mathcal{N}\left(0_m, \frac{1}{n_2}\diag(\vecs)(2I_m-\diag(\vecs)\nu)\sigma^2\right),
    \end{align*}
    by Lemma \ref{lemma: distribution zeta}. For $i\in\{\widehat{\sign}_i\neq\gamma^*_i\} = \{r_i\neq\gamma^*_i\}$, there further holds:
    \begin{enumerate}
        \item if $r_i=1$, then $-a<r_i \zeta_i<0\iff -a+\vecs_i\gamma^*_i<\zeta_i+\vecs_i\gamma^*_i<\vecs_i\gamma^*_i$ where $\vecs_i\gamma^*_i\le 0$; in this case, the fact that $\zeta_i+\vecs_i\gamma^*_i$ is mean zero further suggests $\Prob[-a+\vecs_i\gamma^*_i<\zeta_i+\vecs_i\gamma^*_i<\vecs_i\gamma^*_i]\le \Prob[-a<\zeta_i+\vecs_i\gamma^*_i<0]$;
        \item if $r_i=-1$, then $-a<r_i \zeta_i<0\iff a+\vecs_i\gamma^*_i>\zeta_i+\vecs_i\gamma^*_i>\vecs_i\gamma^*_i$ where $\vecs_i\gamma^*_i\ge 0$; in this case, the fact that $\zeta_i+\vecs_i\gamma^*_i$ is mean zero further suggests $\Prob[a+\vecs_i\gamma^*_i>\zeta_i+\vecs_i\gamma^*_i>\vecs_i\gamma^*_i]\le \Prob[a>\zeta_i+\vecs_i\gamma^*_i>0]$.
    \end{enumerate}
    Therefore, there holds
    \begin{align*}
        \Prob[W_i\ge 0]&\le \Prob\left[- \frac{8}{\nu}\kappa C_XC_D\sigma\frac{1}{\sqrt{n_1}}<\zeta_i<0\right],\\
        &= \int_{-\frac{8}{\nu}\kappa C_XC_D\sigma\frac{1}{\sqrt{n_1}}}^{0}\frac{\sqrt{n_2}}{\sqrt{2\pi \vecs_i(2-\vecs_i\nu)\sigma^2}}e^{-\frac{n_2x^2}{\vecs_i(2-\vecs_i\nu)\sigma^2}}dx,\\
        &\le\frac{1}{\sqrt{\nu}}\frac{\sqrt{n_2}}{\sqrt{n_1}}\frac{8\kappa C_XC_D}{\sqrt{2\pi \vecs_i\nu(2-\vecs_i\nu)}},
    \end{align*}
    and therefore
    \begin{align*}
        \Prob[W_i< 0] = 1- \Prob[W_i\ge 0]\ge\max\left(0, 1-\frac{1}{\sqrt{\nu}}\frac{\sqrt{n_2}}{\sqrt{n_1}}\frac{8\kappa C_XC_D}{\sqrt{2\pi \vecs_i\nu(2-\vecs_i\nu)}}\right),
    \end{align*}
    which ends the proof.
\end{proof}

\section{Proof of Theorem \ref{theorem: directional fdr without split}}
\label{sec: proof thm late}

%We will prove Theorem \ref{theorem: directional fdr} by combining the estimation on the directional FDR bound on the dominant part $\Lambda(\nu)$, and the estimation on the probability of the Gaussian tail event $\Lambda(\nu)^c$. We will first prove the results for Split Knockoff and modified directional FDR control.

First of all, similar to the proof of Theorem \ref{theorem: directional fdr} in Section \ref{sec: proof thm}, the problem of bounding the $\dfdr$ in Theorem \ref{theorem: directional fdr without split} can be deducted from the following inequality
\begin{align}
    \Expect\left[\frac{\sum_{i}1\{W_i\ge T_q, \widehat{\sign}_i\neq\sign(\gamma^*_i)\}}{1+\sum_{i}1\{W_i\le -T_q, \widehat{\sign}_i\neq\sign(\gamma^*_i)\}}\right]\le g(\nu). \label{eq: target late}
\end{align}

For shorthand notations, we abuse the notations a little and define $\xi = \frac{X^T\varepsilon}{n}$ and $\zeta = \tilde{A}_\gamma^T\tilde{y}$ on the full dataset $\D$. This is different from their definitions in the proof of Theorem \ref{theorem: directional fdr}, where $\xi$ is defined on $\D_1$ (in Lemma \ref{lemma: est of sign hard}) and $\zeta$ (in Equation \eqref{eq: kkts}) is defined on $\D_2$. This immediately leads to the hurdle that $\xi$ and $\zeta$ in their current form are \emph{not independent} from each other. %Since $\xi$ is the sufficient statistics for $\beta(\lambda)$, t
%This further causes the $\dfdr$ control not upper bounded by $q$ in Theorem \ref{theorem: directional fdr without split}.

%Then we will first consider condition on the fixed $\|\xi\|_\infty<\kappa\sigma\frac{1}{\sqrt{n}}$, and the set $\zeta\in \Lambda(\nu)$. %For any fixed $\xi$, $\{|W_i|\}_{i=1}^m$ and $\{r_i\}_{i=1}^m$ are fixed. 

Fortunately, we can still reconstruct the estimation on $\Prob[W_i\ge 0]$ in a similar way with that in Lemma \ref{lemma: est of sign}. However, it is worth to mention that $\sign(W_i)$ are no longer independent from each other conditional on $\D_1$ or equivalently conditional on $\xi$ (the sufficient statistics for $\beta(\lambda)$). The proof of Lemma \ref{lemma: estimation of sign} is given in Section \ref{sec: estimation of sign}.

\begin{lemma}
    \label{lemma: estimation of sign}
    Let $\Lambda(\nu)$ be the set of $m$-dimensional vectors given $\xi$, defined by
    \begin{align*}
        \Lambda(\nu):=\left\{u\in \R^m: \forall i,\ \left|u_i-\E[\zeta_i|\xi]\right|<\theta \sqrt{\var[\zeta_i|\xi]}\right\}.
    \end{align*} 
    where $\theta>0$ is a manually chosen parameter to be determined later.
    Given a $\xi$ satisfying $\|\xi\|_\infty<\kappa\sigma\frac{1}{\sqrt{n}}$ for a manually chosen parameter $\kappa$ to be determined later. Suppose $\nu>\max\left\{4C_XC_D^2, \frac{3}{2}\co1\right\}$, where we abuse the notations slightly and define
    \begin{align*}
        C_X:=\left\|\left(\frac{X^T X}{n}\right)^{-1}\right\|_\infty, C_D := \max\{\|D\|_\infty, \|D^T\|_\infty\}, \co1 = \left\|D\left[\frac{X^TX}{n}\right]^{-1}D^T\right\|_2.
    \end{align*}
    Then for any $i\in\{i: \widehat{\sign}_i\neq \sign(\gamma^*_i)\}$, if $m(2-2\Phi(\theta))<1$, there holds
    \begin{align*}
        \Prob\left.\left(W_i\ge 0\right|\xi, \zeta\in\Lambda(\nu)\right)&<\frac{8\kappa C_XC_D}{\sqrt{\pi}\sqrt{\nu}}\frac{1}{1-m(2-2\Phi(\theta))},
    \end{align*}
    where $\Phi(x)$ is the cumulative distribution function of standard normal distribution.
\end{lemma}

%As mentioned previously, d
Due to the dependence between $\zeta$ and $\xi$, $\sign(W_i)$ are no longer independent from each other conditional on $\xi$, the sufficient statistics for $|W|$ and $\widehat{\sign}$. Fortunately, $\zeta$ and $\xi$ are only weakly dependent on each other, which enables an \emph{almost supermartingale structure}. The details are presented in the following.

For shorthand notations, we rearrange the subscripts of the $\{W_i\}_{i=1}^{m}$ conditional on $\zeta$, such that $|W_{(1)}|\ge|W_{(2)}|\ge\cdots\ge|W_{(m^*)}|$, and $\{(1), (2), \cdots, (m^*)\}=\{\widehat{\sign}_i\neq\sign(\gamma^*_i)\}$, where $m^* := |\{r_i\neq\sign(\gamma^*_i)\}|\le m$. Let $0\le J\le m^*$ be an integer satisfying
\begin{align*}
    |W_{(1)}|\ge|W_{(2)}|\ge\cdots\ge|W_{(J)}|\ge T_q>|W_{(J+1)}|\ge\cdots\ge|W_{(m^*)}|.
\end{align*}
In other words, $J = \argmax_{k\le m^*}\{|W_{(k)}|\ge T_q\}$ is the maximum subscript of $k$ satisfying $|W_{(k)}|\ge T_q$. Let $\{B_{(i)}\}_{i=1}^{m^*}$ be Bernoulli random variables defined by $B_{(i)}=1\{W_{(i)}<0\}$ for all $i$, then 
there holds
\begin{align}
    \frac{\sum_{i}1\{W_i\ge T_q, \widehat{\sign}_i\neq\sign(\gamma^*_i)\}}{1+\sum_{i}1\{W_i\le -T_q, \widehat{\sign}_i\neq\sign(\gamma^*_i)\}} & =  \frac{1+\sum_{i:\widehat{\sign}_i\neq\sign(\gamma^*_i), |W_i|\ge T_q}1\{W_i>0\}}{1+\sum_{i:\widehat{\sign}_i\neq\sign(\gamma^*_i), |W_i|\ge T_q}1\{W_i<0\}}-1,\nonumber\\
    & = \frac{1+J}{1+B_{(1)}+B_{(2)}+\cdots+B_{(J)}}-1,\label{eq: tag without split}
\end{align}
which can be bounded by an \emph{almost supermartingale inequality} after considerable efforts on characterizing the weak dependency between $\zeta$ and $\xi$. In fact, it can be shown that the dependency between $\zeta$ and $\xi$ will become weaker with the increase of $\nu$, which leads to the following lemma.

%We further define $B_i:=1\{W_i<0\}$ for all $i$. Then $B_i\sim\mathrm{Bernoulli}(\rho_i)$ are Bernoulli random variables for all $i$.

%By definition, for fixed $\{\zeta_i\}_{i=m^*+1}^m$, $T_q$ can be viewed as a stopping time for filtration $\{\F_i\}_{i=1}^{m}$, where
%\begin{align*}
    %\F_i = \sigma(\{B_1+B_2+\cdots+B_i, B_{i+1}, \cdots, B_{m}\}).
%\end{align*}

\begin{lemma}
    \label{lemma: main_lemma} 
    Let $\tau_\nu$ be measurement on the correlation of $\{B_{(i)}\}_{i=1}^{m^*}$ defined by
    \begin{align*}
        \tau_\nu = 8\ln(2m)\left(e^{48\theta ^2\co2 \frac{m(m+1)}{\nu}}-1\right)e^{\frac{2\fnu m}{\sqrt{\nu}-\fnu}}.
    \end{align*}
    For $\xi$ satisfying $\|\xi\|_\infty<\kappa\sigma\frac{1}{\sqrt{n}}$, if there further holds $\tau_\nu<1$ and $\nu>\gnu$, then
    \begin{align*}
        \E\left[\left.\frac{1+J}{1+B_{(1)}+B_{(2)}+\cdots+B_{(J)}}\right|\xi, \zeta\in\Lambda(\nu)\right]& \le \frac{\sqrt{\nu}}{\sqrt{\nu}-\fnu}\left(1+\tau_\nu\right),
    \end{align*}
    where $\Lambda(\nu)$ is defined in Lemma \ref{lemma: estimation of sign}.
    
    Definition of symbols:
    \begin{enumerate}
        \item $\kappa$, $\theta $ are two manually chosen parameters to be determined later in Lemma \ref{lemma: estimation of sign}. $\co1$ is a constant defined in Lemma \ref{lemma: estimation of sign}.
        \item $\co2 = \|R-\diag(R)\|_2$, where $R:=D\left[\frac{X^T X}{n}\right]^{-1} D^T$, and $\diag(R)$ represents the diagonal matrix setting the off-diagonal elements of $R$ to zero.
        \item $\fnu = \frac{8\kappa C_XC_D}{\sqrt{\pi}}\frac{1}{1-m(2-2\Phi(\theta))}e^{36\theta ^2\co2 \frac{m}{\nu}}$.
        \item $\gnu = \max\left\{4C_XC_D^2, \frac{3}{2}\co1, 12\co2, \fnu^2\right\}$.
    \end{enumerate}
\end{lemma}

\begin{remark}
    This lemma gives a conditional upper bound on the right hand side of Equation \eqref{eq: tag without split} based on the measurement $\tau_\nu$ on the correlation of $\{B_{(i)}\}_{i=1}^{m^*}$ which is only dependent on $m$, $\nu$ and some constants. Specially, when $R$ is a diagonal matrix, e.g. $X$ and $D$ are orthogonal, there holds $\co2 = 0$ and $\tau_\nu=0$. %, and we only require $\nu\sim O(1)$ to achieve desired or even better $\dfdr$ control. 
    The worst case, however, requires $\nu\sim O\left(m^2\ln m\ln^2 p\right)$ to make $\tau_\nu<1$, while one can expect intermediate requirements on $\nu$ in intermediate cases. The proof of Lemma \ref{lemma: main_lemma} is given in Section \ref{sec: main_lemma}.
\end{remark}
Now we can proceed to prove Theorem \ref{theorem: directional fdr without split}. 
By Lemma \ref{lemma: main_lemma}, for $\nu>\gnu$, there holds
\begin{align*}
    \Expect\left.\left[\frac{\sum_{i}1\{W_i\ge T_q, r_i\neq\sign(\gamma^*_i)\}}{1+\sum_{i}1\{W_i\le -T_q, r_i\neq\sign(\gamma^*_i)\}}\right|\xi, \zeta\in \Lambda(\nu)\right]& \le \frac{\sqrt{\nu}}{\sqrt{\nu}-\fnu}\left(1+\tau_\nu\right)-1.
\end{align*}
By the definition of $\Lambda(\nu)$ in Lemma \ref{lemma: estimation of sign}, there holds
\begin{align*}
    \Prob[\zeta\in \Lambda(\nu)^C|\xi]\le & \sum_{i=1}^m \Prob\left[\left|\zeta_i-\E[\zeta_i|\xi]\right|<\theta \var[\zeta_i|\xi]\right]
    \le m\left(2-2\Phi(\theta)\right).
\end{align*}
Therefore, for $\|\xi\|_\infty<\kappa\sigma\frac{1}{\sqrt{n}}$, there further holds for the $\mdfdr$ that
\begin{align}
      &\ \ \ \ \E\left [\left.\frac{\sum_{i}1\{W_i\ge T_q, \widehat{\sign}_i\neq\sign(\gamma_i^*)\}}{\sum_{i}1\{W_i\ge T_q\}+q^{-1}}\right|\xi\right]\nonumber\\
      & \le \E\left [\left.\frac{\sum_{i}1\{W_i\ge T_q, \widehat{\sign}_i\neq\sign(\gamma_i^*)\}}{\sum_{i}1\{W_i\ge T_q\}+q^{-1}} \right| \xi, \zeta\in\lambda(\nu)\right]\Prob[\zeta\in \Lambda(\nu)|\xi]+\ldots\nonumber\\
      &\ \ \ \ \ldots + \E\left [\left.\frac{\sum_{i}1\{W_i\ge T_q, \widehat{\sign}_i\neq\sign(\gamma_i^*)\}}{\sum_{i}1\{W_i\ge T_q\}+q^{-1}} \right| \xi, \zeta\in \Lambda(\nu)^c\right]\Prob [\zeta\in \Lambda(\nu)^c|\xi],\nonumber\\
      & \le \E\left [\left.\frac{\sum_{i}1\{W_i\ge T_q, \widehat{\sign}_i\neq\sign(\gamma_i^*)\}}{\sum_{i}1\{W_i\ge T_q\}+q^{-1}} \right| \xi, \zeta\in \Lambda(\nu)\right]+\Prob [\zeta\in \Lambda(\nu)^c|\xi],\nonumber\\
      & \le q\left[\frac{\sqrt{\nu}}{\sqrt{\nu}-\fnu}\left(1+\tau_\nu\right)-1\right]+m\left(2-2\Phi(\theta)\right).\label{eq: con bound}
\end{align}
By definition, $\xi\sim\mathcal{N}\left(0, \frac{X^T X}{n^2}\sigma^2\right)$. Therefore, there exists $u\sim\mathcal{N}(0, I_p)$, such that $\zeta = \left[\frac{X^TX}{n}\right]^{\frac{1}{2}}\frac{\sigma}{\sqrt{n}}u$. Let $\co3 = \left\|\left[\frac{X^TX}{n}\right]^{\frac{1}{2}}\right\|_\infty$, then
\begin{align*}
    \Prob\left[\|\xi\|_\infty>\kappa\sigma\frac{1}{\sqrt{n}}\right]\le &\Prob\left[\|u\|_\infty>\frac{\kappa}{\co3}\right]
    \le p \left(2-2\Phi\left(\frac{\kappa}{\co3}\right)\right).
\end{align*}
Applying the same procedure as in Equation \eqref{eq: con bound}, there holds for the $\mdfdr$,
\begin{align}
    &\ \ \ \ \E\left [\frac{\sum_{i}1\{W_i\ge T_q, \widehat{\sign}_i\neq\sign(\gamma_i^*)\}}{\sum_{i}1\{W_i\ge T_q\}+q^{-1}}\right]\nonumber\\
    &\le \E\left [\left.\frac{\sum_{i}1\{W_i\ge T_q, \widehat{\sign}_i\neq\sign(\gamma_i^*)\}}{\sum_{i}1\{W_i\ge T_q\}+q^{-1}}\right|\xi, \|\xi\|_\infty<\kappa\sigma\frac{1}{\sqrt{n}}\right] + \Prob\left[\|\xi\|_\infty>\kappa\sigma\frac{1}{\sqrt{n}}\right],\nonumber\\
     & \le q\left[\frac{\sqrt{\nu}}{\sqrt{\nu}-\fnu}\left(1+\tau_\nu\right)-1\right]+m\left(2-2\Phi(\theta)\right)+p \left(2-2\Phi\left(\frac{\kappa}{\co3}\right)\right),\nonumber
\end{align}
where as a reminder $\co3 = \left\|\left[\frac{X^TX}{n}\right]^{\frac{1}{2}}\right\|_\infty$, and other constants are defined in Lemma \ref{lemma: main_lemma}. The same procedure can be applied to $\dfdr$ and the same result can be reached. Define $g(\nu)$ by
\begin{align}
    g(\nu)=:\left\{
    \begin{array}{cll}
         & \left[\frac{\sqrt{\nu}}{\sqrt{\nu}-\fnu}\left(1+\tau_\nu\right)-1\right]+q^{-1}m\left(2-2\Phi(\theta)\right)+q^{-1}p \left(2-2\Phi\left(\frac{\kappa}{\co3}\right)\right),&  \mbox{ if } \nu>\gnu, \tau_\nu<1;\\
         & q^{-1}, & \mbox{ otherwise.}
    \end{array}
    \right.
    \label{def: g(nu)}
\end{align}
and Theorem \ref{theorem: directional fdr without split} can be achieved.
\qed

\begin{remark}
    In Equation \eqref{def: g(nu)}, take $\theta\to\infty$, $\kappa\to\infty$, $\frac{\theta}{\sqrt{\nu}}\to 0$, and $\frac{\kappa}{\sqrt{\nu}}\to 0$ (which means $\nu\to \infty$), there holds $\tau_\nu\to 0$, $\frac{\fnu}{\sqrt{\nu}}\to 0$ and $g(\nu)\to 0$. Meanwhile, it is sufficient to let $\theta\sim O(\ln(m))$, $\kappa\sim O(\ln(p))$, and $\nu\sim O(m^2\ln m\ln^2 p)$ to make $g(\nu)\le 1$.
\end{remark}

\subsection{Proof of Lemma \ref{lemma: estimation of sign}}

\label{sec: estimation of sign}

In this section, we will prove Lemma \ref{lemma: estimation of sign} for Split Knockoffs without sample splitting. The proof will make use of the following lemma, which is a copy of Lemma \ref{lemma: est c(nu)} on the whole dataset $\D$ instead of $\D_1$.

\begin{lemma}
    \label{lemma: est c(nu) without split}
    Conditional on $\xi = \frac{X^T\varepsilon}{n}$ satisfying
    $\|\xi\|_\infty<\kappa\sigma\frac{1}{\sqrt{n}}$. 
   For any $i\in \{\widehat{\sign}_i \neq \sign(\gamma^*_i)\}$, if $\nu>4C_XC_D^2$, there holds
   \begin{align*}
        c(\nu)_i\le \frac{8}{\nu}\kappa C_XC_D\sigma\frac{1}{\sqrt{n}},
   \end{align*}
   where for all $i$, $c(\nu)_i$ satisfies
   \begin{align*}
        \{W_i\ge 0\}\subseteq \{-c(\nu)_i\le r_i \zeta_i\le 0\}, \{W_i\le 0\}\subseteq \{r_i \zeta_i\ge 0\}\cup\{r_i \zeta_i\le -c(\nu)_i\}.
    \end{align*}
\end{lemma}

Lemma \ref{lemma: est c(nu) without split} can be proved by exactly the same approach as Lemma \ref{lemma: interval} and \ref{lemma: est c(nu)}. The details will be omitted here. With Lemma \ref{lemma: est c(nu) without split}, we can now prove Lemma \ref{lemma: estimation of sign}.

\begin{proof}
    For any $\xi$, any $i$ and any interval $U\subseteq %\left(-\theta \sigma\frac{1}{\sqrt{n}}\sqrt{2s-s^2\nu-s^2R_{i, i}}, \theta \sigma\frac{1}{\sqrt{n}}\sqrt{2s-s^2\nu-s^2R_{i, i}}\right) : 
     \Lambda(\nu)_i$, there holds
    \begin{align*}
        &|\Prob(r_i\zeta_i\in U|\xi, \zeta\in\Lambda(\nu)) - \Prob(r_i\zeta_i \in U|\xi)|\\
        =&|\Prob(r_i\zeta_i\in U|\xi, \zeta\in\Lambda(\nu))\Prob(\zeta\in\Lambda(\nu)|\xi) +\Prob(r_i\zeta_i\in U|\xi, \zeta\in\Lambda(\nu))\Prob(\zeta\notin\Lambda(\nu)|\xi) - \cdots\\
        &\cdots -\Prob(r_i\zeta_i\in U|\xi, \zeta\in\Lambda(\nu))\Prob(\zeta\in\Lambda(\nu)|\xi)- \Prob(r_i\zeta_i\in U|\xi, \zeta\notin\Lambda(\nu))\Prob(\zeta\notin\Lambda(\nu)|\xi)|,\\
        = & \Prob(r_i\zeta_i\in U|\xi, \zeta\in\Lambda(\nu))\Prob(\zeta\notin\Lambda(\nu)|\xi),
    \end{align*}
    where the last step is due to $\{r_i\zeta_i\in U\}\cap \{\zeta\notin\Lambda(\nu)\}=\emptyset$. Therefore,
    \begin{align}
        \Prob(r_i\zeta_i \in U|\xi) &\ge \Prob(r_i\zeta_i\in U|\xi, \zeta\in\Lambda(\nu))(1-\Prob(\zeta\notin\Lambda(\nu)|\xi)),\nonumber\\
        &\ge \Prob(r_i\zeta_i\in U|\xi, \zeta\in\Lambda(\nu))(1-m(2-2\Phi(\theta))).\label{ineq: tran temp}
    \end{align}
    With some basic calculation, it can be shown that for all $i$,
    \begin{align*}
        %\E(\zeta_i|\xi)&= -\vecs_i\left\{D\left[\frac{X^TX}{n}\right]^{-1}\left(\xi-\frac{X^TX}{n}\beta^*\right)\right\}_i,\\
        \var(\zeta_i|\xi)& = \frac{\sigma^2}{n}\left(2\vecs_i-\nu\vecs_i^2-\vecs_i^2R_{i, i}\right).
    \end{align*}
    with $R$ defined in Lemma \ref{lemma: main_lemma}. For $\nu> 4C_XC_D^2$, as $R_{i, i}\le \|R\|_\infty\le C_XC_D^2$, there holds
    \begin{align*}
        \var(\zeta_i|\xi)\ge  \frac{\sigma^2}{n}\left(2\vecs_i-\frac{5}{4}\nu\vecs_i^2\right)\ge \frac{3}{4}\frac{\sigma^2}{n}\vecs_i.
    \end{align*}
    Moreover, for $\nu> \frac{3}{2}\co1$, there holds for $C_\nu$ defined in Section \ref{sec: construct split knockoff} that $C_\nu\succeq \frac{1}{\nu}\left(I_m-\frac{R}{\nu}\right)\succeq \frac{1}{3\nu}I_m$, suggesting $\vecs_i\ge\frac{2}{3\nu}$. Therefore there further holds
    \begin{align}
        \var(\zeta_i|\xi)\ge \frac{3}{4}\frac{\sigma^2}{n}\vecs_i \ge \frac{1}{2\nu}\frac{\sigma^2}{n}.\label{ineq: temp}
    \end{align}
    Let $U = [-c(\nu)_i, 0]\cap\Lambda(\nu)_i$, then for $\xi$ satisfying $\|\xi\|_\infty<\kappa\sigma\frac{1}{\sqrt{n}}$, there holds
    \begin{align*}
        \Prob(\zeta_ir_i\in U| \xi) \le &\int_{-c(\nu)_i}^{0}\frac{1}{\sqrt{2\pi\var(\zeta_i|\xi)}}\exp\left(-\frac{(r_ix-r_i\E[\zeta_i|\xi])^2}{2\var(\zeta_i|\xi)}\right)dx,\\
        \le & \frac{c(\nu)_i}{\sqrt{2\pi\var(\zeta_i|\xi)}}
        \le \frac{8\kappa C_XC_D}{\sqrt{\pi}\sqrt{\nu}},
    \end{align*}
    where the last step is from Equation \eqref{ineq: temp} and the estimation of $c(\nu)_i$ in Lemma \ref{lemma: est c(nu) without split}. By Lemma \ref{lemma: est c(nu) without split} and Equation \eqref{ineq: tran temp}, for $\theta$ satisfying $m(2-2\Phi(\theta))<1$, there holds
    \begin{align*}
        \Prob[W_i\ge 0|\xi, \zeta\in\Lambda(\nu)]\le \Prob(r_i\zeta_i\in U|\xi, \zeta\in\Lambda(\nu)) \le \frac{8\kappa C_XC_D}{\sqrt{\pi}\sqrt{\nu}}\frac{1}{1-m(2-2\Phi(\theta))}.
    \end{align*}
    This ends the proof.
\end{proof}

\subsection{Proof of Lemma \ref{lemma: main_lemma}}

\label{sec: main_lemma}

In this section, we will provide the proof of Lemma \ref{lemma: main_lemma}. We first characterize that $\sign(W_i)$ are weakly dependent on each other by exploiting the dependency between $\zeta$ and $\xi$. In the second part, we proceed to prove the supermartingale inequality in Lemma \ref{lemma: main_lemma}. Some commonly used notations in this section are stated below.

\begin{itemize}
    \item Denote $b_n(m)\in\{0, 1\}^n$ as the $n$-dimensional vector which represents the $n$-digit binary number of $m$ for $m\in \{0, 1, \cdots, 2^n-1\}$. For example, $b_3(3)=[0, 1, 1]$.
    \item Denote $a\preceq b$ for two vectors $a\in\R^n$ and $b\in\R^n$, if and only if $a_i\le b_i$ for all $i$.
\end{itemize}

\subsubsection{The Weak Dependency on $W$ statistics}

\label{sec: weak dependent}

Our final goal in this section is to prove the following lemma, which states that $B_i = 1\{W_i<0\}$ are nearly independent Bernoulli random variables.

\begin{lemma}
    \label{lemma: weak dependent}
    %Fix a $\xi$ satisfying $\|\xi\|_\infty<\kappa\sigma\frac{1}{\sqrt{n}}$, and 
    Define $B_i = 1\{W_i<0\}$ for all $i$. There holds for any $\nu>\max\left\{\frac{3}{2}\co1, 12\co2\right\}$, any $\xi$ satisfying $\|\xi\|_\infty<\kappa\sigma\frac{1}{\sqrt{n}}$, any $S\subseteq\{1, 2, 3, \cdots, m\}$, any $k\in \{0, 1, \cdots,2^{|S|}-1\}$ and any $l\in \{0, 1, \cdots,2^{|S^c|}-1\}$ that
    \begin{align*}
        e^{-48\theta ^2\co2 \frac{m(|S|+1)}{\nu}}\le\frac{\Prob[B_S = b_{|S|}(k)|\xi, \zeta\in\Lambda(\nu), B_{S^c} = b_{|S^c|}(l)]}{\prod_{i\in S}\Prob[B_i = b_{|S|}(l)_i|\xi, \zeta\in\Lambda(\nu), B_{S^c} = b_{|S^c|}(l)]}\le e^{48\theta ^2\co2 \frac{m(|S|+1)}{\nu}},
    \end{align*}
    where $\co1$, $\Lambda(\nu)$, $\theta$ are defined in Lemma \ref{lemma: estimation of sign}, $\co2$ is  defined in Lemma \ref{lemma: main_lemma}.
\end{lemma}
%On the way in achieving Lemma \ref{lemma: weak dependent}, the following proposition can be achieved by combining some intermediate results with Lemma \ref{lemma: estimation of sign}, whose proof will also be provided in this section.
We also provide the proof of the following proposition in this section, which can be regarded as the corollary of Lemma \ref{lemma: weak dependent} and  Lemma \ref{lemma: estimation of sign}.
\begin{proposition}
\label{prop: conditional upper bound}
    Given a $\xi$  satisfying $\|\xi\|_\infty<\kappa\sigma\frac{1}{\sqrt{n}}$. Suppose $\nu>\max\left\{4C_XC_D^2, \frac{3}{2}\co1, 12\co2\right\}$, then for $S = \{i: \widehat{\sign}_i\neq \sign(\gamma^*_i)\}$  and any $i\in S$, if $m(2-2\Phi(\theta))<1$, there holds
    \begin{align*}
        \Prob\left.\left(B_i = 0\right|\xi, \zeta\in\Lambda(\nu), B_{S^c} = b_{|S^c|}(l)\right)&<\frac{8\kappa C_XC_D}{\sqrt{\pi}\sqrt{\nu}}\frac{1}{1-m(2-2\Phi(\theta))}e^{36\theta ^2\co2 \frac{m}{\nu}},
    \end{align*}
    for any $l\in \{0, 1, \cdots,2^{|S^c|}-1\}$.
\end{proposition}

Now we begin to prove Lemma \ref{lemma: weak dependent} and Proposition \ref{prop: conditional upper bound}. We start from the Lemma \ref{lemma: exponential diff} which evaluates the dependency between $\zeta$ and $\xi$. %For simplicity, the notations defined in Lemma \ref{lemma: main_lemma} won't be highlighted in the following.

\begin{lemma}
    \label{lemma: exponential diff}
    %Fix a $\xi$ satisfying $\|\xi\|_\infty<\kappa\sigma\frac{1}{\sqrt{n}}$, then the conditional distribution $\zeta|\xi:= u$ satisfy the normal distribution with:
    %\begin{align*}
        %\E(u)&= -[\diag(\vecs)D][\Sigma_X]^{-1}\xi,\\
        %\var(u)& = \frac{\sigma^2}{n}(2s I_m-s^2\nu I_m -s^2 R): = \Sigma,
    %\end{align*}
    %where $R=D\Sigma_X^{-1} D^T$, and $\Sigma_X =\frac{X^T X}{n}$. 
    Let $\{v_i\}_{i=1}^m$ be independent normal random variables satisfying
    \begin{align*}
        \E(v)&= -\diag(\vecs)D\left[\frac{X^TX}{n}\right]^{-1}\xi,\\
        \var(v)& = \frac{\sigma^2}{n}\diag(\vecs)(2 I_m-\nu\diag(\vecs)I_m -\diag(\vecs) \diag(R)): = \Sigma_0,
    \end{align*}
    where $\diag(R)$ is defined in Lemma \ref{lemma: main_lemma}. Then for any $a, b\in\R^m$, satisfying 
    \begin{align*}
        -\theta \sigma\frac{1}{\sqrt{n}}\sqrt{2\vecs_i-\vecs_i^2\nu-\vecs_i^2R_{i, i}}\le a_i<b_i\le \theta \sigma\frac{1}{\sqrt{n}}\sqrt{2\vecs_i-\vecs_i^2\nu-\vecs_i^2R_{i, i}},
    \end{align*}
    for any $i$, there holds for $\nu>\max\left\{\frac{3}{2}\co1, 12\co2\right\}$ and $\xi$ satisfying $\|\xi\|_\infty<\kappa\sigma\frac{1}{\sqrt{n}}$:
    \begin{align*}
        \sqrt{\frac{|\Sigma_0|}{|\Sigma|}}e^{-6\theta ^2\co2 \frac{m}{\nu}}\le\frac{\Prob\left(a\preceq \zeta-\E(\zeta)\preceq b|\xi\right)}{\Prob\left(a\preceq v-\E(v)\preceq b\right)}\le\sqrt{\frac{|\Sigma_0|}{|\Sigma|}}e^{6\theta ^2\co2 \frac{m}{\nu}},
    \end{align*}
    where $\Sigma$ is defined by
    \begin{align*}
        \Sigma := \var(\zeta|\xi) = \frac{\sigma^2}{n}\diag(\vecs)(2 I_m-\nu\diag(\vecs)I_m -\diag(\vecs) R).
    \end{align*}
\end{lemma}

\begin{proof}
    Firstly, there holds for any $a$, $b$ satisfying the constraint in Lemma \ref{lemma: exponential diff} that
    \begin{align*}
        \Prob\left(a\preceq \zeta-\E(\zeta)\preceq b|\xi\right) = \int_{a\preceq x\preceq b}\frac{1}{\sqrt{(2\pi)^m|\Sigma|}}e^{-\frac{1}{2}x^T\Sigma^{-1} x},\\
        \Prob\left(a\preceq v-\E(v) \preceq b\right) = \int_{a\preceq x\preceq b}\frac{1}{\sqrt{(2\pi)^m|\Sigma_0|}}e^{-\frac{1}{2}x^T\Sigma_0^{-1} x}.
    \end{align*}
    Therefore, it is sufficient to show that, for any $a\preceq x \preceq b$, there holds
    \begin{align*}
        e^{-6\theta ^2\co2 \frac{m}{\nu}}\le e^{-\frac{1}{2}x^T(\Sigma^{-1}-\Sigma_0^{-1})x}\le e^{6\theta ^2\co2 \frac{m}{\nu}}.
    \end{align*}
    From the constraint on $a$ and $b$ in Lemma \ref{lemma: exponential diff}, for any $x$ satisfying $a\preceq x \preceq b$, there holds $\|x\|_2^2\le\max_i\left\{\theta ^2\sigma^2\frac{\vecs_i m}{n}\right\}\le \theta ^2\sigma^2\frac{m}{n\nu}$, due to the fact $\vecs_i\le\frac{1}{\nu}$ by definition. Therefore,
    \begin{align}
        \left|\frac{1}{2}x^T(\Sigma^{-1}-\Sigma_0^{-1})x\right|\le \frac{1}{2}\|x\|_2^2\|\Sigma^{-1}-\Sigma_0^{-1}\|_2\le \frac{\theta ^2\sigma^2}{2}\frac{m}{n\nu}\|\Sigma^{-1}-\Sigma_0^{-1}\|_2.\label{ineq: bound square}
    \end{align}
    For the right hand side of Equation \eqref{ineq: bound square}, we have the following decomposition
    \begin{align*}
        \left\|\left[\frac{n}{\sigma^2}\Sigma\right]^{-1}-\left[\frac{n}{\sigma^2}\Sigma_0\right]^{-1}\right\|_2 = &  \left\|\left[\frac{n}{\sigma^2}\Sigma\right]^{-1}-\left[\frac{n}{\sigma^2}\Sigma-\diag(\vecs)^2(\diag(R)-R)\right]^{-1}\right\|_2,\\
        \le &  \left\|\left[\frac{n}{\sigma^2}\Sigma\right]^{-1}\right\|_2\left\|I_m-\left[I_m-\diag(\vecs)^2\left[\frac{n}{\sigma^2}\Sigma\right]^{-1}(\diag(R)-R)\right]^{-1}\right\|_2.
    \end{align*}
    %Recall that $s\le \nu^{-1}$, therefore for $\nu>2\max\{\co1, 2\co2\}$, there holds $\nu I_m\succeq 2R$, and
    For $\nu>\frac{3}{2}\co1$, there holds for $C_\nu$ defined in Section \ref{sec: construct split knockoff} that $C_\nu\succeq \frac{1}{\nu}\left(I_m-\frac{R}{\nu}\right)\succeq \frac{1}{3\nu}I_m$, suggesting $\vecs_i\ge\frac{2}{3\nu}$. By $x\left(2-\frac{5}{3}x\right)\ge \frac{1}{3}$ for $\frac{2}{3}\le x\le 1$, there further holds
    \begin{align}
        \frac{n}{\sigma^2}\Sigma\succeq\frac{1}{\nu}\diag(\vecs)\nu \left(2I_m-\frac{5}{3}\diag(\vecs)\nu I_m\right)\succeq\frac{1}{3\nu}I_m. \label{eq: estimate}
    \end{align}
    %Since $\nu>2\co1$, there holds $\nu I_m\succeq 2R$, therefore
    %\begin{align*}
        %\frac{n}{\sigma^2}\Sigma= 2s I_m-s^2\nu I_m - s^2R\succeq2s I_m-\frac{3}{2}s^2\nu I_m=\frac{1}{\nu}s\nu\left(2-\frac{3}{2}s\nu\right) I_m,
    %\end{align*}
    %moreover
    %\begin{align}
        %C_\nu=\frac{I}{\nu} - \frac{D}{\nu}\left[\frac{X^T X}{n}+\frac{D^T D}{\nu}\right]^{-1}\frac{D^T}{\nu}\succeq\frac{1}{\nu^2}\left(\nu I_m - R\right)\succeq \frac{1}{2\nu}I_m. \label{eq: min eigen}
    %\end{align}
    %From Equation \eqref{eq: min eigen} and definition of $s$, there holds $s\ge \min_{\lambda}C_\nu\ge \frac{1}{2\nu}$, and thus $\frac{1}{2}\le s\nu\le 1$, where $s\nu\le 1$ is from the definition of $s$. Therefore, $s\nu\left(2-\frac{3}{2}s\nu\right)\ge \frac{1}{2}$, and
    %\begin{align}
        %\frac{n}{\sigma^2}\Sigma\succeq\frac{1}{2\nu} I_m\succeq\frac{1}{2}s I_m.\label{eq: estimate}
   % \end{align}
    For any square matrix $A$, if $\sum_{k=1}^\infty A^k$ converges, then $I-(I-A)^{-1}=\sum_{k=1}^\infty A^k$. 
    Therefore, there further holds for $\nu>12\co2$ that
    \begin{align*}
        \frac{\sigma^2}{n}\|\Sigma^{-1}-\Sigma_0^{-1}\|_2
        \le & 3\nu\sum_{k=1}^{\infty}\frac{1}{\nu^{2k}}\left\|\left[\frac{n}{\sigma^2}\Sigma\right]^{-1}\right\|_2^k\|\diag(R)-R\|_2^k,\\
        \le & 3\nu \frac{\frac{3\co2}{\nu}}{1-\frac{3\co2}{\nu}}< 12\co2.
    \end{align*}
    Combining with Equation \eqref{ineq: bound square}, there holds
    \begin{align*}
        \left|\frac{1}{2}x^T(\Sigma^{-1}-\Sigma_0^{-1})x\right|\le 6\theta ^2\co2 \frac{m}{\nu},
    \end{align*}
    which ends the proof.
\end{proof}

The following corollary can be immediately established from Lemma \ref{lemma: exponential diff}.

\begin{corollary}
    \label{coro: exp diff}
    There holds for $\nu>\max\left\{\frac{3}{2}\co1, 12\co2\right\}$, $\xi$ satisfying $\|\xi\|_\infty<\kappa\sigma\frac{1}{\sqrt{n}}$
    \begin{align*}
        e^{-12\theta ^2\co2 \frac{m}{\nu}}\le\frac{\Prob\left(a\preceq \zeta-\E(\zeta)\preceq b|\xi, \zeta\in\Lambda(\nu)\right)}{\Prob\left(a\preceq v-\E(v)\preceq b|v\in\Lambda(\nu)\right)}\le e^{12\theta ^2\co2 \frac{m}{\nu}},
    \end{align*}
    where $\Lambda(\nu)$ is defined in Lemma \ref{lemma: estimation of sign}, $v$ is defined in Lemma \ref{lemma: exponential diff}, $a$, $b$ are two $m$-dimensional vectors satisfying the constraint in Lemma \ref{lemma: exponential diff}.
\end{corollary}
%\begin{remark}
    %Such result can be further generalized from the open set $\{x\in \R^m: a\preceq x\preceq b\}$ to any Borel set $\mathcal{B}$.
%\end{remark}

\begin{proof}
    By the constraint on $a$, $b$ in Lemma \ref{lemma: exponential diff}, there holds
    \begin{align*}
        \{x\in \R^m:a\preceq x\preceq b\}\subseteq \{x\in \R^m:x\in\Lambda(\nu)\}.
    \end{align*}
    Therefore from Lemma \ref{lemma: exponential diff}, there holds
    \begin{align}
        \sqrt{\frac{|\Sigma_0|}{|\Sigma|}}e^{-6\theta ^2\co2 \frac{m}{\nu}}\le\frac{\Prob\left(a\preceq \zeta-\E(\zeta)\preceq b, \zeta\in \Lambda(\nu)|
    \xi\right)}{\Prob\left(a\preceq v-\E(v)\preceq b, v\in \Lambda(\nu)\right)}\le\sqrt{\frac{|\Sigma_0|}{|\Sigma|}}e^{6\theta ^2\co2 \frac{m}{\nu}}.\label{eq: ori}
    \end{align}
    Moreover, take $-a_i=b_i=\theta \sigma\frac{1}{\sqrt{n}}\sqrt{2\vecs_i-\vecs_i^2\nu-\vecs_i^2R_{i, i}}$, then Equation \eqref{eq: ori} becomes
    \begin{align}
        \sqrt{\frac{|\Sigma_0|}{|\Sigma|}}e^{-6\theta ^2\co2 \frac{m}{\nu}}\le\frac{\Prob\left(\zeta\in\Lambda(\nu)|\xi\right)}{\Prob\left(v\in\Lambda(\nu)\right)}\le\sqrt{\frac{|\Sigma_0|}{|\Sigma|}}e^{6\theta ^2\co2 \frac{m}{\nu}}. \label{eq: spe}
    \end{align}
    The result follows by \eqref{eq: ori} / \eqref{eq: spe}.
\end{proof}

Then we can finally prove Lemma \ref{lemma: weak dependent} using Lemma \ref{lemma: interval} and Corollary \ref{coro: exp diff}.

\begin{proof}
    By Lemma \ref{lemma: interval}, for all $i$, there exists some $a_i<b_i$, such that $|a_i|, |b_i|\le \theta \sigma\frac{1}{\sqrt{n}}\sqrt{2\vecs_i-\vecs_i^2\nu-\vecs_i^2R_{i, i}}$ (the boundary of $\Lambda(\nu)$), and $\Prob[B_i = 1|\xi, \zeta\in\Lambda(\nu)] = \Prob[a_i<\zeta_i-\Expect[\zeta_i]<b_i|\xi, \zeta\in\Lambda(\nu)]$. Therefore, for any $k$, there exists $I=[I_1, I_2, \cdots, I_m]$, consisting of some suitable combinations of intervals related with $b_m(k)$, such that $\Prob[B_i = b_m(k)_i|\xi, \zeta\in\Lambda(\nu)] = \Prob[\zeta_i-\Expect[\zeta_i]\in I_i|\xi, \zeta\in\Lambda(\nu)]$. By Corollary \ref{coro: exp diff}, for $\nu>\max\left\{\frac{3}{2}\co1, 12\co2\right\}$ and $\xi$ satisfying $\|\xi\|_\infty<\kappa\sigma\frac{1}{\sqrt{n}}$, there further holds
    \begin{align}
        e^{-12\theta ^2\co2 \frac{m}{\nu}}\le\frac{\Prob[B_S = b_m(k)_S|\xi, \zeta\in\Lambda(\nu)]}{\Prob[v_S-\E[v_S]\in I_S|v\in\Lambda(\nu)]}\le e^{12\theta ^2\co2 \frac{m}{\nu}},\label{eq: part 1}
    \end{align}
    for any $S\subseteq\{1, 2, 3, \cdots, m\}$, where $v$ is defined in Lemma \ref{lemma: exponential diff}. Moreover, for all $i\in S$
    \begin{align}
        e^{-12\theta ^2\co2 \frac{m}{\nu}}\le\frac{\Prob[B_i = b_m(k)_i|\xi, \zeta\in\Lambda(\nu)]}{\Prob[v_i-\E[v_i]\in I_i|v\in\Lambda(\nu)]}\le e^{12\theta ^2\co2 \frac{m}{\nu}}.\label{eq: part 2}
    \end{align}
    By the independence among $\{v_i\}_{i=1}^m$, there holds
    \begin{align}
        \frac{\prod_{i\in S}\Prob[v_i-\E[v_i]\in I_i|v\in\Lambda(\nu)]}{\Prob[v_S-\E[v_S]\in I_S|v\in\Lambda(\nu)]}=1. \label{eq: part 3}
    \end{align}
    Combining Equation \eqref{eq: part 1}, \eqref{eq: part 2} and \eqref{eq: part 3}, there holds
    \begin{align}
        e^{-12\theta ^2\co2 \frac{m(|S|+1)}{\nu}}\le\frac{\Prob[B_S = b_m(k)_S|\xi, \zeta\in\Lambda(\nu)]}{\prod_{i\in S}\Prob[B_i = b_m(k)_i|\xi, \zeta\in\Lambda(\nu)]}\le e^{12\theta ^2\co2 \frac{m(|S|+1)}{\nu}}.\label{part1}
    \end{align}
    Similarly, it can be shown that
    \begin{align}
        &e^{-12\theta ^2\co2 \frac{3m}{\nu}}\le\frac{\Prob[B_S = b_m(k)_S, B_{S^c} = b_m(k)_{S^c}|\xi, \zeta\in\Lambda(\nu)]}{\Prob[B_S = b_m(k)_S|\xi, \zeta\in\Lambda(\nu)]\Prob[B_{S^c} = b_m(k)_{S^c}|\xi, \zeta\in\Lambda(\nu)]}\le e^{12\theta ^2\co2 \frac{3m}{\nu}},\label{part2}\\
        &e^{-12\theta ^2\co2 \frac{3m}{\nu}}\le\frac{\Prob[B_i = b_m(k)_i, B_{S^c} = b_m(k)_{S^c}|\xi, \zeta\in\Lambda(\nu)]}{\Prob[B_i = b_m(k)_i|\xi, \zeta\in\Lambda(\nu)]\Prob[B_{S^c} = b_m(k)_{S^c}|\xi, \zeta\in\Lambda(\nu)]}\le e^{12\theta ^2\co2 \frac{3m}{\nu}}\label{part3},
    \end{align}
    where
    \begin{align*}
        &\frac{\Prob[B_S = b_m(k)_S, B_{S^c} = b_m(k)_{S^c}|\xi, \zeta\in\Lambda(\nu)]}{\Prob[B_{S^c} = b_m(k)_{S^c}|\xi, \zeta\in\Lambda(\nu)]} = \Prob[B_S = b_m(k)_S|\xi, \zeta\in\Lambda(\nu), B_{S^c} = b_m(k)_{S^c}],\\
        & \frac{\Prob[B_i = b_m(k)_i, B_{S^c} = b_m(k)_{S^c}|\xi, \zeta\in\Lambda(\nu)]}{\Prob[B_{S^c} = b_m(k)_{S^c}|\xi, \zeta\in\Lambda(\nu)]} = \Prob[B_i = b_m(k)_i|\xi, \zeta\in\Lambda(\nu), B_{S^c} = b_m(k)_{S^c}].
    \end{align*}
    Combining Equation \eqref{part1}, \eqref{part2} and \eqref{part3} together, there holds
    \begin{align*}
        e^{-48\theta ^2\co2 \frac{m(|S|+1)}{\nu}}\le\frac{\Prob[B_S = b_m(k)_S|\xi, \zeta\in\Lambda(\nu), B_{S^c} = b_m(k)_{S^c}]}{\prod_{i\in S}\Prob[B_i = b_m(k)_i|\xi, \zeta\in\Lambda(\nu), B_{S^c} = b_m(k)_{S^c}]}\le e^{48\theta ^2\co2 \frac{m(|S|+1)}{\nu}}.
    \end{align*}
    Since the above inequality holds for all $k$, there holds
    \begin{align*}
        e^{-48\theta ^2\co2 \frac{m(|S|+1)}{\nu}}\le\frac{\Prob[B_S = b_{|S|}(k)|\xi, \zeta\in\Lambda(\nu), B_{S^c} = b_{|S^c|}(l)]}{\prod_{i\in S}\Prob[B_i = b_{|S|}(k)_i|\xi, \zeta\in\Lambda(\nu), B_{S^c} = b_{|S^c|}(l)]}\le e^{48\theta ^2\co2 \frac{m(|S|+1)}{\nu}},
    \end{align*}
    for all $k\in \{0, 1, \cdots,2^{|S|}-1\}$ and $l\in \{0, 1, \cdots,2^{|S^c|}-1\}$.
\end{proof}

Combining Equation \eqref{part3} with Lemma \ref{lemma: estimation of sign}, we can prove Proposition \ref{prop: conditional upper bound}.

\begin{proof}
    By Equation \eqref{part3}, for $S = \{i: \widehat{\sign}_i\neq \sign(\gamma^*_i)\}$, any $i\in S$, and any $l\in \{0, 1, \cdots,2^{|S^c|}-1\}$, there holds
    \begin{align*}
        &e^{-12\theta ^2\co2 \frac{3m}{\nu}}\le\frac{\Prob[B_i = 0, B_{S^c} =  b_{|S^c|}(l)|\xi, \zeta\in\Lambda(\nu)]}{\Prob[B_i = 0|\xi, \zeta\in\Lambda(\nu)]\Prob[B_{S^c} =  b_{|S^c|}(l)|\xi, \zeta\in\Lambda(\nu)]}\le e^{12\theta ^2\co2 \frac{3m}{\nu}}.
    \end{align*}
    Combining with Lemma \ref{lemma: estimation of sign}, there further holds
    \begin{align*}
        \Prob[B_i = 0|\xi, \zeta\in\Lambda(\nu), B_{S^c} =  b_{|S^c|}(l)]&\le\Prob[B_i = 0|\xi, \zeta\in\Lambda(\nu)]e^{36\theta ^2\co2 \frac{m}{\nu}},\\
        &< \frac{8\kappa C_XC_D}{\sqrt{\pi}\sqrt{\nu}}\frac{1}{1-m(2-2\Phi(\theta))}e^{36\theta ^2\co2 \frac{m}{\nu}}.
    \end{align*}
    This ends the proof.
\end{proof}

\subsubsection{The Almost Supermartingale Inequality}

\label{sec: almost supermartingale}

Now we proceed to prove the almost supermartingale inequality in Lemma \ref{lemma: main_lemma}. $\{B_i\}_{i=1}^m$ defined in Lemma \ref{lemma: weak dependent} are heterogeneous Bernoulli random variables ($\Prob[B_i=1]$ are different). To achieve the almost supermartingale inequality, we first need to transfer the heterogeneous Bernoulli random variables into homogeneous ones. In particular, we first introduce the following general lemma which builds connections between  heterogeneous Bernoulli random variables and homogeneous ones.

%We will first show a lemma which describe the property of the ``linear transformation'' on the probability mass function of multiple Bernoulli random variables.

\begin{lemma}
    \label{lemma: QC eq}
    Let $\{\rho_1, \rho_2, \cdots, \rho_n, \rho\}$ be positive numbers satisfying 
    \begin{align*}
        1> \max_i\rho_i\ge\min_i\rho_i\ge \rho>0.
    \end{align*}
    Let $\{Q_i\}_{i=1}^n$ be $n$ (possibly dependent) random variables satisfying $Q_i\sim \mathrm{Bernoulli}(\rho)$ for all $i$. Define a random set $\A\subset\{1, 2, \cdots, n\}$ such that:
    \begin{itemize}
        \item $\A$ is independent from $\{Q_i\}_{i=1}^n$;
        \item the events $\{i\in \A\}_{i=1}^n$ are independent from each other, and
        \begin{align*}
            \Prob\{i\in \A\} = \frac{1-\rho_i}{1-\rho}.
        \end{align*}
    \end{itemize}
    For all $i$, further define
    \begin{align*}
        C_i=Q_i\cdot 1\{i\in \A\}+1\{i\notin \A\}.
    \end{align*}
    Then there exists a invertible matrix $A(\rho_1, \rho_2, \cdots, \rho_n, \rho)\in \R^{2^n\times 2^n}$, s.t. 
    \begin{align}
        \left[\begin{array}{c}
            \Prob[C = b_n(0)]  \\
            \Prob[C = b_n(1)]  \\
            \cdots \\
            \Prob[C = b_n(2^n-1)]
        \end{array}\right] = A(\rho_1, \rho_2, \cdots, \rho_n, \rho)
        \left[\begin{array}{c}
            \Prob[Q = b_n(0)]  \\
            \Prob[Q = b_n(1)]  \\
            \cdots \\
            \Prob[Q = b_n(2^n-1)]
        \end{array}\right],\label{eq: caq}
    \end{align}
    where $C = [C_1, C_2, \cdots, C_n]$ and $Q = [Q_1, Q_2, \cdots, Q_n]$.
\end{lemma}
\begin{remark}
    $\{C_i\}_{i=1}^n$ are heterogeneous Bernoulli random variables with $\Prob[C_i=1] = \rho_i$.
\end{remark}

\begin{proof}
    
    For any $m\in \{0, 1, \cdots, 2^n-1\}$, there holds
    \begin{align*}
        \Prob[C = b_n(m)] = \sum_{k=0}^{2^n-1}\Prob[Q = b_n(k)]\Prob[C = b_n(m)|Q = b_n(k)].
    \end{align*}
    Since the events $\{i\in \A\}$ are independent from each other, and $\A$ is independent from $\{Q_i\}_{i=1}^n$, $C_i|Q$ are independent from each other, therefore
    \begin{align*}
        \Prob[C = b_n(m)|Q = b_n(k)] = \prod_{i=1}^n \Prob[C_i = b_n(m)_i|Q = b_n(k)].
    \end{align*}
    Moreover, $C_i$ is only dependent on $Q_i$ and $\A$, therefore $C_i|Q = C_i|Q_i$, and there holds
    \begin{align}
        \Prob[C = b_n(m)|Q = b_n(k)] = \prod_{i=1}^n \Prob[C_i = b_n(m)_i|Q_i = b_n(k)_i],\label{eq: cq}
    \end{align}
    where $\Prob[C_i = b_n(m)_i|Q_i = b_n(k)_i]$ can be represented by $\rho_i, \rho$ and binary parameters $(b_n(m)_i, b_n(k)_i) \in \{(0, 0), (0, 1), (1, 0), (1, 1)\}$ using one of the equations below,
    \begin{align*}
        \Prob[C_i = 0|Q_i = 0] = \frac{1-\rho_i}{1-\rho},\ \Prob[C_i = 0|Q_i = 1] = 0,\\
        \Prob[C_i = 1|Q_i = 0] = \frac{\rho_i-\rho}{1-\rho},\ \Prob[C_i = 1|Q_i = 1] = 1.
    \end{align*}
    Now, define the matrix $A(\rho_1, \rho_2, \cdots, \rho_n, \rho)\in \R^{2^n\times 2^n}$ by
    \begin{align}
        \label{eq: Q_transfer}
        A(\rho_1, \rho_2, \cdots, \rho_n, \rho)(m, k) = \prod_{i=1}^n \Prob[C_i = b_n(m-1)_i|Q_i = b_n(k-1)_i],
    \end{align}
    where $A(\rho_1, \rho_2, \cdots, \rho_n, \rho)(m, k)$ represents the $(m, k)$ element of $A(\rho_1, \rho_2, \cdots, \rho_n, \rho)$. By Equation \eqref{eq: cq}, $A(\rho_1, \rho_2, \cdots, \rho_n, \rho)$ in Equation \eqref{eq: Q_transfer} satisfies Equation \eqref{eq: caq}.
    
    We use mathematical induction with respect to $n$ to show that $A(\rho_1, \rho_2, \cdots, \rho_n, \rho)$ in Equation \eqref{eq: Q_transfer} is invertible. For $n=1$, it can be directly calculated that
    \begin{align*}
        A_{\rho_1, \rho} = 
        \begin{bmatrix}
            \frac{1-\rho_1}{1-\rho} & 0\\
            \frac{\rho_1-\rho}{1-\rho} & 1
        \end{bmatrix},
    \end{align*}
    is invertible. Suppose that this argument holds for $n = L$. Then for $n = L+1$, since
    \begin{align*}
        \Prob[C_1 = 0|Q_1 = 0] = \frac{1-\rho_1}{1-\rho},\ \Prob[C_1 = 0|Q_1 = 1] = 0,\\
        \Prob[C_1 = 1|Q_1 = 0] = \frac{\rho_1-\rho}{1-\rho},\ \Prob[C_1 = 1|Q_1 = 1] = 1,
    \end{align*}
    by Equation \eqref{eq: Q_transfer}, $A(\rho_1, \rho_2, \cdots, \rho_{L+1}, \rho)$ can be written in the following form
    \begin{align}
        \label{eq: A_iter}
        A(\rho_1, \rho_2, \cdots, \rho_{L+1}, \rho) = 
        \begin{bmatrix}
            \frac{1-\rho_1}{1-\rho} A(\rho_2, \rho_3, \cdots, \rho_{L+1}, \rho) & 0_{2^L\times 2^L}\\
            \frac{\rho_1-\rho}{1-\rho} A(\rho_2, \rho_3, \cdots, \rho_{L+1}, \rho) & A(\rho_2, \rho_3, \cdots, \rho_{L+1}, \rho)
        \end{bmatrix}.
    \end{align}
    By the induction assumption, the matrix $A(\rho_2, \rho_3, \cdots, \rho_{L+1}, \rho)$ is invertible. Therefore, by Equation \eqref{eq: A_iter}, $A(\rho_1, \rho_2, \cdots, \rho_{L+1}, \rho)$ is invertible.
\end{proof}

In the following, we will use Lemma \ref{lemma: QC eq} to show that weakly dependent heterogeneous Bernoulli random variables can be transferred into weakly dependent homogeneous Bernoulli random variables.

\begin{lemma}
    \label{lemma: key}
    Let $\{B_i\}_{i=1}^n$ be $n$ (possibly dependent) Bernoulli random variables, with $1>\Prob[B_i=1]:=\rho_i>0$ for all $i$. Let $\rho>0$ satisfy
    \begin{align*}
        \min_i\rho_i\ge \rho>0.
    \end{align*}
    Suppose $\{B_i\}_{i=1}^n$ are weakly dependent on each other, i.e. there exists $\epsilon>0$, such that for $B:=[B_1, B_2, \cdots, B_n]$ and any $m\in \{0, 1, \cdots, 2^n-1\}$, there holds
    \begin{align}
        \left|\Prob[B = b_n(m)]-\prod_{i=1}^n\Prob\left[B_i = b_n(m)_i\right]\right|\le \epsilon \prod_{i=1}^n\Prob\left[B_i = b_n(m)_i\right].\label{eq: condition on B}
    \end{align}
    Then there exists $n$ random variables $Q_1, Q_2,\cdots, Q_n$ with the following properties.
    \begin{itemize}
        \item $Q_i\sim \mathrm{Bernoulli}(\rho)$ for all $i$.
        \item For any $m\in \{0, 1, \cdots, 2^n-1\}$, there holds
        \begin{align}
            \label{eq: Q_bound}
            \left|\Prob[Q = b_n(m)]-\prod_{i=1}^n\Prob\left[Q_i = b_n(m)_i\right]\right|\le \epsilon\left[\frac{2-\rho}{\rho}\right]^n \prod_{i=1}^n\Prob\left[Q_i = b_n(m)_i\right].
        \end{align}
        \item Define $\A\subseteq\{1, 2, \cdots, n\}$ as a random set such that $\A$ is independent from $\{Q_i\}_{i=1}^n$, the events $\{i\in \A\}$ are independent from each other, and
        \begin{align*}
            \Prob\{i\in A\} = \frac{1-\rho_i}{1-\rho}.
        \end{align*}
        Further define $C_i=Q_i\cdot 1\{i\in A\}+1\{i\notin A\}$, then there holds
        \begin{align*}
            \{C_1, C_2,\cdots, C_n\} \overset{d}{=}\{B_1, B_2, \cdots, B_n\}.
        \end{align*}
    \end{itemize}
\end{lemma}

\begin{proof}
    
    We show that $\{Q_i\}_{i=1}^n$ whose probability mass function is given by
    \begin{align*}
        \left[\begin{array}{c}
            \Prob[Q = b_n(0)]  \\
            \Prob[Q = b_n(1)]  \\
            \cdots \\
            \Prob[Q = b_n(2^n-1)]
        \end{array}\right]:= A(\rho_1, \rho_2, \cdots, \rho_n, \rho)^{-1}
        \left[\begin{array}{c}
            \Prob[B = b_n(0)]  \\
            \Prob[B = b_n(1)]  \\
            \cdots \\
            \Prob[B = b_n(2^n-1)],
        \end{array}\right]
    \end{align*}
    satisfies the three properties, where $A(\rho_1, \rho_2, \cdots, \rho_n, \rho)$ is defined in Lemma \ref{lemma: QC eq}. The first and third property can be directly validated from the definition of $A(\rho_1, \rho_2, \cdots, \rho_n, \rho)$ in Lemma \ref{lemma: QC eq}, and the details will be omitted. 
    %The latter claim %is easy to prove. By some simple calculation with
    %can be directly validated from the definition of $A(\rho_1, \rho_2, \cdots, \rho_n, \rho)$ from Lemma \ref{lemma: QC eq}. %, and we can show that $\Prob[Q_i = 0] = \frac{1-\rho}{1-\rho_i}\Prob[B_i = 0]=1-\rho$, and thus $Q_i\sim\mathrm{Bernoulli}(\rho)$. 
    It remains to show that $\{Q_i\}_{i=1}^n$ satisfies Equation \eqref{eq: Q_bound}.

    In the following, we will prove a stronger statement in order to prove Equation \eqref{eq: Q_bound}. Let $Y\in \R^{2^n}$ be a positive vector, and define $X := A(\rho_1, \rho_2, \cdots, \rho_n, \rho)^{-1}Y\in \R^{2^n}$. Let $\tilde{Y}\in \R^{2^n}$ and $\tilde{X}\in \R^{2^n}$ be defined by 
    \begin{align*}
        \tilde{Y} =
        \left[\begin{array}{c}
            \prod_{i=1}^{n}\Prob\left[B_i = b_{n}(0)_{i}\right]  \\
            \prod_{i=1}^{n}\Prob\left[B_i = b_{n}(1)_{i}\right]  \\
            \cdots \\
            \prod_{i=1}^{n}\Prob\left[B_i = b_{n}(2^{n}-1)_{i}\right],
        \end{array}\right], \tilde{X} =
        \left[\begin{array}{c}
            \prod_{i=1}^{n}\Prob\left[Q_i = b_{n}(0)_{i}\right]  \\
            \prod_{i=1}^{n}\Prob\left[Q_i = b_{n}(1)_{i}\right]  \\
            \cdots \\
            \prod_{i=1}^{n}\Prob\left[Q_i = b_{n}(2^{n}-1)_{i}\right],
        \end{array}\right].
    \end{align*}
    Then $\tilde{Y}$ and $\tilde{X}$ satisfy $\tilde{X} = A(\rho_1, \rho_2, \cdots, \rho_n, \rho)^{-1}\tilde{Y}$ by the definition of $A(\rho_1, \rho_2, \cdots, \rho_n, \rho)$ in Lemma \ref{lemma: QC eq}. We show in the following that if there exists $\epsilon>0$, such that
    \begin{align}
        |Y - \tilde{Y}|\preceq \epsilon \tilde{Y}, \label{eq: condition on Y}
    \end{align}
    then there holds
    \begin{align}
        |X - \tilde{X}|\preceq \epsilon\left[\frac{2-\rho}{\rho}\right]^{n} \tilde{X}. \label{eq: result on X}
    \end{align}
    This statement is stronger than Equation \eqref{eq: Q_bound} in the sense that $X$ and $Y$ are no longer required to be a probability mass function.

    We prove the above statement using mathematical induction with respect to $n$. For $n=1$, $A(\rho_1, \rho)$, $\tilde{Y}$ and $\tilde{X}$ satisfy
    \begin{align*}
        A(\rho_1, \rho)^{-1} =
        \begin{bmatrix}
            \frac{1-\rho}{1-\rho_1} & 0\\
            \frac{\rho-\rho_1}{1-\rho_1} & 1
        \end{bmatrix},\ 
        \tilde{Y} = 
        \left[
        \begin{array}{c}
            1-\rho_1\\
            \rho_1
        \end{array}
        \right],\ 
        \tilde{X} = 
        \left[
        \begin{array}{c}
            1-\rho\\
            \rho
        \end{array}
        \right].
    \end{align*}
    Let $Y = [(1-\rho_1)a, \rho_1 b]^T$. Then for $Y$ satisfying Equation \eqref{eq: condition on Y}, there holds $|a-1|\le \epsilon$, and $|b-1|\le \epsilon$. $X$ is given by $X=A(\rho_1, \rho)^{-1} Y$ as
    \begin{align*}
        X = \left[(1-\rho)a, (\rho-\rho_1) a+ \rho_1 b\right]^T.
    \end{align*}
    There holds the following inequalities that
    \begin{align*}
        |(1-\rho)a-(1-\rho)|\le & \epsilon (1-\rho)< \epsilon \frac{2-\rho}{\rho}(1-\rho),\\
        |(\rho-\rho_1) a+ \rho_1 b-\rho| \le & |(\rho-\rho_1)(a-1)| + |\rho_1 (b-1)|,\\
        \le & \epsilon\frac{2\rho_1-\rho}{\rho}\rho\le \epsilon\frac{2-\rho}{\rho}\rho,
    \end{align*}
    therefore
    \begin{align*}
        |X-\tilde{X}|=
        \left[
        \begin{array}{c}
            |(1-\rho)a - (1-\rho)|\\
            |(\rho-\rho_1) a+ \rho_1 b-\rho|
        \end{array}
        \right]\preceq
        \left[
        \begin{array}{c}
            \epsilon \frac{2-\rho}{\rho}(1-\rho)\\
            \epsilon\frac{2-\rho}{\rho}\rho
        \end{array}
        \right]= \epsilon \frac{2-\rho}{\rho}\tilde{X}.
    \end{align*}
    which finish the proof for the case $n=1$.
    
    Suppose that the statement has been proved for $n=m$. We consider the case $n=m+1$ in the following. For shorthand notations, denote
    \begin{align*}
        A_{m} = & A(\rho_2, \rho_3, \cdots, \rho_{m+1}, \rho),\\
        A_{m+1} = & A(\rho_1, \rho_2, \cdots, \rho_{m+1}, \rho),
    \end{align*}
    By Equation \eqref{eq: A_iter}, there holds
    \begin{align}
        \label{eq: A_iter_simple}
        A_{m+1} = 
        \begin{bmatrix}
            \frac{1-\rho_1}{1-\rho} A_m & 0_{2^m\times 2^m}\\
            \frac{\rho_1-\rho}{1-\rho} A_m & A_m
        \end{bmatrix},\ \mathrm{and}\ A_{m+1}^{-1} = 
        \begin{bmatrix}
            \frac{1-\rho}{1-\rho_1} A_m^{-1} & 0_{2^m\times 2^m}\\
            \frac{\rho-\rho_1}{1-\rho_1} A_m^{-1} & A_m^{-1}
        \end{bmatrix}.
    \end{align}
    To avoid confusion, define $\tilde{Y}(m+1)$ (equal to $\tilde{Y}$ for $n = m+1$) and $\tilde{Y}(m)$ by
    \begin{align*}
        \tilde{Y}(m+1) =
        \left[\begin{array}{c}
            \prod_{i=1}^{m+1}\Prob\left[B_i = b_{m+1}(0)_{i}\right]  \\
            \prod_{i=1}^{m+1}\Prob\left[B_i = b_{m+1}(1)_{i}\right]  \\
            \cdots \\
            \prod_{i=1}^{m+1}\Prob\left[B_i = b_{m+1}(2^{m+1}-1)_{i}\right],
        \end{array}\right], \tilde{Y}(m) =
        \left[\begin{array}{c}
            \prod_{i=2}^{m+1}\Prob\left[B_i = b_{m+1}(0)_{i}\right]  \\
            \prod_{i=2}^{m+1}\Prob\left[B_i = b_{m+1}(1)_{i}\right]  \\
            \cdots \\
            \prod_{i=2}^{m+1}\Prob\left[B_i = b_{m+1}(2^{m+1}-1)_{i}\right],
        \end{array}\right].
    \end{align*}
    Further define $\tilde{X}(m+1) = A_{m+1}^{-1}\tilde{Y}(m+1)$  (equal to $\tilde{X}$ for $n = m+1$) and $\tilde{X}(m) = A_{m}^{-1}\tilde{Y}(m)$. Then by their definitions and Equation \eqref{eq: A_iter_simple}, there holds
    \begin{align}
        \tilde{X}(m+1) =
        \left[
        \begin{array}{c}
            (1-\rho)\tilde{X}(m)\\
            \rho \tilde{X}(m) 
        \end{array}
        \right],&\ 
        \tilde{Y}(m+1) =
        \left[
        \begin{array}{c}
            (1-\rho_1)\tilde{Y}(m)\\
            \rho_1 \tilde{Y}(m) 
        \end{array}
        \right].\label{decom: tilde}
    \end{align}

    For any positive vector $Y\in \R^{2^{m+1}}$, and $X = A_{m+1}^{-1}Y\in \R^{2^{m+1}}$, we divide them into two parts with a equal size in the following way
    \begin{align}
        X =
        \left[
        \begin{array}{c}
            (1-\rho)X_0\\
            \rho X_1
        \end{array}
        \right],&\ 
        Y =
        \left[
        \begin{array}{c}
            (1-\rho_1)Y_0\\
            \rho_1 Y_1
        \end{array}
        \right]. \label{decom: xy}
    \end{align}
    Then by Equation \eqref{eq: A_iter_simple}, $X_0 = A_m^{-1}Y_0$. %Furthermore, define $Y' := (1-\rho_1)Y_0 + \rho_1 Y_1$, and  $X': = A_m^{-1}Y' = (1-\rho) X_0 + \rho X_1$ for future use.

    If $Y$ satisfies Equation \eqref{eq: condition on Y}, by Equation \eqref{decom: tilde} and \eqref{decom: xy}, there holds
    \begin{align*}
        \left|Y_0 - \tilde{Y}(m)\right|\preceq \epsilon \tilde{Y}(m).
    \end{align*}
    By $X_0 = A_m^{-1}Y_0$ and the induction assumption, there holds
    \begin{align*}
        |(1-\rho)X_0-(1-\rho)\tilde{X}(m)|&  \preceq \epsilon\left[\frac{2-\rho}{\rho}\right]^m (1-\rho) \tilde{X}(m)\preceq \epsilon\left[\frac{2-\rho}{\rho}\right]^{m+1} (1-\rho) \tilde{X}(m),
    \end{align*}
    which is exactly the first half of Equation \eqref{eq: result on X} in the case $n = m+1$.
    
    Furthermore, define $Y' := (1-\rho_1)Y_0 + \rho_1 Y_1\in \R^{2^{m}}$, and  $X': = A_m^{-1}Y' = (1-\rho) X_0 + \rho X_1 \in \R^{2^{m}}$. Then there holds
    \begin{align*}
        \left|Y' - \tilde{Y}(m)\right| & = \left|(1-\rho_1)Y_0 + \rho_1 Y_1-\tilde{Y}(m)\right|\le \epsilon \tilde{Y}(m).
    \end{align*}
    Therefore by induction assumption, there holds
    \begin{align*}
        \left|X' - \tilde{X}(m)\right|\le \epsilon\left[\frac{2-\rho}{\rho}\right]^m \tilde{X}(m).
    \end{align*}
   Therefore, there further holds
    \begin{align*}
        |\rho X_1-\rho \tilde{X}(m)| & = |X'-(1-\rho)X_0  - (\rho-1) \tilde{X}(m)-\tilde{X}(m)|,\\
        & \preceq  |(\rho-1) [X_0-\tilde{X}(m)]|+|X'-\tilde{X}(m)|,\\
        &\preceq \epsilon\left[\frac{2-\rho}{\rho}\right]^m (1-\rho) \tilde{X}(m) + \epsilon\left[\frac{2-\rho}{\rho}\right]^m \tilde{X}(m),\\
        & = \epsilon\left[\frac{2-\rho}{\rho}\right]^{m+1} \rho \tilde{X}(m),
        %& \preceq\epsilon \frac{1-\rho}{\rho}(1+?)\rho \tilde{X}^0(m),
        %&\preceq (\rho_1-\rho)\epsilon \tilde{X}^0(m)\prod_{i=2}^{m+1}\frac{2\rho_i-\rho}{\rho} + \rho_1\epsilon \tilde{X}^0(m)\prod_{i=2}^{m+1}\frac{2\rho_i-\rho}{\rho} \epsilon,\\
        %&= \epsilon \rho \tilde{X}^0(m)\prod_{i=1}^{m+1}\frac{2\rho_i-\rho}{\rho},
    \end{align*}
    which is exactly the second half of Equation \eqref{eq: result on X} in the case $n = m+1$.
\end{proof}

With Lemma \ref{lemma: key} bridging the heterogenuous weakly dependent Bernoulli random variables and homogenuous weakly dependent Bernoulli random variables, we can finally prove Lemma \ref{lemma: main_lemma}.

\begin{proof}
    %In this proof, we will consider conditional on fixed $\xi$ and $\zeta\in \Lambda(\nu)$ as stated in Lemma \ref{lemma: main_lemma}. 
    Let $\nu$ satisfy $\nu>\gnu$. Let $\xi$ satisfy $\|\xi\|_\infty<\kappa\sigma\frac{1}{\sqrt{n}}$. For any value of $B_{\{(m^*+1), \cdots, (m)\}}$, let $\rho_i:=\Prob[B_{(i)} = 1|\xi, \zeta\in\Lambda(\nu), B_{\{(m^*+1), \cdots, (m)\}}]$ for $i\le m^*$, and define $\rho=\min_{i\le m^*}\{\rho_i\}$. %We will assume $\max_i \rho_i<1$ to exclude some trivial cases. 
    Suppose that there exists $1\le J\le m^*$, $B_{(j)}\equiv1$ conditional on $\xi$, $\zeta\in\Lambda(\nu)$ and $B_{\{(m^*+1), \cdots, (m)\}}$, then there holds for $J\le m^*$ that
    \begin{align*}
        \frac{1+J}{1+B_{(1)}+B_{(2)}+\cdots+B_{(J)}} \le \frac{J}{1+B_{(1)}+B_{(2)}+\cdots+B_{(J)}-B_{(j)}},
    \end{align*}
    due to the inequality $\frac{a+c}{b+c}\le \frac{a}{b}$ whenever $0<b\le a$ and $c>0$. Such a result shows that dropping $B_{(j)}$ gives an upper bound of the target. Thus, for simplicity yet without loss of generality, we assume that
    \begin{align*}
        1> \max_i\rho_i\ge\min_i\rho_i\ge \rho>0.
    \end{align*}
    By Lemma \ref{lemma: weak dependent}, given $\xi$, there holds for $S = \{i: \widehat{\sign}_i\neq\gamma^*_i\}$, any $k\in \{0, 1, \cdots,2^{|S|}-1\}$ and any value of $B_{\{(m^*+1), \cdots, (m)\}}$
    \begin{align*}
        e^{-48\theta ^2\co2 \frac{m(|S|+1)}{\nu}}\le\frac{\Prob[B_S = b_{|S|}(k)|\xi, \zeta\in\Lambda(\nu), B_{S^c}]}{\prod_{i\in S}\Prob[B_i = b_{|S|}(k)_i|\xi, \zeta\in\Lambda(\nu), B_{S^c}]}\le e^{48\theta ^2\co2 \frac{m(|S|+1)}{\nu}},
    \end{align*}
    where $B_i = 1\{W_i<0\}$ is defined in Lemma \ref{lemma: weak dependent}. Therefore, there further holds
    \begin{align*}
        &\left|\Prob[B_S = b_m(k)_S|\xi, \zeta\in\Lambda(\nu), B_{S^c}]-\prod_{i\in S}\Prob[B_i = b_m(k)_i|\xi, \zeta\in\Lambda(\nu), B_{S^c}]\right|\\
        \le& (e^{48\theta ^2\co2 \frac{m(m+1)}{\nu}}-1)\prod_{i\in S}\Prob[B_i = b_m(k)_i|\xi, \zeta\in\Lambda(\nu), B_{S^c}],
    \end{align*}
    due to the fact that for $x>0$, $1-\exp(-x)\le \exp(x)-1$.
    By Lemma \ref{lemma: key}, there exists Bernoulli random variables $Q_1, Q_2,\cdots, Q_{m^*}$, such that the following properties hold.
    \begin{enumerate}
        \item $Q_i\sim \mathrm{Bernoulli}(\rho)$ for all $i$.
        \item For any $k\in \{0, 1, \cdots, 2^{m^*}-1\}$, there holds
        \begin{align*}
            \left|\Prob[Q = b_{m^*}(k)]-\prod_{i}\Prob[Q_i = b_{m^*}(k)_i]\right|\le & \left(e^{48\theta ^2\co2 \frac{m(m+1)}{\nu}}-1\right)\left[\frac{2-\rho}{\rho}\right]^{m} \prod_{i}\Prob[Q_i = b_{m^*}(k)_i],\\
            \le & \left(e^{48\theta ^2\co2 \frac{m(m+1)}{\nu}}-1\right)e^{\frac{(2-2\rho)m}{\rho}} \prod_{i}\Prob[Q_i = b_{m^*}(k)_i],\\
            := & \epsilon(m)\prod_{i}\Prob[Q_i = b_{m^*}(k)_i],
        \end{align*}
        where we define $\epsilon(m) = \left(e^{48\theta ^2\co2 \frac{m(m+1)}{\nu}}-1\right)e^{\frac{(2-2\rho)m}{\rho}}$ for simplicity.
        \item \label{property: C} Define $\A\subseteq\{1, 2, \cdots, m^*\}$ as a random set such that $\A$ is independent from $\{Q_i\}_{i=1}^{m^*}$, the events $\{i\in \A\}$ are independent from each other, and
        \begin{align*}
            \Prob\{i\in A\} = \frac{1-\rho_i}{1-\rho}.
        \end{align*}
        Further define $C_i=Q_i\cdot 1\{i\in A\}+1\{i\notin A\}$. There holds
        \begin{align*}
            \{C_1, C_2,\cdots, C_{m^*}\} \overset{d}{=}\{B_{(1)}, B_{(2)}, \cdots, B_{(m^*)}\}|\xi, \zeta\in\Lambda(\nu), B_{\{(m^*+1), \cdots, (m)\}}.
        \end{align*} 
    \end{enumerate}
    Applying the same argument as in \cite{barber2019knockoff}, conditional on $\xi$ and $B_{\{(m^*+1), \cdots, (m)\}}$, $J$ defined above Lemma \ref{lemma: main_lemma} is a stopping time in reverse time with respect to the filtration $\{\F_j\}_{j=1}^{m^*}$ defined as
    \begin{align*}
        \F_j = \sigma\left(\left\{\sum_{i=1}^j \tilde{B}_{(j)}, \tilde{B}_{(j+1)}, \cdots, \tilde{B}_{(m^*)}\right\}\right).
    \end{align*}
    By Property \ref{property: C}, $J$ can be also viewed as a stopping time in reverse time with respect to the filtration $\{\G_j\}_{j=1}^{m^*}$ defined as
    \begin{align*}
        \G_j = \sigma\left(\left\{\sum_{i=1}^j C_{1}, C_{j+1}, \cdots, C_{m^*}\right\}\right),
    \end{align*} 
    and satisfies
    \begin{align}
        \Expect\left[\left.\frac{1+J}{1+B_{(1)}+B_{(2)}+\cdots+B_{(J)}}\right|\xi, \zeta\in\Lambda(\nu), B_{\{(m^*+1), \cdots, (m)\}}\right]=\Expect\left[\frac{1+J}{1+C_{1}+C_{2}+\cdots+C_{J}}\right]. \label{eq: btoc}
    \end{align}
    By definition $C_{i}=Q_{i}\cdot 1\{i\in \A\}+1\{i\notin \A\}$ for all $i$, therefore
    \begin{align}
        \frac{1+J}{1+C_{1}+C_{2}+\cdots+C_{J}} & = \frac{1+|\{i\le J: i\in \A\}|+|\{i\le J: i\notin \A\}|}{1+\sum_{i\le J, i\in\A}Q_{i}+|\{i\le J:i\notin \A\}|},\nonumber\\
        & \le \frac{1+|\{i\le J: i\in \A\}|}{1+\sum_{i\le J, i\in\A}Q_{i}},\label{eq: transfer}
    \end{align}
    where the last step is due to the inequality $\frac{a+c}{b+c}\le \frac{a}{b}$ whenever $0<b\le a$ and $c>0$.
    
    Let $\tilde{Q}_{i} = Q_{i}\cdot 1\{i\in \A\}$, and define
    \begin{align*}
        \F_j' = \sigma\left(\left\{\sum_{i=1}^j \tilde{Q}_{j}, \tilde{Q}_{j+1}, \cdots, \tilde{Q}_{m^*}, \A\right\}\right),
    \end{align*}
    then clearly $\{F_j'\}_{j=1}^{m^*}$ is a filtration. By definition of $\{C_{i}\}_{i=1}^{m^*}$, there holds
    \begin{align*}
        \G_j = \sigma(\{C_1+C_2+\cdots+C_j, C_{j+1}, \cdots, C_{m^*}\})\subseteq \F_j'.
    \end{align*}
    Therefore, $J$ is also a stopping time in reverse time with respect to $\{F_j'\}_{j=1}^{m^*}$. Below, we are going to work on the upper bound of
    \begin{align*}
        \E\left[\frac{1+|\{i\le J: i\in \A\}|}{1+\sum_{i\le J, i\in\A}Q_i}\right].
    \end{align*}
    
    Let $j_\A := \max_i\{i\le j, i\in \A\}$, $k(j, \A):=\sum_{i\le j, i\in\A}Q_i$, and $m(j, \A):=|\{i\le j, i\in\A\}|$. Since $\{Q_i\}_{i=1}^n$ are independent from $\A$, for any $\A$ and $\tilde{m}\ge \tilde{k}\ge 1$, there holds
    \begin{align}
        \Prob\left[Q_{j_\A}=1\left|\mathcal{S}(\A, \tilde{k}, \tilde{m})\right]\right. & = \frac{\Prob[Q_{j_\A}=1, k(j, \A) = \tilde{k}, m(j, \A)=\tilde{m}, Q_{j+1}, \cdots, Q_m,\A]}{\Prob[k(j, \A) = \tilde{k}, m(j, \A)=\tilde{m}, Q_{j+1}, \cdots, Q_m,\A]},\nonumber\\
        & \le \frac{1+\epsilon(m)}{1-\epsilon(m)}\frac{\rho C_{\tilde{m}-1}^{\tilde{k}-1}\rho^{\tilde{k}-1}(1-\rho)^{\tilde{m}-\tilde{k}}}{C_{\tilde{m}}^{\tilde{k}}\rho^{\tilde{k}}(1-\rho)^{\tilde{m}-\tilde{k}}}, \label{eq: epsilon_ineq}\\
        & = \frac{1+\epsilon(m)}{1-\epsilon(m)}\frac{\tilde{k}}{\tilde{m}},
    \end{align}
    when $\epsilon(m)<1$, where $\mathcal{S}(\A, \tilde{k}, \tilde{m})\in \F_j'$ is a generator of the $\sigma$-field $\F_j'$,  which consists of $\A$, $m(j, \A) = \tilde{m},\ k(j, \A) = \tilde{k}$, and $Q_{j+1}, \cdots, Q_{m^*}$. For shorthand notations, let $\tilde{p} = \Prob\left[Q_{j_\A}=1\left|\mathcal{S}(\A, \tilde{k}, \tilde{m})\right]\right.$, for $k(j, \A) = \tilde{k}\ge 1$, there further holds
    \begin{align*}
        \E\left[\left.\frac{m(j, \A)}{1+k(j, \A)-Q_{j_\A}}\right|\mathcal{S}(\A, \tilde{k}, \tilde{m})\right] & = \frac{\tilde{m}}{\tilde{k}}\tilde{p}+\frac{\tilde{m}}{1+\tilde{k}}(1-\tilde{p})=\tilde{m}\frac{\tilde{p}+\tilde{k}}{\tilde{k}(\tilde{k}+1)},\\
        &\le \frac{\tilde{m}+1}{\tilde{k}+1}+\left(\frac{1+\epsilon(m)}{1-\epsilon(m)}-1\right)\frac{1}{(\tilde{k}+1)},\\
        &\le \frac{\tilde{m}+1}{\tilde{k}+1}\left[1+\frac{1}{\tilde{m}+1}\frac{2\epsilon(m)}{1-\epsilon(m)}\right].
        %&\le \frac{1+m(j, \A)}{1+k(j, \A)}\left[1+\frac{1}{j+1}\frac{2\epsilon(j)}{1-\epsilon(j)}\right].
    \end{align*}
    Therefore, there holds for $k(j, \A) = \tilde{k}\ge 1$, that
    \begin{align}
        \E\left[\left.\frac{m(j, \A)}{1+k(j, \A)-Q_{j_\A}} - \frac{1+m(j, \A)}{1+k(j, \A)}\left[1+\frac{1}{m(j, \A)+1}\frac{2\epsilon(m)}{1-\epsilon(m)}\right]\right|\mathcal{S}(\A, \tilde{k}, \tilde{m})\right]\le 0. \label{eq: generator}
    \end{align}
    For the case $k(j, \A) = \tilde{k} = 0$, the above inequality is trivial. 
    
    As Equation \eqref{eq: generator} holds for any $\mathcal{S}(\A, \tilde{k}, \tilde{m})$, which generates the $\sigma-$field $\F_j'$, there holds
    \begin{align}
        \E\left[\left.\frac{m(j, \A)}{1+k(j, \A)-Q_{j_\A}} - \frac{1+m(j, \A)}{1+k(j, \A)}\left[1+\frac{1}{m(j, \A)+1}\frac{2\epsilon(m)}{1-\epsilon(m)}\right]\right|\F_j'\right]\le 0.\label{eq: mar1}
    \end{align}
    If $j_{\A} = j$, by Equation \eqref{eq: mar1}, there holds
    \begin{align}
        \E\left[\left.\frac{1+m(j-1, \A)}{1+k(j-1, \A)} - \frac{1+m(j, \A)}{1+k(j, \A)}\left[1+\frac{1}{m(j, \A)+1}\frac{2\epsilon(m)}{1-\epsilon(m)}\right]\right|\F_j'\right]\le 0.\label{eq: mar2}
    \end{align}
    Meanwhile, if $j_{\A}\neq j$, there trivially holds
    \begin{align}
        \E\left[\left.\frac{1+m(j-1, \A)}{1+k(j-1, \A)} - \frac{1+m(j, \A)}{1+k(j, \A)}\right|\F_j'\right]= 0.\label{eq: mar3}
    \end{align}
    Combining Equation \eqref{eq: mar3} with Equation \eqref{eq: mar2}, there holds that
    \begin{align*}
        '\left[\frac{1+m(j, \A)}{1+k(j, \A)}\right]\prod_{i=1}^{m(j, \A)}\left[1+\frac{1}{i+1}\frac{2\epsilon(m)}{1-\epsilon(m)}\right],
    \end{align*}
    is a supermartingale with respect to $\{\F_j'\}_{j=1}^m$. For any $\A$, there holds
    \begin{align*}
        \E\left[\frac{1+m(n, \A)}{1+k(n, \A)}\right]&= \E\left[\frac{1+|\A|}{1+k(n, \A)}\right]  = \sum_{i=0}^{|\A|}\frac{|\A|+1}{i+1}\Prob[k(n, \A)=i],\\
        &\le \sum_{i=0}^{|\A|}\frac{|\A|+1}{i+1}\left(1+\epsilon(m)\right)\rho^i(1-\rho)^{|\A|-i}C_{|\A|}^i,\\
        & = \rho^{-1}\left(1+\epsilon(m)\right)\sum_{i=0}^{|\A|}\rho^{i+1}(1-\rho)^{|\A|-i}C_{|\A|+1}^{i+1}\le \rho^{-1}\left(1+\epsilon(m)\right).
    \end{align*}
    Therefore, by the optional stopping theorem, there holds
    \begin{align}
        \label{eq: con a ineq}
        \E\left[\frac{1+|\{i\le J: i\in \A\}|}{1+\sum_{i\le J, i\in\A}Q_i}\right]%& \le \rho^{-1}(1+\epsilon(m))\prod_{k=J+1}^{m(n, \A)}\left[1+\frac{1}{k+1}\frac{2\epsilon(m)}{1-\epsilon(m)}\right],\\
        & \le \rho^{-1}(1+\epsilon(m))\prod_{k=1}^{m}\left[1+\frac{1}{k+1}\frac{2\epsilon(m)}{1-\epsilon(m)}\right],\\
        & = \rho^{-1}(1+\epsilon(m))\exp\left[\sum_{k=1}^m \ln\left[1+\frac{1}{k+1}\frac{2\epsilon(m)}{1-\epsilon(m)}\right]\right],\\
        &\le \rho^{-1}(1+\epsilon(m))\exp\left[\ln(2m)\frac{2\epsilon(m)}{1-\epsilon(m)}\right],
    \end{align}
    where the last step use the fact that $\ln(1+x)\le x$ for $x>0$, and $\sum_{i=1}^m \frac{1}{i+1}<\ln (2m)$. Combine the above result with Equation \eqref{eq: btoc} and \eqref{eq: transfer}, when $\epsilon(m)<1$, there holds
    \begin{align*}
        \Expect\left[\left.\frac{1+J}{1+B_{(1)}+B_{(2)}+\cdots+B_{(J)}}\right|\xi, \zeta\in\Lambda(\nu), B_{\{(m^*+1), \cdots, (m)\}}\right]& \le \rho^{-1}(1+\epsilon(m))\exp\left[\ln(2m)\frac{2\epsilon(m)}{1-\epsilon(m)}\right].
    \end{align*}
    
    By Proposition \ref{prop: conditional upper bound}, $\rho = \max_i\{\Prob[B_{(i)} = 1|\xi, \zeta\in\Lambda(\nu), B_{\{(m^*+1), \cdots, (m)\}}]\}$ satisfy
    \begin{align*}
        1-\rho<\frac{\fnu}{\sqrt{\nu}}\Rightarrow\frac{1-\rho}{\rho}<\frac{\fnu}{\sqrt{\nu}-\fnu}.
    \end{align*}
    Therefore
    \begin{align*}
        \epsilon(m) =  \left(e^{48\theta ^2\co2 \frac{m(m+1)}{\nu}}-1\right)e^{\frac{(2-2\rho)m}{\rho}}<\left(e^{48\theta ^2\co2 \frac{m(m+1)}{\nu}}-1\right)e^{\frac{2\fnu m}{\sqrt{\nu}-\fnu}}:=\epsilon^*(m).
    \end{align*}
    When $\tau_\nu<1$, there further holds $\epsilon(m)<\epsilon^*(m)<\frac{1}{8\ln(2m)}<\frac{1}{4}$, and
    \begin{align*}
        \ln(2m)\frac{2\epsilon(m)}{1-\epsilon(m)}<\frac{8}{3}\ln(2m)\epsilon^*(m)<\frac{1}{3}.
    \end{align*}
    Since $e^x<1 + (3e^{\frac{1}{3}}-3)x$ for $x<\frac{1}{3}$, there further holds
    \begin{align*}
        \exp\left[\ln(2m)\frac{2\epsilon(m)}{1-\epsilon(m)}\right]\le 1+[3e^{\frac{1}{3}}-3]\frac{8}{3}\ln(2m) \epsilon^*(m)<1+4\ln(2m)\epsilon^*(m)=1+\frac{1}{2}\tau_\nu
    \end{align*}
    Therefore when $\tau_\nu<1$, there further holds
    \begin{align*}
        \Expect\left[\left.\frac{1+J}{1+B_{(1)}+B_{(2)}+\cdots+B_{(J)}}\right|\xi, \zeta\in\Lambda(\nu), B_{\{(m^*+1), \cdots, (m)\}}\right]& \le \frac{\sqrt{\nu}}{\sqrt{\nu}-\fnu}\left(1+\frac{\tau_\nu}{8\ln(2m)}\right)\left(1+\frac{\tau_\nu}{2}\right),\\
        &\le \frac{\sqrt{\nu}}{\sqrt{\nu}-\fnu}(1+\tau_\nu).
    \end{align*}
    Taking expectation over $B_{\{(m^*+1), \cdots, (m)\}}$ ends the proof.
\end{proof}

\section{Proof of Theorem \ref{thm:pathconsistency}}
\label{sec:proof_path_consistency}

In this section, we prove Theorem \ref{thm:pathconsistency} by constructing the Primal-Dual Witness (PDW) of Split LASSO similarly to that of the traditional LASSO problem \citep{Wainwright09} and Split Linearized Bregman Iterations \citep{SplitLBI,huang2020boosting}. We first introduce the definition of the successful Primal-Dual Witness, which gives an unique solution of Split LASSO; then we introduce the incoherence condition for Split LASSO and establish the no-false-positive and sign consistency of Split LASSO regularization path, i.e. Theorem \ref{thm:pathconsistency}.

The set of witness $(\hat{\beta}, \hat{\gamma}, \hat{\rho})\in \R^p \times \R^m \times \R^m$ is constructed in the following way:

\begin{enumerate}
    \item First, we set $\hat\gamma_{S_0} =0$, and obtain $(\hat\beta , \hat\gamma_{\S_1} )\in \R^p \times \R^{|\S_1|}$ by solving
    \begin{align}
        (\hat\beta , \hat\gamma_{\S_1} ) = \argmin_{(\beta, \gamma_{\S_1})}\left\{\frac{1}{2n}\|y - X\beta\|_2^2 + \frac{1}{2\nu}\|D\beta - \gamma\|_2^2 + \lambda \|\gamma_{\S_1}\|_1\right\}.\label{subproblem}
    \end{align}
    \item Second, we choose $\hat\rho_{\S_1} =\partial \|\hat\gamma_{S_1} \|_1$ as the subgradient of $\|\hat\gamma_{S_1} \|_1$.
    \item Third, $\hat\rho_{\S_0}  \in \R^{\S_0}$ is solved from the KKT condition of Split LASSO:
    \begin{subequations}\label{eq:slasso-kkt}
    \begin{align}
      0 &= - (\Sigma_X + L_D) \beta(\lambda) + \frac{D^T}{\nu} \gamma(\lambda) + \left\{\Sigma_X \beta^* + \frac{X^T}{n} \varepsilon \right\}, \label{eq:slasso-kkta}  \\
      \lambda\rho(\lambda)&= \frac{D\beta(\lambda)}{\nu} -  \frac{\gamma(\lambda)}{\nu}, \label{eq:slasso-kktb}
    \end{align}
    \end{subequations}
    where $\rho(\lambda)\in \partial\|\gamma(\lambda)\|_1$. Then, we check whether the dual feasibility condition $|\hat\rho_j |< 1$ for all $j\in \S_0$ is satisfied. 
    %\item Fourth, we check whether the condition $\hat\rho_{\S_1}  = \sign(\beta^*_{\S_1})$ is satisfied.
\end{enumerate}
If the dual feasibility condition is satisfied, we say the PDW succeed.

\subsection{Uniqueness of the Solution upon Successful PDW}

 In this section, we introduce the following lemma that shows the uniqueness of the solution upon the successful construction of the PDW of Split LASSO.

\begin{lemma}
    When the PDW succeed, if the subproblem \eqref{subproblem} is strongly convex, the solution $(\hat\beta , \hat\gamma )$ is the unique optimal solution for Split LASSO.
\end{lemma}

\begin{proof}
    When the PDW succeed, there holds $\|\hat\rho_{\S_0}\|_\infty< 1$. Therefore $(\hat\beta , \hat\gamma )$ is a set of optimal solution with respect to $\lambda$, where $\hat\rho \in\left\|\hat\gamma \right\|_1$. Let $(\tilde{\beta}, \tilde{\gamma})$ be any other optimal solution of Split LASSO. Define
    \begin{align*}
        F(\beta, \gamma) := \frac{1}{2n}\|y - X\beta\|_2^2 + \frac{1}{2\nu}\|D\beta - \gamma\|_2^2.
    \end{align*}
    Then there holds
    \begin{align}
        F(\hat\beta , \hat\gamma ) + \lambda\langle\hat\rho , \hat\gamma \rangle = F(\tilde\beta, \tilde\gamma)+\lambda\|\tilde\gamma\|_1,\nonumber
    \end{align}
    which is 
    \begin{align}
        F(\hat\beta , \hat\gamma ) - \lambda\langle\hat\rho , \tilde\gamma - \hat\gamma \rangle - F(\tilde\beta, \tilde\gamma) = \lambda(\|\tilde\gamma\|_1-\langle\hat\rho , \tilde\gamma \rangle).\nonumber
    \end{align}
    By Equation \eqref{eq:slasso-kkt}, there holds $\frac{\partial F(\hat\beta , \hat\gamma )}{\partial \hat\beta } = 0$, and $\frac{\partial F(\hat\beta , \hat\gamma )}{\partial \hat\gamma } = -\lambda \hat\rho $. Therefore
    \begin{align}
        F(\hat\beta , \hat\gamma ) + \langle\frac{\partial F(\hat\beta , \hat\gamma )}{\partial \hat\gamma }, \tilde\gamma - \hat\gamma \rangle + \langle\frac{\partial F(\hat\beta , \hat\gamma )}{\partial \hat\beta }, \tilde\beta-\hat\beta \rangle - F(\tilde\beta, \tilde\gamma) = \lambda(\|\tilde\gamma\|_1-\langle\hat\rho , \tilde\gamma \rangle).\label{lhs<0}
    \end{align}
    Since $F$ is convex, left hand side of Equation \eqref{lhs<0} is non-positive, therefore
    \begin{align}
        \|\tilde\gamma\|_1 \le \langle\hat\rho , \tilde\gamma \rangle.\nonumber
    \end{align}
    Since $|\hat\rho_{S_0} |<1$, there holds $\tilde\gamma_{S_0}=0$. Therefore $(\tilde\beta, \tilde\gamma)$ is also an optimal solution for the subproblem \eqref{subproblem}. Therefore, if the subproblem \eqref{subproblem} is strongly convex, $(\hat\beta , \hat\gamma )$ is the only solution for Split LASSO with respect to $\lambda$.
\end{proof}

\subsection{Incoherence Condition and Path Consistency}

In this section, we first present how the incoherence conditions are deducted. Then we  give the proof of Theorem \ref{thm:pathconsistency}.

 From $D[\Sigma_X+L_D]^{-1}\times$ Equation \eqref{eq:slasso-kkta} + $\nu\times$ Equation \eqref{eq:slasso-kktb}, and the fact that $\gamma^* = D\beta^*$, there holds
\begin{align*}
    \lambda\nu \hat\rho = & -\hat\gamma+\frac{D[\Sigma_X+L_D]^{-1}D^T \hat\gamma}{\nu} + D[\Sigma_X+L_D]^{-1}(\Sigma_X+L_D-L_D)\beta^*+\ldots\nonumber\\
    &\ldots+D[\Sigma_X+L_D]^{-1}\frac{X^T}{n}\varepsilon,\\
    = & -H_\nu (\hat\gamma - \gamma^*)+\omega,
\end{align*}
where $\omega := D[\Sigma_X+L_D]^{-1}\frac{X^T}{n}\varepsilon$. Moreover, there holds
\begin{align}
    \lambda \nu
    \begin{bmatrix}
        \hat\rho_{\S_1}\\
        \hat\rho_{\S_0}
    \end{bmatrix}
    =-
    \begin{bmatrix}
        H_\nu^{11} & H_\nu^{10}\\
        H_\nu^{01} & H_\nu^{00}
    \end{bmatrix}
    \begin{bmatrix}
        \hat\gamma_{\S_1} - \gamma^*_{\S_1}\\
        0_{\S_0}
    \end{bmatrix}
    +
    \begin{bmatrix}
        \omega_{\S_1}\\
        \omega_{\S_0}
    \end{bmatrix},\nonumber
\end{align}
which means
\begin{subequations}
    \begin{align}
        \lambda\nu \hat\rho_{\S_1} = & - H_\nu^{11}(\hat\gamma_{\S_1} - \gamma^*_{\S_1}) + \omega_{\S_1}\label{eq1: irr},\\
        \lambda\nu \hat\rho_{\S_0} = & -  H_\nu^{01}(\hat\gamma_{\S_1} - \gamma^*_{\S_1}) + \omega_{\S_0}\label{eq2: irr}.
    \end{align}
\end{subequations}
Since $H_\nu^{11}$ is invertible, solve $\hat\gamma_{\S_1} - \gamma^*_{\S_1}$ from Equation \eqref{eq1: irr}, there holds
\begin{align}
    \hat\gamma_{\S_1} - \gamma^*_{\S_1} = -\lambda\nu [H_\nu^{11}]^{-1} \hat\rho_{\S_1} + [H_\nu^{11}]^{-1}\omega_{\S_1}.\label{solve_gamma}
\end{align}
Plug Equation \eqref{solve_gamma} into Equation \eqref{eq2: irr}, there holds
\begin{align}
    \hat\rho_{\S_0} =  H_\nu^{01} [H_\nu^{11}]^{-1} \hat\rho_{\S_1} +\frac{1}{\lambda\nu}\{ \omega_{\S_0} - H_\nu^{01}[H_\nu^{11}]^{-1}\omega_{\S_1}\}.\label{estimation_rho}
\end{align}
Hence $H_\nu^{01} [H_\nu^{11}]^{-1}$ plays an important role in Equation \eqref{estimation_rho}. Therefore, the incoherence condition \eqref{incoherence} concerns the infinity norm of $H_\nu^{01} [H_\nu^{11}]^{-1}$. Now we are ready to prove Theorem \ref{thm:pathconsistency}.

\begin{proof}[Proof of Theorem \ref{thm:pathconsistency}]
    ~\\
    From Equation \eqref{estimation_rho} and the incoherence condition \eqref{incoherence}, there holds for $\lambda = \lambda_n$
    \begin{align}
        \|\hat\rho_{\S_0}\|_\infty \le & (1-\chi_\nu) + \frac{1}{\lambda_n\nu}[\|\omega_{\S_0}\|_\infty + (1-\chi_\nu)\|\omega_{\S_1}\|_\infty],\nonumber\\
        \le & (1-\chi_\nu) + \frac{2}{\lambda_n}\left\|\frac{\omega}{\nu}\right\|_\infty.\nonumber
    \end{align}
    By definition, $\frac{\omega_i}{\nu}$ is a Gaussian random variable with variance 
    \begin{align*}
        &\frac{1}{\nu}\frac{\sigma^2}{n}D_{i}[\Sigma_X+L_D]^{-1}\Sigma_X\left[\nu\Sigma_X+D^TD\right]^{-1}D_{i}^T\le \frac{1}{\nu} \frac{c_1\sigma^2}{n}D_{i}[\Sigma_X+L_D]^{-1}D_{i}^T\le \frac{c_1\sigma^2}{n},
    \end{align*}
    for a constant $c_1>0$ related with $X$ and $D$. Therefore
    \begin{align}
        \Prob\left(\left\|\frac{\omega}{\nu}\right\|_\infty\ge \frac{\lambda_n\chi_\nu}{2}\right)\le 2m \exp\left(-\frac{\lambda_n^2 \chi_\nu^2 n}{8c_1\sigma^2}\right).\label{concentration1}
    \end{align}
    For $\lambda_n> \frac{4c_1}{\chi_\nu}\sqrt{\frac{\sigma^2\ln m}{n}}$, there further holds
    \begin{align}
        \Prob\left(\|\hat\rho_{\S_0}\|_\infty>1-\frac{\chi_\nu}{2}\right) \le 2e^{-c_2 n\lambda_n^2},\nonumber
    \end{align}
    where $c_2 = \frac{\chi_\nu^2}{16c_1\sigma^2}$. This is the upper bound on the probability that the PWD fails.
    
    From Equation \eqref{solve_gamma}, there holds
    \begin{align}
        \|\hat\gamma_{\S_1} - \gamma^*_{\S_1}\|_\infty \le \lambda_n\nu\| [H_\nu^{11}]^{-1}\|_\infty + \|[H_\nu^{11}]^{-1}\omega_{\S_1}\|_\infty.\nonumber
    \end{align}
    The first part is deterministic, therefore we can only estimate the second part. Similar to Equation \eqref{concentration1}, there holds
    \begin{align}
        \Prob(\left\|[H_\nu^{11}]^{-1}\omega_{\S_1}\right\|_\infty>\nu t) \le 2|S_1|\exp\left(-\frac{t^2C_\mathrm{min}^2n}{2c_1\sigma^2}\right).\nonumber
    \end{align}
    %which holds for any subgaussian noise of scale $\sigma^2$.
    Take $t = \sigma\lambda_n/2C_\mathrm{min}$, by $\ln m> \ln \S_1$, there holds
    \begin{align}
        \Prob\left(\left\|[H_\nu^{11}]^{-1}\omega_{\S_1}\right\|_\infty>\nu \frac{\sigma\lambda_n}{2C_\mathrm{min}}\right) \le 2\exp(-c_3 n\lambda_n^2).\nonumber
    \end{align}
    where $c_3 = \frac{1}{4c_1}$. Therefore, there further holds
    \begin{align}
        \|\hat\gamma_{\S_1}-\gamma^*_{\S_1}\|_\infty\le \lambda_n \nu \left[\frac{\sigma}{2C_\mathrm{min}}+\| [H_\nu^{11}]^{-1}\|_\infty\right],\nonumber
    \end{align}
    with probability greater than $1-2\exp(-c_3n \lambda_n^2)$. Therefore for $c = \max\{4c_1, c_2, c_3\}$, and $\lambda_n> \frac{c}{\chi_\nu}\sqrt{\frac{\sigma^2\ln m}{n}}$, with probability greater than $1-4\exp(-cn \lambda_n^2)$, Split LASSO has no false discovery at $\lambda_n$. In addition, if $\min_{i\in \S_1}{\gamma^*_i}> \lambda_n \nu \left[\frac{\sigma}{2C_\mathrm{min}}+\| [H_\nu^{11}]^{-1}\|_\infty\right]$, with probability greater than $1-4\exp(-cn \lambda_n^2)$, 
    $\hat\gamma$ recovers the sign of $\gamma^*$.
\end{proof}

\end{appendix}

\end{document}